\documentclass[12pt,a4]{article}
\pdfoutput=1
\usepackage{jheppub}
\usepackage{latexsym,diagbox,stackrel}
\usepackage{mathrsfs}
\usepackage{arydshln,leftidx,mathtools}
\usepackage{bbold}
\usepackage{pdflscape}
\usepackage{multirow}
\usepackage{enumerate}
\usepackage{cleveref}
\usepackage{tikz}
\usepackage{graphicx}
\usepackage{slashed}
\usepackage{datetime}
\usetikzlibrary{shapes.geometric, arrows}
\definecolor{light-gray}{gray}{0.95}
\definecolor{light-grayII}{gray}{0.85}

\usepackage
[
color=gray]
{todonotes} 

\def\d{\text{d}}
\def\bd {\boldsymbol}

\newcommand{\C}{\mathbb{C}}

\newcommand{\CP}{\mathbb{CP}}

\newcommand{\R}{\mathbb{R}}

\newcommand{\cI}{\mathcal{I}}

\newcommand{\Z}{\mathbb{Z}}

\newcommand{\dbar}{\bar\partial}
\newcommand{\e}{\mathrm{e}}

\newcommand{\cM}{\mathcal{M}}

\newcommand{\cO}{\mathcal{O}}

\newcommand{\rd}{\, \mathrm{d}}
\newcommand{\Pf}{\, \mathrm{Pf}}
\newcommand{\pf}{\text{Pf} \,}

\newcommand{\be}{\begin{equation}\label}
\newcommand{\ee}{\end{equation}}
\newcommand{\bea}{\begin{eqnarray}\label}
\newcommand{\eea}{\end{eqnarray}}

\newtheorem{thm}{Theorem}

\newtheorem{lemma}{Lemma}[section]

\begin{document}

\preprint{CERN-PH-TH-2015-267, DAMTP-2015-76 \\}

\title{One-loop amplitudes on the Riemann sphere}

\author{Yvonne Geyer$^1$, Lionel Mason$^1$, Ricardo Monteiro$^2$, Piotr Tourkine$^3$ \vspace{.4cm}
\\ \small{$^1$Mathematical Institute, University of Oxford, Woodstock Road, Oxford OX2 6GG, UK
\\ \vspace{.1cm}
$^2$CERN, Theory Group, Geneva, Switzerland
\\ \vspace{.1cm}
$^3$DAMTP, University of Cambridge, Wilberforce Road, Cambridge CB3 0WA, UK}}

\abstract{
The scattering equations provide a powerful framework for the study of
scattering amplitudes in a variety of 
theories. Their derivation from ambitwistor string theory led to
proposals for formulae at one loop on a torus for 10 dimensional
supergravity, and we recently showed how these can be reduced to the
Riemann sphere and checked in simple cases.  We also proposed
analogous formulae for other theories including maximal super-Yang-Mills theory and supergravity in
other dimensions at one loop. We give further details of these results and extend them in two directions. Firstly,
we propose new formulae for the one-loop integrands of Yang-Mills
theory and gravity in the absence of supersymmetry. These follow from
the identification of the states running in the loop as expressed in
the ambitwistor-string correlator. Secondly, we   give a systematic
proof of the non-supersymmetric formulae using the worldsheet  factorisation properties of the nodal
Riemann sphere underlying the scattering equations at one loop. Our formulae have the same decomposition under the recently introduced Q-cuts as one-loop integrands and hence give the correct amplitudes.
}

\maketitle

\section{Introduction} 
\paragraph{Background.}
From its inception, string theory has provided remarkable conceptual
simplification to the computation of scattering amplitudes, see for
example~\cite{Bern:1991aq} or more recently \cite{Mafra:2014gja} and
references therein.  
%
Twistor string theories
\cite{Witten:2003nn,Berkovits:2004hg,Skinner:2013xp}, provide models
that give rise to formulae
\cite{Roiban:2004yf,Cachazo:2012pz,Cachazo:2012kg,Cachazo:2012da} that
are worldsheet reformulations of conventional field theory
amplitudes. Not only do they benefit from the conceptual and technical
simplifications of string based ideas, but  they also take advantage of
the geometry of twistor space to exhibit properties of amplitudes that are not apparent from a space-time perspective.  These are now understood as examples of
ambitwistor-string theories \cite{Mason:2013sva,Casali:2015vta} that
underlie the formulae by Cachazo, He and Yuan (CHY) for massless
scattering amplitudes in a wide variety of theories
\cite{Cachazo:2013iea,Cachazo:2013iaa,Cachazo:2013hca,Cachazo:2013gna,Cachazo:2014nsa,Cachazo:2014xea}.
These formulae are often expressed in terms of moduli integrals, but
are essentially algebraic as there are as many delta functions as
integration variables, and the integrals localise to a sum of residues
supported at solutions to the \emph{scattering equations}.

Given $n$ null momenta $k_i$, the scattering equations determine $n$
points $\sigma_i$ on a Riemann sphere up to M\"obius
transformations. Introduced first by Fairlie and Roberts
\cite{Fairlie:1972zz,robertsphd,Fairlie:2008dg} to construct classical minimal
surfaces as string solutions, they also determine the saddle-point of high-energy string scattering
 \cite{Gross:1987ar}. 
More recently, they were found to  underpin the remarkable formulae for tree-level scattering
amplitudes in gauge theory and gravity that arise from twistor-string
theories \cite{Witten:2004cp} and the more recent CHY formulae 
\cite{Cachazo:2013hca}.  The derivation of these formulae from ambitwistor string theories \cite{Mason:2013sva} led to
proposals for formulae for loop amplitudes on higher-genus Riemann surfaces by
Adamo, Casali and Skinner (ACS) \cite{Adamo:2013tsa}, following the
standard string paradigm. They extended the CHY formulae for type II
supergravity in 10 dimensions to 1-loop in terms of scattering
equations on an elliptic curve (and, in principle, to $g$-loops on
curves of genus $g$). 
The only check that could be done at the time was factorisation at the
boundary of the moduli space.
It remained an open question as to whether these formulae 
compute amplitudes correctly due to the difficulties of solving these new  
scattering equations on the torus.

The ACS 1-loop proposal was investigated further by Casali and one of
us \cite{Casali:2014hfa}. 
To obtain the correct
pole structure of the field theory integrand, they were led to change the form of the
scattering equations as discussed below, 
and it was argued that the formulae were reproducing the known
integrands of four-points supergravity amplitudes at a triple cut.
A conjecture for ``scalar $n$-gon integrands'' was proposed, for an
expression on the elliptic curve that should give rise to loop
integrands based on permutations of polygon scalar integrals. This
supported the fact that the formulas could be valid for arbitrary loop
momentum, but the question remained open.

In~\cite{Geyer:2015bja}, we demonstrated how the formalism of the scattering equations gives rational expressions for the integrands of scattering amplitudes, making the loop-level problem essentially as simple as the tree-level one. We started by making further alterations to the
scattering equations on the torus so as to obtain a globally well-defined loop momentum.  
 We  then showed that formulae on the torus, such as the ACS and $n$-gon conjectures, reduce to ones on the
Riemann sphere.  This followed from a contour integral argument in the $\tau$-plane for the modular parameter
 of the torus (analogous to the use of the residue
theorem at tree level in \cite{Dolan:2013isa}).
On the Riemann sphere, the Jacobi theta functions reduce to elementary
rational functions.
We review our procedure below, and recall how these formulae now
involve off-shell momenta at a pair of points corresponding to the
loop momenta.  The newly inserted points are subject to off-shell
versions of the scattering equations.
These formulae on the sphere were furthermore generalised to provide conjectures for other amplitudes for which no first-principle
(i.e.\ ambitwistor string) derivation exists.
We proposed formulae for super Yang-Mills
theory at 1-loop, that were checked explicitly at four points and
numerically at five and six points.  The analogous formulae for
biadjoint scalar theories at 1-loop on the Riemann sphere were
subsequently studied in \cite{He:2015yua}, where they were also
verified at low point order
.

\paragraph{Summary of this paper.}

In this paper, we first review the ideas of \cite{Geyer:2015bja},
giving full details that were omitted through lack of space there and
some improvements. 
In \S\ref{sec:review}, we give a different formulation of the scattering equations on the torus following remarks from ACS that the one given in \cite{Geyer:2015bja} might not factorise correctly.  We then review how ambitwistor string amplitude formulae on the torus can be reduced to the Riemann sphere. The gravitational formulae are based on 1-loop extensions of the CHY Pfaffians.  These are obtained from the limit on the Riemann sphere of the worldsheet correlator of vertex operators on the torus described by ACS as a sum over spin structures (although their factorisation limit on the Riemann sphere misses some terms).  If one of the 1-loop Pfaffians is replaced by a 1-loop extension of the Parke-Taylor factor, super-Yang-Mills amplitudes are obtained.  
If both 1-loop Pfaffians are replaced by 1-loop Parke-Taylor factors, it was shown by \cite{He:2015yua} (see also \cite{Baadsgaard:2015hia}) that certain subtleties arise as additional degenerate solutions of the scattering equations contribute, and diagrams with bubbles on the external legs need to be considered.

The first main set of new results in this paper are presented in
\S\ref{sec:non-supersymm-theor}, where we provide a detailed study of the
individual contributions of the Neveu-Schwarz and Ramond sectors to the
one-loop amplitudes.
The Ramond sector running
through the loop corresponds to one Pfaffian, together with the
contribution from the odd spin structure which we can ignore in our
analysis by restricting the external kinematics to 7 or fewer
dimensions.\footnote{In string theory, the
  $\epsilon_{10}\epsilon_{10}$ term is actually an
  $\epsilon_{8}\epsilon_{8}$ one, see for instance
  \cite{Peeters:2000qj}. It is not clear how this situation transposes
to the ambitwistor string.}
Two further Pfaffians terms combine to give the contribution of the NS
sector running through the loop.  We express these two terms as a
reduced Pfaffian of CHY type for a larger $(n+2)\times (n+2)$ matrix
of co-rank two for the NS sector.  In this matrix we are able to see
the number of NS polarisation states running in the loop and we can
adjust this to give different theories in different dimensions (we
comment on dimensional reduction in \cref{sec:dimens-reduct}).
  If we drop the Ramond sector terms, we 
 obtain gauge and gravity amplitudes at one loop  that are
 non-supersymmetric. 
The resulting formulae have been subjected to various checks at
 low point order in this section and proved systematically in the subsequent one.
 
  A subtlety that arises here follows from an analysis of \cite{He:2015yua} in which it is argued that a degenerate class of solutions to the scattering equations might contribute
 nontrivially for non-supersymmetric theories and that of \cite{Baadsgaard:2015hia} who point out that on these degenerate solutions there is a risk of divergence, and some regularization might be required.   For our proposed integrands we show (in the subsequent section) that no regularization is required at the divergent solutions.  Nevertheless, we propose that these degenerate solutions should not be included as we see in our proof in the subsequent section that they do not contribute to the Q-cuts, and so are not needed in the final formula.  It seems most likely that they correspond to degenerate contributions that will vanish under dimensional regularization and are thrown away in the derivation of Q-cuts.

The second set of new results discussed in \S\ref{sec:factorization} give a full proof at one loop for the $n$-gon conjecture, and for the non-supersymmetric gauge, gravity and bi-adjoint scalar amplitudes.  The basic strategy is to study factorisation of the Riemann sphere.   The only poles that can arise in the formula, apart from the explicit $1/\ell^2$, where $\ell$ is the loop momentum, arise from factorisation of the Riemann sphere, i.e., the bubbling off of an additional sphere.  We can use this to identify all the poles involving the loop momentum and the corresponding residues.  We can also use factorisation ideas to identify the fall-off as $\ell\rightarrow \infty$.  This immediately gives  the poles and residues in the case of the $n$-gon conjecture.  For  gauge and gravity amplitudes, we also need to study the Pfaffians that arise (the Parke-Taylor factors in the Yang-Mills cases are rather easier to understand in this context). The poles and residues that we find give perfect agreement with the Q-cut representation of the amplitude, as obtained recently in \cite{Baadsgaard:2015twa}, and this completes a proof of our formulae; the Q-cut procedure applied to our formula will yield the correct Q-cut representation.\footnote{We remark that  the results of \cite{Cachazo:2015aol} (published after the first version of our paper) show that our formulae contain also spurious poles, which we do not discuss here. They show however, that these  terms will not contribute  to the Q-cut nor to the final  scattering amplitude, because they vanish in dimensional regularization.  Their paper serves to fill in a gap in our proof.}
  We are restricted to a proof for the non-supersymmetric theories because we do not have  formulae for tree amplitudes with two Ramond sector particles.

In \S\ref{sec:discussion}, we summarise and discuss further aspects and developments of these ideas.

\section{Review}\label{sec:review}

\subsection{The scattering equations on a torus}
\label{sec:scatt-equat-torus}

In this section, we define the scattering equations on a torus.  These are motivated by the definitions given in \cite{Adamo:2013tsa, Casali:2014hfa}, but the definition has been changed so that they are holomorphic and single valued on the torus with a well defined loop momentum.   

We use the complex coordinate $z$ on the elliptic curve $\Sigma_q =\C/\{\Z\oplus \Z \tau\}$ where $q=\e^{2\pi i \tau}$.
The scattering equations are equations for $n$ points $z_i\in \Sigma_q$ that depend on a choice of $n$ momenta $k_i \in \R^{1,d-1}$, $i=1,\ldots n$.  To define them we first construct a meromorphic 1-form $P(z,z_i|q)d z$ on $\Sigma_q$, with values in $\C^d$, that satisfies
\begin{equation}
\dbar P=2\pi i \sum_i k_i \bar \delta(z-z_i) dz\, , 
\label{eq:dbarP}
\end{equation}
where we define the complex double delta function by
\begin{equation}
  \label{eq:deltabar-def}
\bar \delta (f(z)):= \dbar \frac{1}{2\pi i f(z)} =\delta(\Re f)\delta(\Im f) d\overline{f(z)}\, .
\end{equation}
Introducing $\ell\in\R^{1,d-1}$ to parametrise the zero modes of $P$, and
setting $\;z_{ij}=z_i-z_j$, our choice of solution of
eq.~\eqref{eq:dbarP} for $P(z,z_i|q)$ is
\begin{equation}
\label{P-def}
P(z,z_i) =2\pi i\,\ell d z + \sum_i k_i \Bigg(
\frac{\theta'_1 (z-z_i)}{\theta_1 (z-z_i)} + \frac{\theta'_1
  (z_i-z_0)}{\theta_1 (z_i-z_0)} + \frac{\theta'_1 (z_0-z)}{\theta_1
  (z_0-z)}
  \Bigg)dz\, .
\end{equation}
Here the prime denotes $\partial/\partial z$,  $z_0$ is some choice of reference point, and $\theta_1=\theta_{11}$ where
standard theta functions are defined by
\begin{equation}
\theta_{ab}(z|\tau)=
\sum_{n\in\Z}
q^{(n-a/2)^2/2}
 e^{ 2i\pi (z-b/2)(n-a/2)}\, .\label{eq:theta-def}
\end{equation}
Since $\theta_1(z)\underset{z\to0}{\sim} z$, there are poles at $z=z_i$, $i=1,\ldots , n$ but momentum conservation implies that the coefficient of $\theta'_1(z_0-z)/\theta_1(z_0-z)$ is in fact zero, so $P$ is holomorphic at $z_0$.  We include the last term to make the double periodicity manifest.    
Theta functions are trivially periodic under $z\rightarrow z+1$, but
under $z\rightarrow z+\tau$ we have 
\begin{equation}
\frac{ \theta'(z+\tau)}{\theta(z+\tau)}
=\frac{\theta'(z)}{\theta(z)}-2\pi i\, .\label{eq:theta-periodicity}
\end{equation}
It is easy to see that our  expression for $P$ is doubly periodic in $z$ as a consequence of momentum conservation, but it is also doubly periodic in the $z_i$ as a consequence of the extra last term involving the reference point in \eqref{P-def}.  
Using this, we define the scattering equations to be
\begin{equation}
\mathrm{Res}_{z_i} P^2(z)=2k_i\cdot P(z_i)=0\, ,  \qquad
P^2(z_0')=0\,. \label{eq:SE}
\end{equation}
where $z'_0$ is another choice of reference point. On the support of the other scattering equations, $P^2(z_0')$ is global and holomorphic in $z'_0$ and hence independent of  this $z_0'$. Because the sum of residues of $P^2$ vanishes, the first scattering
equation follows from those at $i=2,\ldots, n$. Translation invariance
of the framework implies that we must fix the location of $z_1$ by
hand. On the support of the equations at $z_i$, $P^2(z_0')$ is global
and holomorphic, hence constant in $z_0'$, depending only on
$\tau$. Therefore, the final equation $P^2(z_0')=0$ is independent of
$z_0'$ and serves to determine $\tau$.

Some remarks are in order here. Since our $P$ is meromorphic and
doubly periodic both in $z$ \emph{and} the $z_i$, so are the
scattering equations.  It differs from the previous versions in the
literature in the choice of an additive `constant' term in $\ell$ that
depends on the $z_i$ and $k_i$.  The ACS version is not holomorphic in
the $z_i$; this leads to non-holomorphic scattering equations and it
was argued in \cite{Casali:2014hfa} that they do not give the correct
$1/\ell^{2}$ pole structure. A holomorphic version was proposed there
for which factorisation was checked, which is also the version used in
\cite{Ohmori:2015sha}.  However, that version is not doubly periodic
so the scattering equations are not well defined on the elliptic curve
for fixed constant $\ell$; there are different numbers of solutions on
the different fundamental domains of the lattice as well as those
related by $SL(2,\Z)$ as observed numerically
in~\cite{Casali:2014hfa}.\footnote{This fact leads to a well-known
  apparent ambiguity in the definition of the loop momentum in all
  first quantized theories (worldline, strings \cite{D'Hoker:1988ta}).
  This ambiguity drops out of the physical observables after
  integration of the loop momentum and does not alter the modular
  properties of the string amplitudes.
However, the case of the first quantized ambitwistor string is undoubtedly more
subtle because of the presence of the scattering equations and the fact that we must integrate only over a real contour in the loop momentum variable. Therefore
we must proceed by making two assumptions.  Firstly, we must cure the ambiguity in the loop momentum in
the integrand by defining  $P$ by 
\eqref{P-def}. Secondly, we want to define the integration cycle of
the theory (in the sense of \cite{Ohmori:2015sha,Witten:2012bh}) as
including only the solutions to the scattering equations in the
fundamental domain, as described below.}
The version in \cite{Geyer:2015bja} is
holomorphic and doubly periodic, but concerns were raised about
factorisation by Adamo, Casali \& Skinner, who suggested this approach.

With this version of the scattering equations, the ACS proposal for the 1-loop integrand of type-II supergravity amplitudes takes the form
\begin{equation}
\cM^{(1)}
_{\mathrm{SG}}
=\int \cI_q 
\,d^d \ell \,  d\tau \, \bar\delta (P^2(z_0'))\prod_{i=2}^n
\bar\delta(k_i\cdot P(z_i)) dz_i^2  \,.\label{elliptic-amp}
\end{equation}
Here we have written $\bar\delta(k_i\cdot P(z_i)) dz_i^2$ to give an expression that in total transforms as a 1-form.  This is because $P$ is a 1-form but $\bar \delta(f)$ has negative weight in $f$, so that we need two $dz_i$ factors to yield a 1-form.  In the ACS proposal, it is assumed that we are in the critical case of type II supergravity with $d=10$.  In this case, $\cI_q=\cI(k_i,\epsilon_i, z_i|q)$ is a function also of $\epsilon_i$,  the gravitational polarisation data, and is  the expression obtained as a sum over spin structures of the worldsheet correlator of vertex operators. It consists of certain Pfaffians and theta constants that arise as partition functions that are described later, and in more detail in \cite{Adamo:2013tsa}.   In this special case, this formula is doubly periodic in the $z_i$ and modular invariant, i.e., invariant under\footnote{The invariance under $\tau\rightarrow \tau+1$ is trivial, and  under $\tau\rightarrow -1/\tau$, we have $\ell\rightarrow \tau\ell$, $dz\rightarrow dz/\tau$, and the transformation of $\cI_q$ can be deduced from conventional string theory since the worldsheet correlator is essentially the square of the holomorphic part of the worldsheet correlator there. For the counting in the $n$-gon case, observe that $\bar\delta(P^2(z_0))$ transforms also as $\tau^{-2}$, as it implicitly has a factor of $dz_0^2$.\label{fn:mod}}   $\tau\rightarrow \tau+1, -1/\tau$.

In \cite{Casali:2014hfa}, it was shown that, when $n=4$, as in conventional string theory, $\cI_q$ is independent of $z_i$ and $q$, and factors out of the integral.  The nontrivial remaining integral is the $n=4$ version of the more general integral
\begin{equation}
\cM^{(1)}_{n\mathrm{-gon}}=\int \,d^{d} \ell \,  d\tau \, \bar\delta
(P^2(z_0'))\prod_{i=2}^n \bar\delta(k_i\cdot P(z_i)) d z_i^2
\,,\label{eq:ngon}
\end{equation}
where the integral can be checked to be modular invariant for
dimension $d=2n+2$ (see the modular weight of $\ell$ in
footnote~\ref{fn:mod}).  In \cite{Casali:2014hfa}, this was
conjectured to be equivalent to a sum over permutations of $n$-gons
and, if so at $n=4$, this would confirm the 4-particle supergravity conjecture at
1-loop.

In both formulae, leaving aside the integration over the loop momentum variable $\ell$, there are as many delta functions as integration variables and these restrict the integral to a sum over a discrete set of solutions to the scattering equations.  Each term in that sum consists of the integrand evaluated at the corresponding solution divided by a Jacobian.  

\subsection{From a torus to a Riemann sphere}
Here we use a residue theorem (or integration by parts in our
notation) to reduce the formula on the elliptic curve to one on the
nodal Riemann sphere at $q=0$. 
Our argument relies on the intuitive fact that the scattering equation imposed
by $\bar\delta (P^2(z_0))$ has a separate status from the others,
serving to fix $\tau$, and can be analysed on the $\tau-$plane alone. We can use the residue theorem to convert it into an equation
enforcing $q=0$.  Such `global residue theorems' have already been
applied to tree-level CHY formulae by \cite{Dolan:2013isa} to relate
the scalar CHY formulae to their Feynman diagrams.  
We apply the same strategy here, and we will be left with scattering
equations that have off-shell momenta associated to $\ell$, and a
formula for the 1-loop integrand based on off-shell scattering
equations on the Riemann sphere (in fact a forward limit of those of
\cite{Naculich:2014naa}).

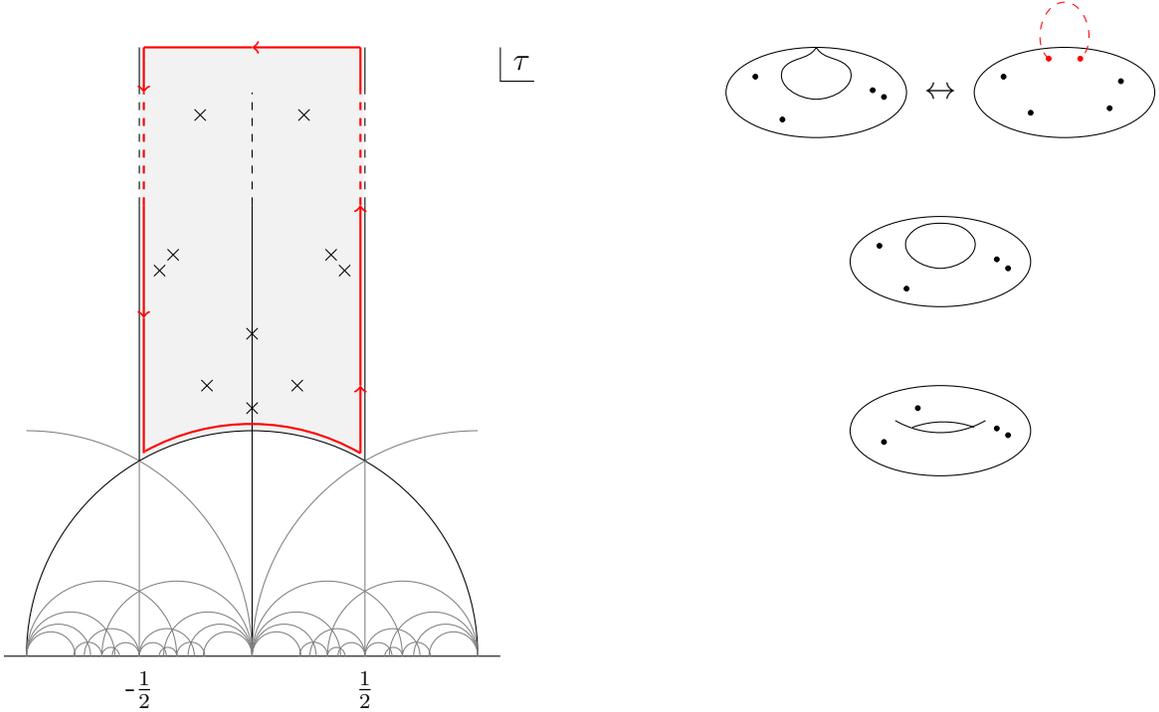
\begin{figure}[ht]
 \begin{tikzpicture} [scale=3]
  \filldraw [fill=light-gray, draw=white] (0.5,0.866) arc [radius=1, start angle=60, end angle= 120] -- (-0.5,2.7) -- (0.5,2.7) -- (0.5,0.866);
  \draw [black] (1.1,0) -- (-1.1,0);
  \draw [black] (0,0) -- (0,2);
  \draw [black,dashed] (0,2) -- (0,2.5);
  \node at (0.5,-0.15) {$\frac{1}{2}$};
  \node at (-0.5,-0.15) {-$\frac{1}{2}$};
  \node at (1.195,2.63) {$\tau$};
  \draw (1.1,2.7) -- (1.1,2.55) -- (1.25,2.55);
  \draw [gray] (0,0) arc [radius=1, start angle=0, end angle= 90];
  \draw [black] (1,0) arc [radius=1, start angle=0, end angle= 180];
  \draw [gray] (1,1) arc [radius=1, start angle=90, end angle= 180];
  \draw [gray] (0.5,0) -- (0.5,0.866);
  \draw [gray] (-0.5,0) -- (-0.5,0.866);
  \draw (0.5,0.866) -- (0.5,2);
  \draw (-0.5,0.866) -- (-0.5,2);
  \draw [dashed] (0.5,2) -- (0.5,2.5);
  \draw [dashed] (-0.5,2) -- (-0.5,2.5);
  \draw (0.5,2.5) -- (0.5,2.7);
  \draw (-0.5,2.5) -- (-0.5,2.7);
  \draw [red,thick,->] (0.48,2.7) -- (0,2.7);
  \draw [red,thick] (0,2.7) -- (-0.48,2.7);
  \draw [red,thick] (0.48,2.5) -- (0.48,2.7);
  \draw [red,thick,<-] (-0.48,2.5) -- (-0.48,2.7);
  \draw [red,thick,dashed] (0.48,2) -- (0.48,2.5);
  \draw [red,thick,dashed] (-0.48,2) -- (-0.48,2.5);
  \draw [red,thick,->] (0.48,1.2) -- (0.48,2);
  \draw [red,thick,->] (0.48,0.9) -- (0.48,1.2);
  \draw [red,thick,<-] (-0.48,1.5) -- (-0.48,2);
  \draw [red,thick] (-0.48,0.9) -- (-0.48,1.5);
  \draw [red,thick] (0.48,0.9) arc [radius=0.97, start angle=60, end angle= 119];
  \draw (-0.025,1.125) -- (0.025,1.075);
  \draw (-0.025,1.075) -- (0.025,1.125);
  \draw (0.225,1.225) -- (0.175,1.175);
  \draw (0.225,1.175) -- (0.175,1.225);
  \draw (-0.225,1.225) -- (-0.175,1.175);
  \draw (-0.225,1.175) -- (-0.175,1.225);
  \draw (-0.025,1.455) -- (0.025,1.405);
  \draw (-0.025,1.405) -- (0.025,1.455);
  \draw (-0.385,1.735) -- (-0.435,1.685);
  \draw (-0.385,1.685) -- (-0.435,1.735);
  \draw (0.385,1.735) -- (0.435,1.685);
  \draw (0.385,1.685) -- (0.435,1.735);
  \draw (0.325,1.805) -- (0.375,1.755);
  \draw (0.325,1.755) -- (0.375,1.805);
  \draw (-0.325,1.805) -- (-0.375,1.755);
  \draw (-0.325,1.755) -- (-0.375,1.805);
  \draw (0.205,2.425) -- (0.255,2.375);
  \draw (0.255,2.425) -- (0.205,2.375);
  \draw (-0.205,2.425) -- (-0.255,2.375);
  \draw (-0.205,2.375) -- (-0.255,2.425);
  \draw [gray] (0,0) arc [radius=0.333, start angle=0, end angle= 180];
  \draw [gray] (-0.334,0) arc [radius=0.333, start angle=0, end angle= 180];
  \draw [gray] (0.666,0) arc [radius=0.333, start angle=0, end angle= 180];
  \draw [gray] (1,0) arc [radius=0.333, start angle=0, end angle= 180];
  \draw [gray] (0,0) arc [radius=0.196, start angle=0, end angle= 180];
  \draw [gray] (-0.608,0) arc [radius=0.196, start angle=0, end angle= 180];
  \draw [gray] (0.392,0) arc [radius=0.196, start angle=0, end angle= 180];
  \draw [gray] (1,0) arc [radius=0.196, start angle=0, end angle= 180];
  \draw [gray] (0,0) arc [radius=0.142, start angle=0, end angle= 180];
  \draw [gray] (-0.716,0) arc [radius=0.142, start angle=0, end angle= 180];
  \draw [gray] (0.284,0) arc [radius=0.142, start angle=0, end angle= 180];
  \draw [gray] (1,0) arc [radius=0.142, start angle=0, end angle= 180];
  \draw [gray] (0,0) arc [radius=0.108, start angle=0, end angle= 180];
  \draw [gray] (-0.784,0) arc [radius=0.108, start angle=0, end angle= 180];
  \draw [gray] (0.216,0) arc [radius=0.108, start angle=0, end angle= 180];
  \draw [gray] (1,0) arc [radius=0.108, start angle=0, end angle= 180];
  \draw [gray] (0.5,0) arc [radius=0.122, start angle=0, end angle= 180];
  \draw [gray] (-0.5,0) arc [radius=0.122, start angle=0, end angle= 180];
  \draw [gray] (0.744,0) arc [radius=0.122, start angle=0, end angle= 180];
  \draw [gray] (-0.256,0) arc [radius=0.122, start angle=0, end angle= 180];
  \draw [gray] (0.5,0) arc [radius=0.060, start angle=0, end angle= 180];
  \draw [gray] (-0.5,0) arc [radius=0.060, start angle=0, end angle= 180];
  \draw [gray] (0.620,0) arc [radius=0.060, start angle=0, end angle= 180];
  \draw [gray] (-0.380,0) arc [radius=0.060, start angle=0, end angle= 180];
  \draw [gray] (0.666,0) arc [radius=0.039, start angle=0, end angle= 180];
  \draw [gray] (0.412,0) arc [radius=0.039, start angle=0, end angle= 180];
  \draw [gray] (-0.588,0) arc [radius=0.039, start angle=0, end angle= 180];
  \draw [gray] (-0.334,0) arc [radius=0.039, start angle=0, end angle= 180];
  \draw [gray] (0.334,0) arc [radius=0.062, start angle=0, end angle= 180];
  \draw [gray] (0.790,0) arc [radius=0.062, start angle=0, end angle= 180];
  \draw [gray] (-0.666,0) arc [radius=0.062, start angle=0, end angle= 180];
  \draw [gray] (-0.210,0) arc [radius=0.062, start angle=0, end angle= 180];
  \draw (3.05,1) circle [x radius=0.4, y radius=0.2];
  \draw (2.85,1.045) arc [radius=0.4, start angle=240, end angle=300];
  \draw (3.2,1.015) arc [radius=0.4, start angle=70, end angle=110];
  \draw [fill] (2.95,1.1) circle [radius=.3pt];
  \draw [fill] (2.8,0.95) circle [radius=.3pt];
  \draw [fill] (3.35,0.98) circle [radius=.3pt];
  \draw [fill] (3.3,1.01) circle [radius=.3pt];
  \draw (3.05,1.75) circle [x radius=0.4, y radius=0.2];
  \draw (3.05,1.92) to [out=180, in=60] (2.9,1.85) to [out=250, in=180] (3.05,1.72) to [out=0, in=290] (3.2,1.85) to [out=120, in=0] (3.05,1.92);
  \draw [fill] (2.78,1.82) circle [radius=.3pt];
  \draw [fill] (2.9,1.63) circle [radius=.3pt];
  \draw [fill] (3.35,1.72) circle [radius=.3pt];
  \draw [fill] (3.3,1.76) circle [radius=.3pt];
  \draw (2.5,2.5) circle [x radius=0.4, y radius=0.2];
  \draw (2.5,2.7) to [out=240, in=60] (2.35,2.6) to [out=250, in=180] (2.5,2.47) to [out=0, in=290] (2.65,2.6) to [out=120, in=300] (2.5,2.7);
  \draw [fill] (2.23,2.57) circle [radius=.3pt];
  \draw [fill] (2.35,2.38) circle [radius=.3pt];
  \draw [fill] (2.8,2.48) circle [radius=.3pt];
  \draw [fill] (2.75,2.51) circle [radius=.3pt];
  \draw (3.6,2.5) circle [x radius=0.4, y radius=0.2];
  \draw [dashed,red] (3.53,2.65) to [out=140, in=180] (3.6,2.9) to [out=0, in=40] (3.67,2.65);
  \draw [fill] (3.33,2.57) circle [radius=.3pt];
  \draw [fill] (3.45,2.41) circle [radius=.3pt];
  \draw [fill] (3.8,2.43) circle [radius=.3pt];
  \draw [fill] (3.85,2.55) circle [radius=.3pt];
  \draw [fill,red] (3.53,2.65) circle [radius=.3pt];
  \draw [fill,red] (3.67,2.65) circle [radius=.3pt];
  \node at (3.05,2.5) {$\leftrightarrow$};
 \end{tikzpicture}
 
\caption{Residue theorem in the fundamental domain.}
\label{fund-dom}
\end{figure}  

In order to obtain a formula for the amplitude on the Riemann sphere,
we need that $\cI_q:=\cI(\ldots|q)$ be holomorphic as a function of
$q$ on the fundamental domain
$D_\tau=\{|\tau|\geq 1, \Re \tau \in[-1/2,1/2]\}$ for the modular
group.  In the case of the $n$-gon conjecture below, we will put
$\cI_q=1$, and this will be obvious.
For supergravity, however, $\cI_q$ is a product of two 1-loop analogues of CHY Pfaffians that in particular have many contributions of the form $1/ \theta_1(z_i-z_j)$, which provide potential poles when $z_i\rightarrow z_j$, and it is conceivable that as $q$ varies, these might lead to poles in $q$. However,  such poles are suppressed by the scattering equations for generic choices of the momenta.   As $z_i\to z_j$ for $i,j \in I$ and $I$ some subset of $1, \ldots , n$, $P$ is well approximated by its counterpart on the Riemann sphere near the concentration point, and it is easily seen that such factorisation of the $z_i$ can only occur if the corresponding partial sum of the momenta for $i\in I$ becomes null.  See \S\ref{sec:factorization} for a detailed discussion of the argument.  Thus, if the momenta are in general position, we cannot have $z_i\to z_j$ on the support of $k_i\cdot P(z_i)=0$, and so our $\cI_q$ will have no poles.  

 It was shown in
\cite{Adamo:2013tsa} that the holomorphicity of the supergravity
integrand at $q=0$ is a consequence of the GSO projection. For other
values of $q$, the possible poles in the integrand can only occur when
$z_i\rightarrow z_j$, but the standard factorisation argument
\cite{Dolan:2013isa} applies here also to imply that this can only
happen when the momenta are factorising and hence non-generic. The main
argument is then
\begin{eqnarray}
\cM^{(1)}_{SG}&=&\int \cI_q\,d^d \ell \,  \frac{d q}{ q} \, \dbar \left(\frac{1}{2\pi i P^2(z_0')}\right)\prod_{i=2}^n \bar\delta(k_i\cdot P(z_i)) d z_i \nonumber \\
&=& -\int \cI_q\,d^d \ell \,  \dbar\left(\frac{d q}{2\pi i q} \right)\,  \frac{1}{ P^2(z_0')}\prod_{i=2}^n \bar\delta(k_i\cdot P(z_i)) d z_i \nonumber \\
&=& -\int \cI_0\,d^d \ell \,    \frac{1}{  P^2(z_0')}\prod_{i=2}^n \bar\delta(k_i\cdot P(z_i)) d z_i \,
\Big|_{q=0}\, .
\end{eqnarray}
In the first line, we put $d\tau=dq/(2\pi i q)$ and inserted the definition of $\bar\delta (P^2(z_0'))$.  In the second line, we integrated by parts in the domain $D_\tau$, yielding a delta function supported at $q=0$ that is then integrated out.  The boundary terms cancel because of the modular invariance.  This is equivalent to a contour integral argument in the fundamental domain $D_\tau$, as in figure \ref{fund-dom}. The sum of the residues at the poles of $1/P^2(z_0' |q)$ simply gives the contribution from the residue at the top, $q=0$, since the contributions from the sides and the unit circle cancel by modular invariance.  

The fundamental domain for $z$ maps,
\begin{equation}
\sigma= \e^{2\pi i (z-\tau/2)}\,,
\end{equation}
 to $\{\e^{-\pi\Im \tau} \leq|\sigma|\leq\e^{\pi\Im \tau}\}$, with the
 identification $\sigma\sim q \sigma$. As  $q\rightarrow0$, we obtain
 $\sigma\in\CP^1$ with $0,\infty$  identified, giving a double point
 corresponding to the pinching of $\Sigma_q$ at a non-separating
 degeneration as illustrated in figure \ref{fund-dom}. We have
 $dz=\frac{d\sigma}{2\pi i\sigma}$ and, at $q=0$,
\begin{equation}
 \frac{\theta'_1 (z-z_i)}{\theta_1 (z-z_i)}dz=\frac{\pi}{\tan \pi
   (z-z_i)} \,dz =-\frac{d\sigma}{2\,\sigma} +
 \frac{d\sigma}{\sigma-\sigma_i}\,.
 \label{eq:thetaq0}
\end{equation}

Using momentum conservation we obtain 
\begin{equation}
\label{Psphere}
P(z)=P(\sigma)= \ell \,\frac{d\sigma}{\sigma} +\sum_{i=1}^n \frac{k_i\, d\sigma}{\sigma-\sigma_i}\, ,
\end{equation}
where here we have translated $\ell$ by\footnote{This was the extra term that we included in  $\ell$ to make $P$ doubly periodic at constant $\ell$.  In this limit on the Riemann sphere, it no longer plays a useful role and corresponds to taking $\sigma_0=\infty$. } $\sum_{i} k_i \cot \pi
z_{i0 }$. 

If we now consider the function $P^2(\sigma)$, we find that it has double poles at $0$, $\infty$ (along with the usual simple poles at $\sigma_i$).  Defining 
$$
S=P^2-\ell^2 \rd \sigma^2/\sigma^2,
$$
we find $S$ now has only simple poles. The vanishing of the residues of $S$ gives our {\em off-shell} scattering equations 
\be{SE2n}
0=\mathrm{Res}_{\sigma_i}S =k_i\cdot P(\sigma_i) = \frac{k_i\cdot \ell}{\sigma_i} + \sum_{j\neq i}\frac{k_i\cdot k_j}{\sigma_i-\sigma_j}\, , 
\ee
at $\sigma_i$.  The sum of the residues of $\sigma_\alpha\sigma_\beta S$ must vanish with $\sigma_\alpha=(1,\sigma)$ in affine coordinates, so that the equations for $i=2,\ldots , n$ imply the vanishing of the residues of $S$ at $\sigma_1$, $0$ and $\infty$. Thus any $n-1$ of these equations imply all $n+2$, hence $S$ is holomorphic and, having negative weight, vanishes, so that $P^2=\ell^2\rd \sigma^2/\sigma^2$.

With this, the 1-loop formula becomes
\begin{equation}\label{1-loop}
\cM^{(1)}_{SG}= -\int \cI_0\,d^d \ell \, \frac{1}{\ell^2}\prod_{i=2}^n \bar\delta(k_i\cdot P(\sigma_i))\frac{d \sigma_i}{\sigma_i^2}\, ,
\end{equation}
where we have used the identity $\bar\delta(\lambda f)=\lambda^{-1}\bar\delta(f)$ to give $\bar{\delta}(k_i\cdot P(z_i)) dz_i=
\bar{\delta}(k_i\cdot P(\sigma_i))d\sigma_i/\sigma_i^2$. The formula \eqref{1-loop} is our new proposal for the supergravity loop integrand, with $\cI_0$ the $q=0$ limit of the ACS correlator.

For the simpler `$n$-gon' conjecture presented in \cite{Casali:2014hfa}, we now take $\cI_q=1$. For both this and supergravity, modular invariance is no longer an issue on the Riemann sphere, and the new formulae make sense in any dimension.  However, the link to a formula on the elliptic curve will only be valid in the critical dimension.

The integration by parts sets $q$ to zero, and this is also the regime in
full string theory where one extracts the field theory or
$\alpha'\to0$ limit of the amplitudes. The difference here is that
this limit is obtained by application of the residue theorem, so we
are not throwing away any terms, whereas in string theory we
would be projecting out the contribution of massive states running in the
loop by doing so.
At the moment it is unclear if the similarity between the method we
use here and string theory is just a consequence of the fact that both
strings are physical and hence factorise properly at the boundary of
the moduli space, or if this goes deeper.
In any case, the similarity between the $\alpha'\to0$ limit and our
IBP will allow us to reuse some standard technology from string
theory.\footnote{Several restrictions apply; there is no fully fledged
  well defined Heterotic ambitwistor string (see
  appendix~\ref{sec:motiv-form-ambitw}), there are no winding modes
  which can become massless at self-dual radii in compactifications to
  enhance abelian to non-abelian gauge groups, see
  appendix~\ref{sec:dimens-reduct}, and it is yet unknown how to
  include the contribution of non-perturbative states of
  supergravity.}

\subsection{The \texorpdfstring{$n$}{n}-gon conjecture, partial fractions and shifts}\label{sec:n-gon}

The question arises as to how the $\ell$ appearing in \eqref{1-loop} relates to the loop momentum flowing in any given propagator.  We will see that the answer requires a new way of expressing 1-loop amplitudes. The expression \eqref{1-loop} is a representation of the one-loop contribution to the scattering amplitude of a theory specified by $\cI_0$.  In this subsection, we consider the choice where $\cI_0=1$, which was conjectured in \cite{Casali:2014hfa} to give rise to a permutation sum of polygons.  
When $n=4$, the $n$-gon conjecture implies the supergravity conjecture \cite{Casali:2014hfa}.

For $n=4$, the off-shell scattering equations can be solved exactly with two solutions\footnote{This problem is identical to that arising in factorisation as studied in \cite{Casali:2014hfa} except that now $\ell$ is off-shell.  It was  conjectured there to have $(n-1)!-2(n-2)!$ solutions giving 2 at $n=4$.} given explicitly in \S\ref{4pt-soln}.   
After substituting into \eqref{1-loop}, this yields
\begin{equation}\label{4pt-result}
\hat{\cM}^{(1)}_4 = \frac{1}{\ell^2} \sum_{\sigma \in S_4} 
\frac{1}{2\ell\cdot k_{\sigma_1}(2\ell\cdot (k_{\sigma_1}+ k_{\sigma_2})+2k_{\sigma_1}\cdot k_{\sigma_2})(-2\ell\cdot k_{\sigma_4})}\,,
\end{equation}
where we defined the loop integrand as $\hat{\cM}^{(1)}$,
\begin{equation}\label{4pt-result-int}
{\cM}^{(1)} = \int d^D\ell \;\hat{\cM}^{(1)}.
\end{equation}
This result is not obviously equivalent to the permutation sum of the boxes
\begin{equation}
  \label{eq:Ibox}
I^{1234} = \frac{1}{\ell^2(\ell+k_3)^2(\ell+k_3+k_4)^2(\ell-k_2)^2}\,,
\end{equation}
as the only manifest propagator in $\cM^{(1)}_4$ is the pre-factor ${1}/{\ell^2}$, and all the other denominator factors are linear in $\ell$. 
However,  the partial fraction identity
\begin{equation}
  \label{eq:Di-partfrac}
  \frac1{\prod_{i=1}^{n}D_{i}} =
  \sum_{i=1}^{n}\frac{1}{D_{i}\prod_{j\neq i}(D_{j}-D_{i})}
\end{equation}
can be applied to a contribution such as \eqref{eq:Ibox} (this
identity is easily proven by induction).   The right-hand-side of this identity is a sum of terms with a single factor of the type $D_i=(\ell+K)^2$, and several factors of the type $D_{j}-D_{i}=2\ell\cdot K +\mathcal{O}(\ell^0)$. We then perform a shift in the loop momentum for each term such that the corresponding $D_i$ is simply $\ell^2$. Applying this procedure to the permutation sum, we precisely obtain $\hat{\cM}_4^{(1)}$.

We are now in a position to address the $n$-gon conjecture of \cite{Casali:2014hfa}. It states that $\cI$=1 corresponds to a permutation-symmetric sum of $n$-gons, which can be written as
\begin{align}
\hat{\cM}^{(1)}_{n\text{-gon}}= \frac1{\ell^2} \sum_{\rho\in S_n}
 \frac{1}{{\prod_{i=1}^{n-1} }\bigg(2\ell\cdot\sum_{j=1}^i k_{\rho_i} +\left(\sum_{j=1}^i
k_{\rho_i}\right)^2\bigg) } .
\label{formulangon}
\end{align}
In our previous work \cite{Geyer:2015bja}, we verified this equality analytically at four points, using the explicit solutions to the scattering equations in Appendix~\ref{4pt-soln}, and numerically at five points.   We will see later in \S\ref{factorize} that this can be proved by factorisation arguments.

The $n=2$ and $3$ examples are also instructive.  The bubble (2-gon) example gives\footnote{Henceforth, we use a capital symbol $K$ to distinguish a possibly massive momentum.}
\begin{align}
  \frac1{\ell^2(\ell+K)^2} & =
  \frac{1}{\ell^2(2\ell\cdot K+K^2)} + \frac{1}{(\ell+K)^2(-2\ell\cdot K-K^2)} \nonumber \\
& \stackrel{\textrm{shift}}{\longrightarrow} \frac1{\ell^2}
\left( \frac1{2\ell\cdot K+K^2} + \frac1{-2\ell\cdot K+K^2} \right),\label{bubble}
\end{align}
where a shift $\ell \to \ell -K$ was applied to the second term. If $K$ is null, the bubble vanishes, which is also the result of dimensional regularisation. The triangle (3-gon) with massless corners, $k_1^2=k_2^2=k_3^2=0$, also vanishes:
\begin{align}
  \frac1{\ell^2(\ell+k_1)^2(\ell-k_3)^2} \stackrel{\textrm{shift}}{\longrightarrow} -\frac{ \ell\cdot(k_1+k_2+k_3)}{4\,\ell^2(\ell\cdot k_1)(\ell\cdot k_2)(\ell\cdot k_3)} =0.\label{triangle}
\end{align}

 The partial fraction identity \eqref{eq:Di-partfrac} can also be
 application of the residue theorem to the following integral
 \begin{equation}
   \label{eq:Di-partfracz}
   \frac{1}{2\pi i}\oint_{|z|=\epsilon}\frac1{z\,\prod_{i=1}^{n}(D_{i}-z)}\,.
 \end{equation}
Typical integrands for theories like gauge theory or gravity depend on the loop momentum also in the numerators, and not simply through propagator factors in the denominators. The loop momentum in the numerators should also be shifted.
 For more general amplitudes, this can be achieved with a
 shift in the loop momentum together with a contour integral argument, and this has been explored and considerably
 generalised in  \cite{Baadsgaard:2015twa} and reviewed in \S\ref{sec:factorization}.

\subsection{Supersymmetric theories}
\label{sec:supersymm-theor-1}

Supergravity and Yang-Mills one-loop amplitudes were also expressed in \cite{Geyer:2015bja}
on the Riemann sphere using different choices for $\cI_q$ in \eqref{1-loop}.
Here we present further details of these calculations. While the former are readily derived from
the type II RNS ambitwistor string, the Yang-Mills one was simply
conjectured (see some motivational comments in
\cref{sec:motiv-form-ambitw}). We show that these
integrands pass several non-trivial consistency checks, and later show that they factorise correctly in \S\ref{sec:factorization}.

\subsubsection{Supergravity}
\label{sec:supergravity}

Let us start by recalling the form of genus-one graviton amplitudes in
ambitwistor string, as derived by ACS in ref.~\cite{Adamo:2013tsa}.
As in the usual RNS string, the worldsheet correlator incorporates a
GSO projection to remove the unwanted states. The integrand $\cI_q$ is the worldsheet correlator of the vertex operators
 resulting from Wick contractions. The
main difference from a conventional string integrand is the absence
of $XX$ contractions. This forbids in particular the appearance of an
exponential factor of the form
$\exp(\sum k_{i}k_{j} \langle X(z_{i})X(z_{j})\rangle)$ since these holomorphic plane
waves have trivial OPE's.
The $\cI_q$ of the ACS proposal is a sum over spin structures on the
torus.  The odd-odd spin structure gives a fermionic 10-dimensional
zero-mode integral that leads to a $10$-dimensional Levi-Civita
$\epsilon$ symbol.  This will vanish if all the polarisation data and momenta are taken in less than $10$-dimensions and for simplicity we will assume that this is the case in the following\footnote{In doing this, we miss the term that leads to the Green-Schwarz anomaly \cite{Mafra:2012kh}.} and focus only on the even spin structures labelled by $\bd{\alpha}=2,3,4$ (with $\bd{\alpha}=1$ the odd one).  With this, the ACS proposal for the amplitude explicitly reads as \eqref{elliptic-amp} with 
\begin{equation}
\cI_q:=\frac14\sum_{\bd{\alpha};\bd{\beta}=2,3,4}
	    (-1)^{\bd{ \alpha}+ \bd{\beta}} Z_{ \bd{\alpha}; \bd{\beta}}(\tau)
	\ {\rm Pf}(M_{ \bd{\alpha}})\,{\rm Pf}(\widetilde{M}_{ \bd{\beta}})
\label{e:graveven}
\end{equation}
(the $1/4$ comes from the two GSO projections).  Above, in
eq.~\eqref{elliptic-amp}, we referred to ${\cI_q}$ as the worldsheet
CFT correlator that includes partition functions.  The vertex
operators are naturally a product of two factors that we refer to as
``left'' and ``right'', and since left and right parts essentially
decouple, the full correlator decomposes also as a product as follows:
\begin{equation}
  \mathcal{I}_q = \mathcal{I}^L_q\,\mathcal{I}^R_q, \qquad \text{with} \qquad \mathcal{I}^L_q=\frac12\left(Z_2\,\pf(M_2) -Z_3\,\pf(M_3) +Z_4\,\pf(M_4)\right) .
\label{defIsugra}
\end{equation}  
with a similar but tilde'd definition for $\mathcal{I}^{R}$.
The matrix $M_{\bd\alpha}$ is a generalisation of the CHY matrix, and
comes from a straightforward application of Wick's theorem to the left parts of the vertex operators in the spin structure $\bd{\alpha}$. It
is defined as
\begin{equation}
M_\alpha = \left( \begin{matrix}
A \;\;& -C^T \\
C  \;\;&  B
\end{matrix} \right),
\end{equation}
where
\begin{equation}
A_{ij} = k_i\cdot k_j \,S_\alpha(z_{ij}|\tau), \quad B_{ij} =
\epsilon_i\cdot \epsilon_j \,S_\alpha(z_{ij}|\tau), \quad C_{ij} =
\epsilon_i\cdot k_j \,S_\alpha(z_{ij}|\tau) \quad \text{for} \quad
i\neq j,\label{eq:Malpha}
\end{equation}
and
\begin{equation}
A_{ii} = 0, \quad B_{ii} = 0, \quad C_{ii} = -\epsilon_i\cdot P(z_i),
\end{equation}
where $P(z_i)$ was given in \eqref{P-def}.
The torus free fermion propagators, or Szeg\H{o} kernels, are defined by
\begin{equation}
  S_{\bd \alpha}(z_{ij}|\tau) = \frac{\partial_{z}\theta_{1}(0|\tau)}{\theta_{1}(z_{ij}|\tau)}\frac{\theta_{\bd \alpha}(z_{ij}|\tau)}{\theta_{\bd \alpha}(0|\tau)}\sqrt{\d z_i}\sqrt{\d z_j}
\label{eq:szego-even}
\end{equation}
in even spin-structures $\bd{\alpha}=2,3,4$, while $S_{1}$ is given by
\begin{equation}
  \label{eq:szego-odd}
  S_{1}(z|\tau)=\frac{\partial_{z}\theta_{1}(z|\tau)}{\theta_{1}(z|\tau)}\,.
\end{equation}
Here $\bd{\alpha}:=(a,b)=(0,0),(0,1),(1,0)$ are the even characteristics and $(1,1)$ is the odd one. In the notation $\bd{\alpha}=1,2,3,4$ used above, these correspond to $\bd{\alpha}=3,4,2$ and $\bd{\alpha}=1$, respectively.

The tilded matrix $\tilde M_{\bd\alpha}$ is defined as $M_{\bd\alpha}$
is, with different polarisation vectors $\tilde\epsilon$, such that
the polarisation tensors
$\epsilon_i^{\mu\nu}=\epsilon_i^\mu \tilde\epsilon_i^\nu$
correspond to the NS-NS states of supergravity,
graviton, the dilaton and the B-field. In terms of $\epsilon^{\mu\nu}$, the dilaton corresponds to the trace part, the $B$-field to the skew part, and the graviton to the traceless symmetric part.

The $Z_{\bd{\alpha; \beta}}$ are the CFT partition functions in the
${\bd{\alpha}; \bd{\beta}}$ spin structures. In terms of the ambitwistor string theory, they have a factor of $1/\eta(\tau)^{16}$ from the $(P,X)$ system and $\theta_{\bd\gamma}(0|\tau)^4/\eta(\tau)^4$ from each of the $\psi_r$, $r=1,2$,  fermion systems.  The power 16 is
twice the number of transverse directions of $10d$ Minkowski
space, while the fourth power is one-half as appropriate for the spin
1/2 fermions.\footnote{Alternatively the partition functions of the various
  $(b,c)$ and $(\beta,\gamma)$ ghosts account for the reduction to the
physical transverse modes.}
The theta functions have been defined in \eqref{eq:theta-def}, while
the Dedekind eta function is defined by
\begin{equation}
  \label{eq:dedekind-def}
  \eta(\tau)=q^{1/24}\prod_{n=1}^{\infty}{(1-q^{n})}\,.
\end{equation}

With this, the $Z_{\bd \gamma}$ of \eqref{defIsugra} 
are given by
\begin{align}
  \label{eq:pt-funs}
Z_{\bd \gamma} =
  \frac{\theta_{\bd{\gamma}}(0|\tau)^4}{\eta(\tau)^{12}}\quad ({\bd
  \gamma}={\bd \alpha, \bd\beta}) \, .
\end{align}

Applying our contour argument to go from the torus to the nodal
Riemann sphere, we are interested in the limit $q\to0$. The partition
functions do possess $1/\sqrt{q}$ poles which extract higher order
terms in the Szeg\H{o} kernels. Hence we need the following
$q$-expansions:
\begin{equation}
Z_2(\tau) = 16+O(q^{2}),\quad
Z_3(\tau)= \frac1{\sqrt{q}}+8+O(q),\quad
Z_4(\tau) = \frac1{\sqrt{q}}-8+O(q).
\label{eq:Z-q-exp}
\end{equation}
and
\begin{equation}
\begin{aligned}
S_1(z_{ij}|\tau) &\to \frac{1}{2}\,\frac{1}{\sigma_ i-\sigma_ j} \left(\sqrt{\frac{\sigma_ i}{\sigma_ j}}+ \sqrt{\frac{\sigma_ j}{\sigma_ i}}\right)  \sqrt{d\sigma_ i} \sqrt{d\sigma_ j}  ,
  \\ 
S_2(z_{ij}|\tau) &\to \frac{1}{2}\,\frac{1}{\sigma_ i-\sigma_ j} \left(\sqrt{\frac{\sigma_ i}{\sigma_ j}}+ \sqrt{\frac{\sigma_ j}{\sigma_ i}}\right)  \sqrt{d\sigma_ i} \sqrt{d\sigma_ j} ,
  \\ 
S_3(z_{ij}|\tau) &\to \left(\frac{1}{\sigma_ i-\sigma_ j}+\sqrt{q} \,\frac{\sigma_ i-\sigma_ j}{{\sigma_ i}{\sigma_ j}}\right) \sqrt{d\sigma_ i} \sqrt{d\sigma_ j} ,
  \\ 
S_4(z_{ij}|\tau) &\to \left(\frac{1}{\sigma_ i-\sigma_ j}-\sqrt{q} \,\frac{\sigma_ i-\sigma_ j}{{\sigma_ i}{\sigma_ j}}\right) \sqrt{d\sigma_ i} \sqrt{d\sigma_ j} .
\end{aligned}\label{eq:szego-limit}
\end{equation}
in terms of the coordinates $\sigma=e^{2\pi i (z-\tau/2)}$.
The limit of $P(z_i)$ required for the components $C_{ii}$ was already
given in \eqref{Psphere}. 
The $q=0$ residue of \eqref{defIsugra} is then given by 
\begin{equation}
  \begin{aligned}
  \label{eq:sugra-sphere}
   \mathcal{I}^{L}=\frac{1}{2\sqrt{q}} \left( \pf(M_{3})\big|_{q^{0}} -
    \pf(M_{4})\big|_{q^{0}} \right)+
\frac12\left(\pf(M_{3}) \big|_{\sqrt{q}}+\pf(M_{4}) \big|_{\sqrt{q}}\right)+\\
  4\left(\pf(M_{3}) \big|_{q^0}+\pf(M_{4}) \big|_{q^0}-2\pf(M_{2} )\big|_{q^0}\right)+O(\sqrt{q})
  \end{aligned}
\end{equation}
where the symbol $(\cdot)|_{q^r}$ with $r=0,1/2$ means that we extract
the coefficient of $q^r$ in the Taylor expansion around
$q=0$.\footnote{In the original ACS paper, the $O(\sqrt{q})$ were not
    included in the analysis of the factorisation channel.}

Some simplifications occur at this stage. Firstly, it is easy to see from eq.~\eqref{eq:szego-limit} that 
\begin{equation}
  \pf(M_{3})\big|_{q^{0}} = \pf(M_{4})\big|_{q^{0}}
  \label{eq:m3m4q0}    
  \end{equation}
  which reflects the projection of the ambitwistor string ``NS--tachyon''
  (we come back on this later). Then, we also have that
\begin{equation}
  \pf(M_{3})\big|_{\sqrt{q}} = -\pf(M_{4})\big|_{\sqrt{q}}
  \label{eq:m3m4q12}    
  \end{equation}
Using the two previous identities, we finally lend on eq.~(11)
presented in our previous work \cite{Geyer:2015bja}, which we
reproduce here;
\begin{equation}
  \label{eq:IL0-sugra}
  \mathcal{I}_{0}^{L} = \pf(M_{3}) \big|_{\sqrt{q}}+
  8\left(\pf(M_{3}) \big|_{q^0}-\pf(M_{2} )\big|_{q^0}\right)\,.
\end{equation}

The structure of this object may appear to be quite complicated with
regard to the extreme simplicity of one-loop maximal supergravity
integrands. It is actually a lot simpler than it looks, thanks to the
use of standard stringy theta function identities~\cite{Casali:2014hfa}.
The simplest identities involve products of up to three Szeg\H{o} kernels,
\begin{equation}
\sum_{\alpha=2,3,4} (-1)^\alpha Z_\alpha \prod_{r=1}^m S_\alpha (w_{(r)}|\tau) = 0, \quad 
\text{for} \quad m=0,1,2,3,
\end{equation}
where the $w_{(r)}$ can be arbitrary. At $n=0$, this is the well known
Jacobi's identity
$\theta_{2}(0,\tau)^{4}-\theta_{3}(0,\tau)^{4}+\theta_{4}(0,\tau)^{4}=0$. For
$m>3$, the analogous identities are valid only for
\begin{equation}
\label{sumztheta}
w_{(1)}+\ldots+w_{(m)}=0.
\end{equation}
Let us consider the case $m=4$. In our application, the condition on
the $w_{(r)}$ is naturally achieved by the set
$(z_{ij},z_{jk},z_{kl},z_{li})$, and the corresponding identity is
\begin{equation}
  \sum_{\alpha=2,3,4} (-1)^\alpha Z_\alpha \prod_{r=1}^4 S_\alpha
  (w_{(r)}|\tau)  = (2\pi)^4 \,,
\end{equation}
where we have ellipsed the global form degree $dz_i dz_j dz_k dz_l$.
Applied to \eqref{defIsugra}, these identities tell us that
$\mathcal{I}^L$ is a constant for four-point scattering
\cite{Casali:2014hfa}. This follows from the structure of the
Pfaffians, or equivalently from the structure of the vertex
operators. As in string theory, only the terms with 8 $\psi$'s or more
contribute. At $n$ points, each term in $\Pf(M_\alpha)$ is a product
of $m$ Szeg\H{o} kernels of type $\alpha$ and $m-n$ factors
$C_{ii}$. The Szeg\H{o} kernels of type $\alpha$ appear with arguments
which precisely satisfy the condition \eqref{sumztheta}. At four
points, the sum over spin structures ensures that no $C_{ii}$
contributes, as $m<4$ for those terms, whereas the $m=4$ identity
implies that $\mathcal{I}^L$ is a constant. For $n>4$, the sum over
spin structures ensures that there are no terms with more than $n-4$
factors of the type $C_{ii}$. The classical reference
is~\cite{Tsuchiya:1988va}, while~\cite{Broedel:2014vla} provide an
all-$n$ form for them.
Since the loop momentum enters explicitly in $\mathcal{I}^L$ only
through $C_{ii}$, this means that $\mathcal{I}^L$ is a polynomial of
order $n-4$ in the loop momentum, which is always contracted with a
polarisation vector. This discussion holds for any value of $\tau$. In
the limit $q\to 0$ ($\tau\to i\infty$), the Riemann identities become
algebraic identities, and can be easily checked at low multiplicity.

\subsubsection{Super-Yang-Mills theory}
\label{sec:super-yang-mills}
The supergravity amplitude was derived in \cite{Geyer:2015bja} from the genus-one ambitwistor string expression of \cite{Adamo:2013tsa}, as described above. However, a Yang-Mills analogue of the latter on the torus is not known, despite the progress in formulating an ambitwistor string version of gauge theory at tree level \cite{Mason:2013sva,Geyer:2014fka,Casali:2015vta,Adamo:2015gia}. Nevertheless, a proposal for super Yang-Mills amplitudes was given in \cite{Geyer:2015bja}, using the tree-level case and the relation between gauge theory and gravity.

At tree level, CHY \cite{Cachazo:2013hca} found that the expression for the gauge theory amplitude is obtained from the supergravity one by substituting one Pfaffian by a Parke-Taylor factor. The fact that gauge theory has only one Pfaffian, depending on a set of polarisation vectors ($\epsilon_i^\mu$), while gravity has two Pfaffians, each depending on a different set of polarisation vectors ($\epsilon_i^\mu$ and $\tilde\epsilon_i^\mu$), is a clear manifestation of gravity as a `square' of gauge theory, in agreement with the Kawai-Lewellen-Tye relations \cite{Kawai:1985xq} and with the Bern-Carrasco-Johansson (BCJ) double copy \cite{Bern:2008qj,Bern:2010ue}. At loop-level, the BCJ double copy is known to hold at one-loop in a variety of cases, including certain classes of amplitudes at any multiplicity \cite{Mafra:2012kh,Boels:2013bi,Bjerrum-Bohr:2013iza,Mafra:2014gja,He:2015wgf}, so it is natural to propose that one-loop formulae based on the scattering equations will also exhibit this property. The proposal of \cite{Geyer:2015bja} is that the super Yang-Mills amplitude is determined by\footnote{This gives the planar (single-trace) contribution to the amplitude. At one loop, the non-planar (double-trace) contribution is determined by the planar part for any gauge theory involving only particles in the adjoint representation of $SU(N_c)$ \cite{Bern:1994zx}.}
\begin{equation}
\mathcal{I}^{SYM}= \mathcal{I}^L_0\,\mathcal{I}^{PT},
\label{defIsym}
\end{equation}
where $\mathcal{I}^L_0$ is defined in \eqref{eq:IL0-sugra}. The one-loop analogue of the Parke-Taylor factor was conjectured to be
\begin{equation}
\label{eq:PTloop-def}
\mathcal{I}^{PT}=\sum_{i=1}^n \frac{\sigma_{\ell^+\,\ell^-}}{\sigma_{\ell^+\,i}\sigma_{i+1\, i}\sigma_{i+2\, i+1}\ldots \sigma_{i+n\,\,\ell^-}}\, ,
\end{equation}
where $\sigma_{\ell^+}$ and $\sigma_{\ell^-}$ represent the pair of insertion points of the loop momentum, and where we identify the labels $i\sim i+n$. With our choice of coordinates on the Riemann sphere we have fixed $\sigma_{\ell^+}=0$ and $\sigma_{\ell^-}=\infty$, so that
\begin{equation}
\mathcal{I}^{PT}= - \sum_{i=1}^n \frac1{\sigma_{i}\sigma_{i+1\, i}\sigma_{i+2\, i+1}\ldots \sigma_{i+n-1\,i+n}}\, .
\end{equation}

In Appendix~\ref{sec:motiv-form-ambitw}, we present a motivation for our conjecture based on the
heterotic ambitwistor models.

\subsubsection{Checks}
\label{sec:checks}

The conjectures above were verified explicitly in \cite{Geyer:2015bja}
at low multiplicity. In a later section, we will provide a proof for
these conjectures at any multiplicity, based on the factorisation
properties of the formulae. In this section, however, we will simply
give more details of the checks reported in
\cite{Geyer:2015bja}. These were performed in four dimensions, where
there exist simple known expressions for $N=4$ super Yang-Mills theory
and $N=8$ supergravity. These expressions should match our $D=10$
formulae when we restrict the external data to four dimensions, as we
argue in \cref{sec:dimens-reduct}. We make use of the spinor-helicity
formalism; see e.g.~\cite{Elvang:2013cua} for a review. The
polarisation vectors for positive and negative helicities are
represented as
\begin{equation}
\epsilon_i^{(+)}= \frac{|\eta\rangle [i|}{\langle i \eta\rangle} \ , \qquad 
\epsilon_i^{(-)}= \frac{|i\rangle [\eta|}{[ i \eta ]} \ ,
\label{treenum5}
\end{equation}
where $\eta=|\eta\rangle [\eta|$ is a reference vector. The four-point checks were performed analytically, using the solutions to the scattering equations presented in Appendix~\ref{4pt-soln}, whereas the five-point checks were performed numerically.

For the theories at hand, due to supersymmetry, the only external helicity configurations which lead to a non-vanishing amplitude have at least two particles of each helicity. We verified that our formulae for both super Yang-Mills theory and supergravity vanish if that condition is not satisfied.

For $n=4$, non-vanishing amplitudes must have two particles of each helicity type. Let us label the negative-helicity particles as $r$ and $s$. The loop integrands for these super Yang-Mills and supergravity amplitudes are well known \cite{Green:1982sw,Bern:2010tq}. After the application of our shift procedure, they are given by
\begin{align}
\hat{\cM}^{(1)}_{4,\,\text{SYM}}=  \frac1{\ell^2} \sum_{\rho\in \text{cyc}(1234)}
 \frac{N_4}{{\prod_{i=1}^{3} }\left(2\ell\cdot\sum_{j=1}^i k_{\rho_i} +\Big(\sum_{j=1}^i
k_{\rho_i}\Big)^2\right) }
\end{align}
and
\begin{align}
\hat{\cM}^{(1)}_{4,\,\text{SG}}=  \frac1{\ell^2} \sum_{\rho\in S_4}
 \frac{N_4^2}{{\prod_{i=1}^{3} }\left(2\ell\cdot\sum_{j=1}^i k_{\rho_i} +\Big(\sum_{j=1}^i
k_{\rho_i}\Big)^2\right) },
\end{align}
where we sum over cyclic permutations for gauge theory and over all permutations for gravity. The numerator 
\begin{equation}
N_4= \langle rs \rangle^4 \,\frac{[12][34]}{\langle 12 \rangle\langle 34 \rangle}
\end{equation}
is given by a permutation-invariant kinematic function, times the factor $\langle rs \rangle^4$ involving the negative-helicity particles. The fact that this numerator appears squared in gravity with respect to gauge theory is the simplest one-loop example of the BCJ double copy.\footnote{In the supergravity case, we could have distinguished the choice of $r,s$ in $\epsilon_i^\mu$ and $\tilde r,\tilde s$ in $\tilde \epsilon_i^\mu$, with the obvious outcome of substituting $\langle rs \rangle^8$ by $\langle rs \rangle^4\langle \tilde r\tilde s \rangle^4$ in the final result.} We verified that these expressions match our formulae. The amplitude for supergravity follows from the $n$-gon conjecture \eqref{formulangon}. This is due to the fact that, at four points, the quantities $\mathcal{I}^L$ and $\mathcal{I}^R$ are constant \cite{Casali:2014hfa}, as discussed above, each coinciding with the numerator $N_4$.

For $n=5$, we will consider the case of two negative-helicity particles (for two positive helicities, we should simply exchange the chirality of the spinors in the formulae). The complete integrands involve both pentagon and box integrals. In their shifted form, they are given by
\begin{align}
\hat{\cM}^{(1)}_{5,\,\text{SYM}}=\frac{1}{\ell^2}& \sum_{\rho\in \text{cyc}(12345)}
\frac{1}{\prod_{i=1}^4 \big(2\ell\cdot\sum_{j=1}^i k_{\rho_i} +(\sum_{j=1}^i k_{\rho_i})^2\big)} \times
\nonumber \\
& \times \left( N^{\text{pent}}_{\rho_1\rho_2\rho_3\rho_4\rho_5} \,+ 
\frac{1}{2} \sum_{i=1}^4 N^{\text{box}}_{[\rho_i,\rho_{i+1}]} \,\frac{2\ell\cdot\sum_{j=1}^i k_{\rho_i} +(\sum_{j=1}^i k_{\rho_i})^2}{k_{\rho_i}\cdot k_{\rho_{i+1}}} \,   \right)
\label{5ptsym}
\end{align}
and
\begin{align}
\hat{\cM}^{(1)}_{5,\,\text{SG}}=\frac{1}{\ell^2}& \sum_{\rho\in S_5}
\frac{1}{\prod_{i=1}^4 \big(2\ell\cdot\sum_{j=1}^i k_{\rho_i} +(\sum_{j=1}^i k_{\rho_i})^2\big)} \times
\nonumber \\
& \times \left( (N^{\text{pent}}_{\rho_1\rho_2\rho_3\rho_4\rho_5})^2 \,+ 
\frac{1}{4} \sum_{i=1}^4 (N^{\text{box}}_{[\rho_i,\rho_{i+1}]})^2 \,\frac{2\ell\cdot\sum_{j=1}^i k_{\rho_i} +(\sum_{j=1}^i k_{\rho_i})^2}{k_{\rho_i}\cdot k_{\rho_{i+1}}} \,   \right).
\label{5ptsugra}
\end{align}
A valid choice for the pentagon and box numerators was present in \cite{Carrasco:2011mn},
\begin{equation}
N^{\text{pent}}_{12345} = \langle rs \rangle^4 \,
\frac{[12][23][34][45][51]}{[12]\langle 23\rangle [34]\langle 41\rangle
-\langle 12\rangle [23] \langle 34\rangle [41]}
\end{equation}
and
\begin{equation}
N^{\text{box}}_{[1,2]} = N^{\text{box}}_{[1,2]345} = N^{\text{pent}}_{12345}-N^{\text{pent}}_{21345}.
\end{equation}
The numerator $N^{\text{box}}_{[1,2]}$ corresponds to a box with one massive corner, $K=k_1+k_2$, and is independent of the ordering of 3,4,5. We verified that our expressions match these integrands. There are other choices of numerators leading to the same integrands, such as the choice in \cite{He:2015wgf}, which extends to any multiplicity (in the MHV sector, i.e. two negative helicities), and arises as the dimensional reduction of the superstring-derived numerators of \cite{Mafra:2014gja}. In that case, the pentagon numerators depend on the loop momentum, but \eqref{5ptsym} and \eqref{5ptsugra} take the same form, because the relevant shifts are of the type
\begin{equation}
N^{\text{pent}}_{12345}(\ell-k_1)=N^{\text{pent}}_{23451}(\ell).
\end{equation}
Here, we define the loop momentum as flowing between the first and last leg of the numerator, and this behaviour with respect to shifts follows from cyclic symmetry. It is therefore trivial to translate between the shifted representation of the integrand and the standard one.

\section{Non-supersymmetric theories}
\label{sec:non-supersymm-theor}

In this section, we describe new formulae for Yang-Mills theory and
gravity amplitudes in the absence of supersymmetry. The main tool in
arriving at these formulae is the detailed analysis of the sum over
spin structures (or GSO sum), which was part of the formulae for supergravity and
super Yang-Mills theory presented in \cite{Geyer:2015bja} and reviewed
above in \eqref{eq:IL0-sugra}.

On the torus, these GSO sectors correspond to the various states
propagating in the loop. Once taken down to the sphere, we will see how
they provide amplitudes with $n$ external on-shell gravitons (or
gluons) and additional NS-NS, R-NS, NS-R or R-R additional
off-shell states (resp. NS or R), running in the loop.  
In particular, we are able to see that the $M_2$ contribution in $\eqref{eq:IL0-sugra}$ corresponds to the Ramond sector.  Furthermore the $M_3$ contributions naturally combine as a reduced Pfaffian of an $(n+2)\times (n+2)$ matrix in which the number of NS states running in the loop can be chosen at will.  

Taken individually, these one-loop amplitudes are non-supersymmetric. Using
these building blocks, one can engineer various types of
amplitudes.
For gravity, we discuss both NS-NS gravity (graviton, dilaton,
B-field) and pure Einstein gravity (graviton only).
We later show that our formulae match the known 4-point one-loop amplitudes in Yang-Mills theory and gravity, in a certain helicity sector.

A subtlety that arises however is that a class of degenerate solutions to the scattering equation becomes nontrivial (and in fact potentially divergent) for these non-supersymmetric amplitudes, as described by \cite{Baadsgaard:2015hia,
He:2015yua} for the biadjoint scalar theories.  So we first rephrase the scattering equations in a more $SL(2,\C)$ invariant manner to give a less degenerate formulation of these solutions.  In the next section, we will see that the contribution of these degenerate solutions is finite for our proposed formulae, and can furthermore be discarded without changing the integrated amplitude.

\subsection{General form of the one-loop scattering equations}
\label{sec:general-form-one}
Before proceeding, we rewrite our previous expressions in
order to use their different building blocks for non-supersymmetric
theories. The reason for this, as pointed out in \cite{He:2015yua}, is
that the one-loop scattering equations on the sphere possess, in their
general form, more solutions than are obvious from
\eqref{SE2n}.  We used part of the $SL(2,\mathbb{C})$
freedom on the Riemann sphere to fix the positions of the
loop-momentum insertions at $\sigma_{\ell^+}=0$ and
$\sigma_{\ell^-}=\infty$ as was natural from the
degeneration of the torus into a nodal Riemann sphere. However there are  extra solutions to the scattering
equations for which $\sigma_{\ell^+}=\sigma_{\ell^-}$ with the remaining
$\sigma_i$ then satisfy the tree-level scattering equations (these solutions do arise in the previous gauge fixing with $\sigma_1=1$, in which all the $\sigma_i= 1$ also, but this gauge is much more awkward to deal with for these solutions).  We will see that these
extra solutions do not contribute to the formulae for maximal
supergravity and super Yang-Mills theory given in \cite{Geyer:2015bja}
and reviewed above, but do contribute for generic theories, e.g. the
biadjoint scalar theory. As discussed in \cite{He:2015yua}, the total
number of solutions contributing is $(n-1)!-(n-2)!$, of which
$(n-1)!-2(n-2)!$ are the `regular' solutions considered in
\eqref{SE2n}, and $(n-2)!$ are the `singular' solutions for which
$\sigma_{\ell^+}=\sigma_{\ell^-}$.

Hereafter, we will write the one-loop formulae based on the general scattering equations as
\be{1-loopgeneral}
\cM^{(1)}= -\int d^d \ell \, \frac{1}{\ell^2} \int 
\frac{d\sigma_{\ell^+} d\sigma_{\ell^-} d^n\sigma}{\text{vol}\,SL(2,\mathbb{C})} \;\; \hat{\cI}
\;\bar\delta(\text{Res}_{\sigma_{\ell^+}} S) \bar\delta(\text{Res}_{\sigma_{\ell^-}} S) \prod_i{}' \;\bar\delta(\mathrm{Res}_{\sigma_i}S), 
\ee
where we should not fix the positions of both $\sigma_{\ell^+}$ and $\sigma_{\ell^-}$ in choosing the $SL(2,\mathbb{C})$ gauge, to avoid losing the `singular' solutions. Since
\begin{equation}
 P=   \left( \frac{\ell}{\sigma-\sigma_{\ell^+}} - \frac{\ell}{\sigma-\sigma_{\ell^-}} + \sum_{i=1}^n 
\frac{k_i}{\sigma-\sigma_i} \right) d\sigma
\end{equation}
and
\begin{equation}
S= P^2-  \left( \frac{\ell}{\sigma-\sigma_{\ell^+}} - \frac{\ell}{\sigma-\sigma_{\ell^-}}\right)^2 d\sigma^2,
\end{equation}
the scattering equations take the form
\begin{align}
\mathrm{Res}_{\sigma_i}S = k_i\cdot P(\sigma_i) = \frac{k_i\cdot
  \ell}{\sigma_i-\sigma_{\ell^+}} -\frac{k_i\cdot
  \ell}{\sigma_i-\sigma_{\ell^-}} + \sum_{j\neq i}\frac{k_i\cdot
  k_j}{\sigma_i-\sigma_j} =0\, ,  \\
\mathrm{Res}_{\sigma_{\ell^-}}S = -\sum_{i}\frac{\ell\cdot k_j}{\sigma_{\ell^-}-\sigma_i} =0\, ,\\
\mathrm{Res}_{\sigma_{\ell^+}}S = \sum_{i}\frac{\ell\cdot k_j}{\sigma_{\ell^+}-\sigma_i} =0\, .
\end{align}
In the formula \eqref{1-loopgeneral}, the prime on the product denotes the fact that only $n-3$ of these equations should be enforced (with the those at $\sigma_{\ell^\pm}$ on an equal footing now with the others).  The three remaining equations are a consequence of the three relations between the equations following from the identities described following \eqref{SE2n}.

The interesting part of formula \eqref{1-loopgeneral} is the quantity $\cI$ specifying the theory. We introduced
\begin{equation}
\hat{\cI} = \frac1{(\sigma_{\ell^+\,\ell^-})^4} \, \cI\,,
\label{ihatdef}
\end{equation}
so that $\hat{\cI}$ has the same $SL(2,\mathbb{C})$ weight in $\{\sigma_{\ell^+},\sigma_{\ell^-},\sigma_i\}$, as required by the integration, whereas $\cI$ has zero weight in $\{\sigma_{\ell^+},\sigma_{\ell^-}\}$. The $n$-gon formula now corresponds to
\begin{equation}
\cI^{n-\text{gon}} = \prod_{i=1}^n \left( \frac{\sigma_{\ell^+\,\ell^-}}{\sigma_{i\,\ell^+}\,\sigma_{i\,\ell^-}} \right)^2\,.
\end{equation}
The relation to the $n$-gon representation in \cite{He:2015yua} follows from the identity
\begin{equation}
\sum_{\alpha\in S_n} \, \frac{\sigma_{\ell^+\,\ell^-}}{ \sigma_{\ell^+\,\alpha(1)}\,
\sigma_{\alpha(1)\,\alpha(2)} \ldots \sigma_{\alpha(n-1)\,\alpha(n)}\, \sigma_{\alpha(n)\,-\ell}} =
\prod_{i=1}^n  \frac{\sigma_{\ell^+\,\ell^-}}{\sigma_{\ell^+\,i}\,\sigma_{i\,\ell^-}}.
\end{equation}
This follows by  induction and  partial fractions.

For supergravity and for super Yang-Mills theory, we have
\begin{equation}
\cI^{SG} = \cI^L_0\, \cI^R_0 \qquad \text{and} \qquad \cI^{SYM} =  \cI^L_0 \,\cI^{PT} \,,
\end{equation}
where $\cI^{PT}$ was defined in \eqref{eq:PTloop-def}.
The quantities $\cI^L_0$ and $\cI^R_0$ are defined as in \eqref{eq:IL0-sugra}, but the Szeg\H{o} kernels in the matrices $M_\alpha$ are now expressed as
\begin{equation}
\begin{aligned}
S_2(z_{ij}|\tau) &\to \frac{1}{2}\,\frac{1}{\sigma_{i\,j}} \left(\sqrt{\frac{\sigma_{i\,\ell^+}\,\sigma_{j\,\ell^-}}{\sigma_{j\,\ell^+}\,\sigma_{i\,\ell^-}}}+ \sqrt{\frac{\sigma_{j\,\ell^+}\,\sigma_{i\,\ell^-}}{\sigma_{i\,\ell^+}\,\sigma_{j\,\ell^-}}} \right)  \sqrt{d\sigma_ i} \sqrt{d\sigma_ j} ,
  \\ 
S_3(z_{ij}|\tau) &\to \frac{1}{\sigma_{i\,j}} \left(1 +\sqrt{q} \;\frac{(\sigma_{i\,j}\,\sigma_{\ell^+\,\ell^-})^2}{\sigma_{i\,\ell^+}\,\sigma_{i\,\ell^-}\,\sigma_{j\,\ell^+}\,\sigma_{j\,\ell^-}}\right) \sqrt{d\sigma_ i} \sqrt{d\sigma_ j} ,
\end{aligned}\label{eq:szegolimitsl2c}
\end{equation}
in the limit $q\to 0$.

Regarding the `singular' solutions to the scattering equations, it is
clear that they do not contribute in the $n$-gon case, since
$\hat{\cI}^{n-\text{gon}}\to 0$ for $\sigma_{\ell^+}\to
\sigma_{\ell^-}$. However, they do contribute in the case of the
non-supersymmetric Yang-Mills and gravity formulae to be presented
below, and some care is needed in their evaluation, due to the factor
$(\sigma_{\ell^+\,\ell^-})^{-4}$ in \eqref{ihatdef}. It is easy to see
that $\cI^{PT}\sim
\mathcal{O}\left((\sigma_{\ell^+\,\ell^-})^2\right)$. We will show in \S\ref{sec:powercounting} that
\begin{align}
\pf(M_2)|_{q^0}, \, \pf(M_3)|_{q^0} & \sim \mathcal{O}\left((\sigma_{\ell^+\,\ell^-})^2\right), \nonumber\\
\left(\pf(M_2)-\pf(M_3)\right)|_{q^0}, \, \pf(M_3)|_{\sqrt{q}} & \sim \mathcal{O}\left((\sigma_{\ell^+\,\ell^-})^3\right).
\end{align}
This is irrespective of the context there of taking the limit of large $\ell$, or of considering the `singular' solutions.
The contributions from these solutions to our formulae are therefore finite, as expected, and they vanish in the case of $\hat{\cI}^{SG}$ and $\hat{\cI}^{SYM}$.
Furthermore, we will see that the degenerate solutions do not contribute to the Q-cuts and hence to the integrated amplitudes, and so can be discarded.
It would, however, be useful to have an explicit formula for the limit.

\subsection{Contributions of GSO sectors and the NS Pfaffian}
\label{sec:analys-indiv-gso}
We now turn to the individual contributions of each GSO sector to the
supergravity amplitudes. This analysis is based on standard
string theory, the reader is referred to standard string textbooks such as that by Polchinski, or \cite{Adamo:2013tsa} for further details.

We work in dimension $d$ for $ d\leq10$ by
dimensional reduction from $d=10$. Since there are no winding modes, taking the radii of
compactification to zero is enough to decouple the Kaluza-Klein
modes, see appendix~\ref{sec:dimens-reduct} for further comments.
We consider first the ``left'' and ``right'' sectors
independently.\footnote{In string theory, this is justified by the
  chiral splitting of the worldsheet correlator whose dramatic consequences include  the KLT
  relations~\cite{Kawai:1985xq,BjerrumBohr:2010hn}. In the ambitwistor string this follows from KLT orthogonality
  \cite{Cachazo:2013gna}.} These
consists of $\mathcal{N}=1$ sYM multiplets in $d=10$, and their
dimensional reduction is well known~\cite{Brink:1976bc}. The 10
dimensional vector $A^{(10)}_{\mu}$ splits into a $d$-dimensional
vector and $10-d$ scalars -- we mention the case of fermions
below.

The important point for the present analysis is that the partition
functions $Z_{a,b}$ as defined in eq.~\eqref{eq:pt-funs} are those of
particular sectors of the theory. Precisely, $a=0$ and $a=1$
correspond to the NS and R sectors, while $b=0,1$ correspond to the
periodicity of the boundary conditions. Thus
\begin{center}
  \begin{tabular}[h]{rcl}
  $Z_{3},Z_{4}$& $\longleftrightarrow$ & NS sector,\\
  $Z_{1},Z_{2}$ &  $\longleftrightarrow$ &R sector.
  \end{tabular}
\end{center}
Here we will ignore the odd spin structure $Z_1$ as it only contributes when the kinematics are fully in $d=10$.  This is similarly the case for correlators, whose chiral residues at $q=0$
(i.e. $\mathcal{I}^{L}$ and $\mathcal{I}^{R}$) we gave in
eq.~(\ref{eq:sugra-sphere}).  So we define
\begin{align}
  \mathcal{I}_{NS} &=\pf(M_{3})\big|_{\sqrt{q}}+8\,
                                    \pf(M_{3})\big|_{q^{0}}\label{eq:chiral-NS}\,,\\ 
\mathcal{I}_{R\phantom{S}}&=8\pf(M_{2})\big|_{q^{0}}\label{eq:chiral-R}\,.
\end{align}

In $10$ dimensions, these correspond to chiral integrands for one vector
and one Majorana--Weyl fermion.
When we reduce  to $d<10$ dimensions, the problem that one
faces is how to decide which parts of the integrand
(\ref{eq:chiral-NS}) correspond to the $10-d$ scalars and which part
corresponds to the vector.
Following in particular the string theory analysis of
\cite{Tourkine:2012vx}, it is easy to identify first the scalar
contribution by reading off the (vanishing) coefficient
$\frac12(\pf(M_{3})\big|_{q^{0}}-\pf(M_{4})\big|_{q^{0}})$ of the
$1/\sqrt{q}$ pole in eq.~(\ref{eq:sugra-sphere}).  Ignoring
the minus sign of the GSO projection, it corresponds to the (vanishing)
propagation of the unphysical scalar state
$\delta(\gamma_{1}) \delta(\gamma_{2})c\tilde c \exp(i k\cdot X)$.
With this we identify the scalar integrand as
$\pf(M_{3})\big|_{q^{0}}$ (recalling eq.~(\ref{eq:m3m4q0})) and we can deduce  
\begin{subequations}
  \begin{align}
    \mathcal{I}^{L}_{\mathrm{scal}} &=  \pf(M_{3})\big|_{q^{0}}\label{eq:chiral-scalar}\\
    \mathcal{I}^{L}_{\mathrm{vect}} &= \pf(M_{3})\big|_{\sqrt{q}}+ (d-2)\,
                                      \pf(M_{3})\big|_{q^{0}}\label{eq:chiral-vector}\\
    \mathcal{I}^{L}_{\mathrm{ferm}} &=
                                      -c_{d}\,\pf(M_{2})\big|_{q^{0}}\, .\label{eq:chiral-fermion}
  \end{align}
\end{subequations}
The fermion integrand eq.~(\ref{eq:chiral-fermion}) comes
with a constant $c_{d}$ that follows from dimensional reduction of the 10d
Majorana-Weyl spinor, which produces an 8d Weyl spinor, four 6d
simplectic-Weyl spinors, and four 4d Majorana spinors. From
eq.~(\ref{eq:sugra-sphere}) we read off $c_{10}=8$, therefore we have
$c_{8}=8$, $c_{6}=2$, $c_{4}=2$. 

We can therefore obtain the reduced gravitational states in the loop by taking the tensor
product of the two sectors
\begin{subequations}
\label{eq:chiral-int-grav}
  \begin{align}
    \label{eq:grav-scalar}
    ( \pf(M_{3})\big|_{\sqrt{q}}+ (d-2)\,\pf(M_{3})\big|_{q^{0}})^{2}&=
    \mathcal{I}_{\mathrm{NS-NS-grav}}\\
    \label{eq:scalar}
    (\pf(M_{3})\big|_{q^{0}})^{2}&=\mathcal{I}_\mathrm{scalar}\\
    \label{eq:vect}
    (\pf(M_{3})\big|_{q^{0}})( \pf(M_{3})\big|_{\sqrt{q}}+ (d-2)\,\pf(M_{3})\big|_{q^{0}})&=\mathcal{I}_\mathrm{vector}\,,
  \end{align}
\end{subequations}
from the NS--NS sector. Here, by NS--NS gravity 
in eq.~(\ref{eq:grav-scalar}) we mean Einstein gravity coupled to a B-field and a dilaton.
The squares are to be understood as incorporating a replacement of
the $\epsilon$'s by $\tilde \epsilon$'s in the second factor.

In the R--NS and NS--R sectors, we have
\begin{subequations}
  \begin{align}
    \label{eq:fermion}
    -c_{d}\,\pf(M_{2})\big|_{q^{0}}\pf(M_{3})\big|_{q^{0}}&= \mathcal{I}_\mathrm{fermion}\,,\\
    \label{eq:gravitino}
    -c_{d}\,\pf(M_{2})\big|_{q^{0}}\left( \pf(M_{3})\big|_{\sqrt{q}}+ (d-2)\,
    \pf(M_{3})\big|_{q^{0}}\right)&= \mathcal{I}_\mathrm{gravitino}\,.
  \end{align}
\end{subequations}
The R--R states in $d=10$ simply involve the square
\begin{equation}
  \label{eq:RR}
  \left(8\pf(M_{2})\big|_{q^{0}}\right)^{2}=\mathcal{I}_{\mathrm{RR}}\,,
\end{equation}
and it would be interesting to investigate this sector further.

With these interpretations of how different fields in the loops correspond to different  ingredients of the 1-loop correlator, we can make the following proposals. 

\paragraph{Pure YM and gravity amplitudes.} Firstly, by adjusting the building
blocks in \eqref{eq:chiral-int-grav} in an appropriate way, we conjecture that a four-dimensional
one-loop pure gravity amplitude can be written as follows;
\begin{equation}
  \label{eq:pure-gravity}
  \mathcal{I}_{\mathrm{pure-grav}}^{(d=4)}=( \pf(M_{3})\big|_{\sqrt{q}}+ 2\,\pf(M_{3})\big|_{q^{0}})^{2}-2\,(\pf(M_{3})\big|_{q^{0}})^{2}
\end{equation}
where the subtraction removes the two scalar degrees of freedom of the
dilaton and B-field
\footnote{The single degree of freedom of the
  B-field in four dimensions is the axion.}. This subtraction is analogous to the
prescription of~\cite{Johansson:2014zca}, where scalars with fermionic
statistics were introduced to implement the BCJ double copy in
loop-level amplitudes of pure gravity theories.

Using the prescription reviewed in sec.~\ref{sec:super-yang-mills}, we
can also build four-dimensional pure YM amplitudes, by simply
multiplying the vector integrand of eq.~\eqref{eq:chiral-vector}
with the Parke-Taylor factor (\ref{eq:PTloop-def}),
\begin{equation}
  \label{eq:pure-YM}
  \mathcal{I}_{\mathrm{pure-YM}}^{(d=4)}=
  ( \pf(M_{3})\big|_{\sqrt{q}}+ 2\,\pf(M_{3})\big|_{q^{0}})\,\mathcal{I}^{PT}.
\end{equation}

We will perform checks on these amplitudes in the next
subsection and give a general proof in the next section. 
Note that although these standard string ideas are suggestive of the above proposals, they do not constitute a proof, so it is important to produce an independent proof.\footnote{Amongst other issues, a point that is missing is
  that the abelian gauge groups do not get enhanced at self-dual radii
  of compactification as there are no winding modes that could become massless.}

\paragraph{Pfaffian structure of the new amplitudes.}

A feature of the previous formulae is that they
provide information on the structure of tree-level amplitudes.
The finite residue that we extract at $q=0$ coincides with the residue
at the factorisation channel $q\simeq \ell^{2}\to0$. The only
difference between our expression and a ``single cut'' is the presence
of $1/\ell^{2}$ and the full $d$-dimensional integral
$\int d^{d}\ell$.
Therefore, we have  a variety of tree-level amplitudes with
$n+2$ (on-shell) particles, in a forward limit configuration where
$k_{n+1}=-k_{n+2}=\ell$ are off-shell, but traced over their polarization states.

One may therefore expect that  the integrands of the pure gravity and
Yang-Mills amplitudes \eqref{eq:pure-gravity} and (\ref{eq:pure-YM})
can be reformulated to look more like CHY Pfaffians. 

For Yang-Mills, this can be done as follows: the full supergravity integrands
$\hat{\cI}^{L,R}_0=\frac{1}{\sigma_{\ell^+\,\ell^-}}\cI^{L,R}_0$ can be expressed more compactly in terms of a single NS sector matrix $M_{NS}$,  defined explicitly below, as
\begin{equation}\label{eq:defInt_MNS}
 \hat{\cI}^{L,R}_0=\sum_r\pf'(M_{\text{NS}}^r)-\frac{c_d}{\sigma_{\ell^+\,\ell^-}^2}\pf(M_2)\,,
\end{equation}
where
$\pf'(M_{\text{NS}}^r)\equiv\frac{-1}{\sigma_{\ell^+\,\ell^-}}\pf({M_{\text{NS}}^r}_{(\ell^+\,\ell^-)})$, and the brackets $(\ell^+\,\ell^-)$ indicate that the rows and columns associated to $\ell^+$ and $\ell^-$ have been removed. In particular, this implies that
\begin{equation}
  \sum_r \pf'(M_{\text{NS}})=\pf(M_3)\,\big|_{\sqrt{q}}+(d-2)\pf(M_3)\,\big|_{q^0}\,.
\label{eq:PfNs-def}
\end{equation}
The matrix $M_{\text{NS}}^r$ is defined by
\begin{equation}\label{eq:defMNS}
 M_{\text{NS}}^r=\begin{pmatrix}A & -C^T\\ C & B \end{pmatrix}\,,
\end{equation}
and more specifically
\begin{equation}  \label{eq:defMNS_details}
\begingroup
\renewcommand*{\arraystretch}{1.4}
 M_{\text{NS}}^r=\left( \begin{array}{cccc:cccc}
    0 & 0 & \frac{\ell\cdot k_i}{\sigma_{\ell^+i}} & \frac{\ell\cdot k_j}{\sigma_{\ell^+j}} & -\epsilon^r\cdot P(\sigma_{\ell^+}) & \frac{\ell\cdot\epsilon^r}{\sigma_{\ell^+\ell^-}} & \frac{\ell\cdot\epsilon_i}{\sigma_{\ell^+i}} & \frac{\ell\cdot\epsilon_j}{\sigma_{\ell^+j}}\\
    & 0 & -\frac{\ell\cdot k_i}{\sigma_{\ell^-i}} & -\frac{\ell\cdot k_j}{\sigma_{\ell^-j}} & -\frac{\ell\cdot\epsilon^r}{\sigma_{\ell^-\ell^+}} & -\epsilon^r\cdot P(\sigma_{\ell^-}) & -\frac{\ell\cdot\epsilon_i}{\sigma_{\ell^-i}} & -\frac{\ell\cdot\epsilon_j}{\sigma_{\ell^-j}}\\
    &   & 0 & \frac{k_i\cdot k_j}{\sigma_{ij}} & \frac{k_i\cdot\epsilon^r}{\sigma_{\ell^+i}} & \frac{k_i\cdot\epsilon^r}{\sigma_{\ell^-i}} & -\epsilon_i\cdot P(\sigma_{i}) & \frac{k_i\cdot\epsilon_j}{\sigma_{ij}}\\
    &   &   & 0 & \frac{\epsilon^r\cdot k_j}{\sigma_{\ell^+j}} & \frac{k_j\cdot\epsilon^r}{\sigma_{\ell^-j}} & \frac{k_j\cdot\epsilon_i}{\sigma_{ji}} & -\epsilon_j\cdot P(\sigma_j)\\
    \hdashline
    & & & & 0 & \frac{d-2}{\sigma_{\ell^+\ell^-}} & \frac{ \epsilon^r\cdot\epsilon_i}{\sigma_{\ell^+i}} & \frac{ \epsilon^r\cdot\epsilon_j}{\sigma_{\ell^+j}}\\
    & & & &   & 0 & -\frac{\epsilon^r\cdot\epsilon_i}{\sigma_{\ell^-i}} & \frac{\epsilon^r\cdot\epsilon_j}{\sigma_{\ell^-j}}\\
    & & & &   &   & 0 & \frac{\epsilon_i\cdot \epsilon_j}{\sigma_{ij}}\\
    & & & &   &   &   & 0\\
  \end{array} \right)\,.
  \endgroup
\end{equation}
The sum runs over a basis of polarisation vectors $\epsilon^r$, and
$d$ denotes the space-time dimension. Note in particular that the reduced
Pfaffian is well-defined since this matrix has indeed co-rank
two. Similar to the structure at tree-level, the vectors $(1,\dots,1,0,\dots,0)$
and $(\sigma_{\ell^+},\sigma_{\ell^-},\sigma_1,\dots,\sigma_n,0,\dots,0)$ span the
kernel of the matrix $M$ on the support of the scattering
equations.

The proof of eq.~\eqref{eq:PfNs-def}, relies on standard properties of
Pfaffians, and the interested reader is referred to \cref{sec:MNS}. In
this form, the NS contribution to the integrand is very suggestive of
a worldsheet CFT correlator, and indeed it is not hard to see that
this Pfaffian arises form an off-shell sphere correlator with two
points whose polarizations should be replaced by a photon propagator
in a physical gauge.

The gravity case uses also $M_{NS}$, and is treated in more
details later in sec.~\ref{sec:non-supersymm-theor-1}, when we discuss
the factorisation properties of these pure Yang-Mills and gravity
amplitudes. Basically, we simply decompose the difference of squares
in eq.~\eqref{eq:pure-gravity} as a product.

To conclude this discussion, we note that the fermion contribution of
eq.~\eqref{eq:fermion} for a two-fermion-$n$-graviton integrand seems to
arise naturally as a factorised product of Pfaffians.
Although amplitudes with fermions have been computed in
\cite{Adamo:2013tsa}, no Pfaffian-like form for higher-point
amplitudes is known, partly because of the non-polynomial nature of
the spin-field OPE's that prevents the naive re-summing of the correlators. The $n$-point amplitude is known,
however~\cite{Mafra:2011nv}; using the procedure
of~\cite{Gomez:2013wza} it is possible to cherry-pick the
2-fermion-$n$-boson component of the string amplitude
(using~\cite{Mafra:2015vca} for instance). It would be interesting to
see if a Pfaffian arises in doing so.
It is possible that eq.~\eqref{eq:fermion} is different from the
generic -- i.e.~non-forward -- amplitude due to terms vanishing with
$\ell^{2}$. Nevertheless, this hints at some unexpected simplicity.

\subsection{Checks on all-plus amplitudes}
\label{sec:gso-projection}

We have presented proposals for the integrands of four-dimensional $n$-particle amplitudes in non-supersymmetric gauge theory and gravity. In the gravity case, we distinguished between the cases of pure gravity and the theory consisting of the NS-NS sector of supergravity, namely the theory with a graviton, a dilaton and a B-field. While we focused on four dimensions for the sake of being explicit, it is clear that analogous constructions can be made of theories with different matter couplings in various dimensions, including different degrees of supersymmetry if we also introduce fermions.

We checked our conjectures against know expressions for the simplest class of non-supersymmetric four-dimensional amplitudes. These are the amplitudes for which all external legs have the same helicity, which we will choose to be positive. The supersymmetric Ward identities \cite{Grisaru:1979re} lead to the following relations for these non-supersymmetric amplitudes:
\begin{equation}
\begin{aligned}
\cM^{(1)}_{\text{pure-YM}}(\text{all-plus}) =&\, 2\, \cM^{(1,\,\text{scalar})}_{\text{pure-YM}}(\text{all-plus}) \\
\cM^{(1)}_{\text{NS-NS-grav}}(\text{all-plus}) =& \,2\, \cM^{(1)}_{\text{pure-grav}}(\text{all-plus}) =4\, \cM^{(1,\,\text{scalar})}_{\text{pure-grav}}(\text{all-plus}) \ .
\label{allplusscalars}
\end{aligned}
\end{equation}
The superscript on the right-hand side indicates an amplitude where
only one real minimally-coupled scalar is running in the loop. For
gauge theory and for pure gravity, the two helicity states running in
the loop are effectively equivalent to two real scalars, hence the
factor of two, whereas for NS-NS gravity there are two extra states
(dilaton and axion), leading to four real scalars. We checked at four and five points that
\begin{equation}
\pf(M_{3})\big|_{\sqrt{q}}(\text{all-plus})=0 \ .
\label{derallpluszero}
\end{equation}
From this simple fact, it is easy to see that our conjectured expressions satisfy the relations \eqref{allplusscalars}. We believe this to hold for any multiplicity. These statements also apply to amplitudes with one helicity distinct from all others (say one minus, rest plus), which also satisfy the relations \eqref{allplusscalars}.

We have explicitly checked that our conjectures for pure gauge theory and gravity match the (shifted) integrands for all-plus amplitudes in the case of $n=4$. For concreteness, we will write down the integrands explicitly. The Feynman rules for the all-plus amplitudes take a particularly simple form in light-cone gauge, because such amplitudes correspond to the self-dual sector of the theory \cite{Chalmers:1996rq,Cangemi:1996rx}. The rules for the vertices and external factors in all-plus amplitudes in gauge theory can be taken to be \cite{Boels:2013bi}
\begin{equation}
(i^+,j^+,k^-)=X_{i,j} \, f^{a_i a_j a_k}, \qquad e_i^{(+)}=\frac{1}{\langle \eta i\rangle^2},
\label{extstates}
\end{equation}
whereas in gravity they are
\begin{equation}
(i^{++},j^{++},k^{--})=X_{i,j}^2, \qquad e_i^{(++)}=\frac{1}{\langle \eta i\rangle^4} \ .
\end{equation}
We are again making use of the spinor helicity formalism, and taking $\eta=|\eta\rangle [\eta|$ to be a reference vector. Gauge invariance implies that the amplitudes are independent of the choice of $\eta$. The object $X_{i,j}$ is defined with respect to the spinors $|\hat{i}]=K_i |\eta\rangle$, which can be defined for any (generically off-shell) momentum $K_i$ using the reference spinor $|\eta\rangle$,
\begin{equation}
X_{i,j}=-[\hat{i}\hat{j}] = - X_{j,i}\ .
\end{equation}
The direct ``square'' relation between the rules in gauge theory and in gravity makes the BCJ double copy manifest for these amplitudes \cite{Monteiro:2011pc}.

Using the diagrammatic rules above, we can write the (shifted) integrand for the single-trace contribution to gluon scattering as
\begin{align}
\hat{\cM}^{(1,\,\text{scalar})}_{\text{pure-YM}}(1^+2^+3^+4^+) =  \frac1{ \ell^2 \prod_{i=1}^n \langle \eta i\rangle^2}  \sum_{\rho\in \text{cyc}(1234)} \left(I^{\text{box-YM}}_{\rho_1\rho_2\rho_3\rho_4} +I^{\text{tri-YM}}_{[\rho_1,\rho_2]\rho_3\rho_4} + \frac1{2}\,I^{\text{bub-YM}}_{[\rho_1,\rho_2][\rho_3,\rho_4]} \right)  \ ,
\end{align}
with
\begin{align}
I^{\text{box-YM}}_{1234} =&  \frac{X_{\ell,1}X_{\ell+1,2}X_{\ell-4,3}X_{\ell,4}}{(2\ell\cdot k_1)(2\ell\cdot (k_1+k_2)+2k_1\cdot k_2)(-2\ell\cdot k_4)} 
\ , \nonumber \\
I^{\text{tri-YM}}_{[1,2]34} =& \frac{X_{1,2}}{2k_1\cdot k_2} \Bigg(
 \frac{X_{\ell,1+2}X_{\ell,3}X_{\ell+3,4}}{(2\ell\cdot k_3)(2\ell\cdot (k_3+k_4)+2k_3\cdot k_4)}
\nonumber \\
& \quad + \frac{X_{\ell,1+2}X_{\ell-4,3}X_{\ell,4}}{(-2\ell\cdot (k_3+k_4)+2k_3\cdot k_4)(-2\ell\cdot k_4)}   + \frac{X_{\ell+4,1+2}X_{\ell,3}X_{\ell,4}}{(-2\ell\cdot k_3)(2\ell\cdot k_4)}
 \Bigg) \ ,\nonumber \\
I^{\text{bub-YM}}_{[1,2][3,4]} =& \frac{X_{1,2}X_{3,4}X_{\ell,1+2}X_{\ell,3+4}}{(2k_1\cdot k_2)^2} 
\Bigg( \frac1{2\ell\cdot (k_1+k_2)+2k_1\cdot k_2} + \frac1{-2\ell\cdot (k_1+k_2)+2k_1\cdot k_2}  \Bigg)\nonumber\ .
\end{align}
Notice that there is no contribution from external-leg bubbles. As discussed in \cite{He:2015yua} in the context of the biadjoint scalar theory, such contribution must be proportional to the tree-level amplitude, which vanishes for the all-plus helicity sector. We should mention that the `singular' solutions of the scattering equations, for which $\sigma_{\ell^+}=\sigma_{\ell^-}$, give a directly vanishing contribution to the all-plus loop integrand.

For the scattering of gravitons, we have
\begin{align}
\hat{\cM}^{(1,\,\text{scalar})}_{\text{pure-grav}}(1^+2^+3^+4^+)=  \frac1{ \ell^2 \prod_{i=1}^n \langle \eta i\rangle^4} \sum_{\rho\in S_4}
\left(I^{\text{box-grav}}_{\rho_1\rho_2\rho_3\rho_4} + \frac1{2}\, I^{\text{tri-grav}}_{[\rho_1,\rho_2]\rho_3\rho_4} + \frac1{4}\,I^{\text{bub-grav}}_{[\rho_1,\rho_2][\rho_3,\rho_4]} \right) \ ,
\end{align}
where $I^{\text{box-grav}}$, $I^{\text{tri-grav}}$ and $I^{\text{bub-grav}}$ are respectively obtained from $I^{\text{box-YM}}$, $I^{\text{tri-YM}}$ and $I^{\text{bub-YM}}$ through the substitution $X_{.,.}\to (X_{.,.})^2$ in the numerators.

\section{Proof for non-supersymmetric amplitudes at one-loop}
\label{sec:factorization}

We now give a full proof of the formulae for one-loop amplitudes derived above for non-supersymmetric theories, i.e., the $n$-gons, biadjoint scalar theory, Yang-Mills and gravity.\footnote{In particular, the proof holds for both the NS sector (including the B-field and the dilaton for gravity) and the pure theories.}
There are three main ingredients in our proof.  The first is to identify the poles in our formulae arising from factorisation or bubbling of the Riemann sphere, which allows us to determine the location of the poles and their residues. Since the $1/\ell^2$ is already apparent, this analysis of factorisation will lead to the identification of the residue at two poles. The second is the theory of `Q-cuts' introduced in \cite{Baadsgaard:2015twa} that expresses  a general one-loop amplitude in terms of tree amplitudes that is perfectly adapted to the factorisation of our formulae (this is perhaps not completely surprising as their construction was motivated by our formulae).  The third is the discussion of \cite{Cachazo:2015aol} on the spurious poles of our formulae. The terms in the integrand that possess these poles  scale homogeneously with the loop momentum, and are therefore discarded in dimensional regularization. These poles are also explicitly discarded in the  `Q-cut' procedure and so will not contribute to the Q-cut decomposition of our formulae.We will therefore not discuss these spurious poles in any detail, and will concentrate instead on the physical poles that generate  the `Q-cuts'.

It is standard that an amplitude must factorise in the sense that if a partial sum of the external momenta  $k_I=\sum_{i\in I}k_i$, where $I\subset \{1,\ldots, n\}$, becomes null, then there will be a pole corresponding to a propagator that has $k_I$ flowing through it.  Furthermore, the residue is the product of two tree amplitudes for the theory in question with external legs consisting of $\pm k_I$ and the elements of $I$ or its complement $\bar I$.   A well known property of the scattering equations \cite{Cachazo:2013hca} is that factorisation of the momenta corresponds precisely to factorisation of the Riemann surface, i.e., the concentration at a point of the points corresponding to the partial sum. This concentration point  can then be blown up to give a bubbled-off Riemann sphere connected to the original at the concentration point, see  below (or \cite{Dolan:2013isa}).

Our scattering equations at 1-loop will give worldsheet factorisation channels that lead to poles associated to loop momenta, but these are not immediately recognizable as loop propagators; they instead correspond to poles of the form of those in the sum of \eqref{formulangon}.  These however can be understood as naturally arising in the `Q-cuts' of \cite{Baadsgaard:2015twa}.  These are a systematic extension of the contour integral argument that leads to the partial fractions expansion of \eqref{eq:Di-partfrac} applicable to any 1-loop integrand.  They follow from a two-step process.  The first follows the contour integral argument of \eqref{eq:Di-partfracz}.  Consider a one-loop integrand 
\begin{equation}
\cM(\ell,k_i,\epsilon_i)=\frac{N(\ell,\ldots)}{D_{I_1}\ldots D_{I_m}}\,,
\end{equation}
for the theory under consideration, where $N$ is a polynomial numerator, and $D_I=(\ell+k_I)^2$ a propagator.  We shift the loop momentum $\ell\rightarrow \ell + \eta$ where $\eta$ is in some higher dimension than the physical momenta and polarization vectors, so that the only shift that occurs in the invariants is $\ell^2\rightarrow \ell^2+z$, with all other inner products remaining unchanged.  One then runs the contour integral argument that expresses the amplitude as the residue at $z=0$ of $\cM(\ell+\eta,k_i,\epsilon_i)/z$ in terms of  the sum of the other residues of this expression.  Such residues arise at shifted propagators $1/(D_I+z)$ with poles at $z=-D_I$. One then shifts $\ell\rightarrow \ell-k_I$ in each of these new residues so that $z$ becomes $\ell^2$.  This gives a representation of a 1-loop amplitude as a sum of terms of the form
\begin{equation}
\frac{1}{\ell^2}\left[
\frac{\tilde 
N(\ell)}{(2\ell\cdot k_{I_1} +k_{I_1}^2)
\ldots (2\ell\cdot k_{I_m}+ k_{I_m}^2)
}
\right]\,,
\end{equation}
giving a generalization of the partial fraction formulae of \eqref{eq:Di-partfrac}.  

In order  to interpret constituents of this expression as tree amplitudes, \cite{Baadsgaard:2015twa} considers a further contour integral argument with integrand
\begin{equation}
\frac{\cM(\alpha \ell)}{\alpha-1}\,,
\end{equation}
where $\cM(\ell)$ is now the expression with shifted $\ell$s obtained above.  The residue at $\alpha=1$ returns the original $\cM(\ell)$.  The residues  at zero and infinity can be discarded as they vanish in dimensional regularization.
It can then be argued \cite{Baadsgaard:2015twa} that the finite residues finally yield the `Q-cut' decomposition
\begin{equation}
\label{Qcuts}
 \cM\Big|_{\text{Q-cut}}\equiv\sum_{I}\cM_I^{(0)}(\dots,\tilde \ell_I, \tilde \ell_I+k_I)\frac{1}{\ell^2 \,(2\ell\cdot k_I+k_I^2)}\cM_{\bar I}^{(0)}(-\tilde \ell_I,-\tilde \ell_I-k_I,\dots)\,,
\end{equation}
where $\tilde \ell_I=\alpha(\ell +\eta)$, with $\alpha=-k_I^2/2\ell\cdot k_I$, $\eta^2=-\ell^2$, $\eta\cdot\ell=\eta\cdot k_i=0$, and $\cM_I^0$ and $\cM_{\bar I}^0$ are now tree amplitudes.

\begin{figure}[ht]
\begin{center}
 \begin{tikzpicture}[scale=3]
 \draw (1.75,1.2) to [out=30, in=180] (2.125,1.35) to [out=0, in=150] (2.5,1.2);
 \draw (1.75,0.8) to [out=330, in=180] (2.125,0.65) to [out=0, in=210] (2.5,0.8);
 \draw (1.75,1.2) -- (1.375,1.4);
 \draw (1.75,0.8) -- (1.375,0.6);
 \draw (2.5,1.2) -- (2.875,1.4);
 \draw (2.5,0.8) -- (2.875,0.6);
 \draw (1.75,1.07) -- (1.3,1.14);
 \draw (1.75,0.93) -- (1.3,0.86);
 \draw (2.5,1.07) -- (2.95,1.14);
 \draw (2.5,0.93) -- (2.95,0.86);
 \draw [fill, light-grayII] (1.75,1) circle [x radius=0.2, y radius=0.4];
 \draw (1.75,1) circle [x radius=0.2, y radius=0.4];
 \draw [fill, light-grayII] (2.5,1) circle [x radius=0.2, y radius=0.4];
 \draw (2.5,1) circle [x radius=0.2, y radius=0.4];
 \node at (0.72,1) {$\displaystyle \mathcal{M}^{(1)}\Big|_{\text{Q-cut}}= \,\sum_I$};
 \node at (2.125,1.5) {$\scriptstyle\ell$};
 \node at (2.125,0.5) {$\scriptstyle \ell+K_I$};
 \node at (1.75,1) {$\scriptstyle I$};
 \node at (2.5,1) {$\scriptstyle \overline{I}$};
 \node at (3.7,1) {$\displaystyle =\sum_I\frac{\cM_I^{(0)} \cM^{(0)}_{\bar I}}{\ell^2(2\ell\cdot K_I +K_I^2)}$};
\end{tikzpicture}
\end{center}
\caption{Representation of the amplitude as a sum over Q-cuts.}
\label{fig:q-cut}
\end{figure}  

We will see in this section that factorisation of the worldsheet in our formulae gives precisely these poles and residues. It is not exactly equivalent to the Q-cuts, because we find $\alpha=1$ in our case {\it off the pole}, rather than $\alpha=-k_I^2/2\ell\cdot k_I$, as in \eqref{Qcuts}. However, on any pole of \eqref{Qcuts}, our formulae reproduce the same residue as the Q-cut formulae. The difference between the loop integrands we obtain and those given by the Q-cut representation will be hoomogeneous in the loop momenta and vanish upon loop integration. Apart from poles and residues, the other piece of evidence required to prove our result (for which $\alpha=1$) is the UV behaviour of the loop integrand, which we also determine from a factorisation argument. These results on poles, residues and UV behaviour combine to prove that our formulae have the same Q-cut decompositions as one loop integrands and so give the correct amplitude under integration as desired.

\begin{thm}{}\label{thm:fact}
Consider a subset $I\subset\{1,\ldots, n, \ell^+,\ell^-\}$ with one fixed point in $I$ and two in its complement $\bar I$ and suppose  that we have solutions to the scattering equations such that $\sigma_i= \sigma_I + \varepsilon x_i+\cO(\varepsilon^2)$ for $i\in I$ with $x_i=O(1)$, $\sigma_{ij}=O(1)$ for $j\in\bar I$  and $\varepsilon \rightarrow 0$. 
Then we must also have  $\tilde{s}_I =O(\varepsilon)$ where 
 \begin{equation}
  \tilde{s}_I=\begin{cases}
               k_I^2 & \sigma_{\ell^+},\sigma_{\ell^-}\in\bar I\,,\\
               2\ell\cdot k_I + k_I^2 & \sigma_{\ell^+}\in I, \sigma_{\ell^-}\in \bar I\,.
              \end{cases}
 \end{equation}
 
 Our 1-loop formulae $\cM^{(1)}$ on the Riemann sphere for n-gons, bi-adjoint scalar theory, Yang-Mills and gravity have poles at $\tilde{s}_I= 0$ with residue for $\sigma_{\ell^+},\sigma_{\ell^-}\in\bar I$ is given by the separating degeneration of the nodal Riemann sphere (case I of figure \ref{fig:degen}) 
 \begin{equation}
  \cM^{(1)}(\dots,-\ell,\ell)=\cM_I^{(0)}(\dots)\frac{1}{s_I}\cM_{\bar I}^{(1)}(\dots,-\ell,\ell) \,,
 \end{equation}
 and for $\sigma_{\ell^+}\in I, \sigma_{\ell^-}\in \bar I$ by the Q-cut degeneration
 \begin{equation}\label{eq:Q-cut fact}
  \cM^{(1)}(\dots,-\ell,\ell)=\cM_I^{(0)}(\dots,\ell_I,\ell_I+k_I)\frac{1}{\ell^2}\frac{1}{\tilde{s}_I}\cM_{\bar I}^{(0)}(-\ell_I,-\ell_I-k_I,\dots)\,,
 \end{equation}
 where $\ell_I=\ell+\eta$, with $\eta^2=-\ell^2$, $\eta\cdot\ell=\eta\cdot k_i=0$ which is case II of figure \ref{fig:degen}.
\end{thm}

\begin{proof}
 We restrict ourselves here to outlining the idea of the proof, all details will be developed in \cref{sec:fact_SE}. The central observation is that poles in \eqref{1-loop} occur only if a subset $I$ of the marked points approach the same marked point $\sigma_I$; so that $\sigma_i \rightarrow \sigma_I+\varepsilon x_i +\cO(\varepsilon^2)$ for $i\in I$.
This is conformally equivalent to a degeneration of Riemann sphere into two components, connected by a double point. All such poles receive contributions from both the measure and scattering equations. In particular, whether a pole occurs for a given integrand reduces to a simple scaling argument in the degeneration parameter $\varepsilon$, and we can straightforwardly identify the residues.

To be more explicit, for some $m\in I$ fix $\sigma_m= \varepsilon$ so that $x_m=1$ is the new fixed point on the $I$ component of the degenerate Riemann surface, then the measure and the scattering equations factorise as
 \begin{equation}
  d \mu\equiv\prod_{i=2}^n \bar{\delta}(k_i\cdot P(\sigma_i))d\sigma_i=\varepsilon^{2(|\cI|-1)}
  \frac{d\varepsilon}{\varepsilon}  \,\bar{\delta}(s_I+\varepsilon\mathcal{F})\,d\mu_I\, d\mu_{\bar I}\,.
\end{equation}
Moreover, \cref{sec:fact_int} provides details on how the integrands for $n$-gons, Yang-Mills theory and gravity factorise as well;
\begin{equation}
 \cI^{(1)}=\varepsilon^{-2(|\cI|-1)}\cI_I^{(0)}\cI_{\bar I}^{(0)}\,,
\end{equation}
where $\cI_{I,\bar I}^{(0)}$ depend only on on-shell momenta $k_i$ and the loop momentum in the on-shell combination $\ell_I=\ell+\eta$, $\eta^2=-\ell^2$. The full amplitude therefore factorises on the expected poles, and the residue gives the Q-cut factorisation described above;
\begin{equation}
 \cM^{(1)}=\int \cI^{(1)} d\mu=\int  \frac{d\varepsilon}{\varepsilon} \,\bar{\delta}(\tilde{s}_I+\varepsilon\mathcal{F})\,\cI_I^{(0)}\cI_{\bar I}\,d\mu_I^{(0)}\, d\mu_{\bar I}=\cM_I^{(0)}\frac{1}{\tilde{s}_I}\cM_{\bar I}^{(0)}\,.
\end{equation}
\end{proof}

\begin{thm}{}\label{lemma:powercounting}
 The amplitudes $\cM^{(1)}$ scale as $\ell^{-N}$ for $\ell\rightarrow\infty$, where
 \begin{center}
\begin{tabular}{l|c}
 theory & scaling $\ell^{-N}$\\ \hline 
 $n$-gon & $N=2n$\\
 supergravity & $N=8$\\
 super Yang-Mills & $N=6$\\
 pure gravity & $N=4$\\
 pure Yang-Mills & $N=4$\\
 bi-adjoint scalar & $N=4$
\end{tabular}
\end{center}
\end{thm}

\begin{proof} 
  This  follows from the fact that as $\ell\rightarrow \infty$,  the insertions of $\sigma_{\ell^+}$ and $\sigma_{\ell^-}$ must approach each other. This is conformally equivalent to a degeneration of the worldsheet into a nodal Riemann sphere with no further insertions and another Riemann sphere carrying all the external particles, see case III in figure \ref{fig:degen}. This is also the configuration that corresponds to the singular/degenerate solutions described in the previous sections and so our analysis of fall-off in $\ell$ will also give information about the finiteness of the contributions from these degenerate solutions. 
 We  give the full details in \cref{sec:powercounting}.\end{proof}

The UV behaviour found here is not optimal for maximal super-Yang-Mills theory.  At one-loop, amplitudes in this theory are well known to be a sum of boxes, and these scale manifestly as $\ell^{-8}$ for large $\ell$.  However, once a box has been subject to partial fractions and shifts, the individual terms scale only as $\ell^{-5}$.  Were there no shifts, the partial fraction sum would nevertheless have to fall off as $\ell^{-8}$, because this is then simply a different expression for the box. However, with shifts, the $\ell$ in different terms is shifted by different amounts and the cancellations between the terms are affected, leading to a fall off as $\ell^{-6}$, weaker than the $\ell^{-8}$ exhibited by the unshifted sum. These two types of representations for the integrand -- unshifted and shifted -- differ manifestly by a quantity that does not contribute to the (integrated) amplitude. In view of this, it is perhaps a surprise that, in the maximal supergravity case, our formula does show the optimal fall-off. This is because the full permutation sum in supergravity introduces all the needed cancellations of the error terms caused by the shifts, whereas the cyclic sum in super-Yang-Mills is not enough. It would be interesting to understand more systematically how to recover the optimal fall-off in our approach, including the subtle cancellations in lower supersymmetry studied in \cite{Bern:2007xj}.

\Cref{thm:fact} and \cref{lemma:powercounting} will now allow us to prove that the representation of one-loop amplitudes from the nodal Riemann sphere is equivalent to the Q-cut representation reviewed above.

\begin{thm}{}
 $\cM^{(1)}_{\mathrm{BS}}$, $\cM^{(1)}_{\mathrm{YM}}$ and $\cM^{(1)}_{\mathrm{gravity}}$ with the degenerate solutions ommitted are representations of the one-loop amplitudes for the bi-adjoint scalar theory, Yang-Mills and gravity respectively.
\end{thm}
\begin{proof}\footnote{Reference \cite{Cachazo:2015aol}, which appeared after the first version of this paper, demonstrates the existence of spurious poles in our formulae that were missed in the first version.   However, they also showed that these spurious poles are homogeneous in the loop momenta, and hence do not contribute to the Q-cut.  Thus the reference serves to complete our proof.}
We use the fact that our formula must be rational in the external data and $\ell$, and that the only poles in our formulae arise when the punctures come together, i.e., factorisation as discussed above.  This theorem  is then a consequence of the correct factorisation on Q-cuts and the scaling behaviour in $\ell$. Consider first the quantity
 \begin{equation}
    \Delta=\cM^{(1)}-\cM\Big|_{\text{Q-cut}}^{\alpha=1}\,,
 \end{equation}
where the last term is given by \eqref{Qcuts} with $\alpha=1$ in the scaling of $\tilde \ell$, which is not the case for the Q-cut representation. The quantity $\Delta$ possesses two classes of poles in $\ell$: the Q-cut poles from \cref{thm:fact} and the spurious poles analysed in \cite{Cachazo:2015aol}. 
The first class of poles are now given by \cref{thm:fact}, and so are cancelled by  the corresponding poles in $\cM\big|_{\text{Q-cut}}^{\alpha=1}$. The second class are homogeneous in the loop momenta and so 
%
 \begin{equation}
    \cM^{(1)}=\cM\Big|_{\text{Q-cut}}^{\alpha=1}\,+ \text{terms vanishing upon integration}.
 \end{equation}
Finally, the difference between $\cM\big|_{\text{Q-cut}}^{\alpha=1}$ and the Q-cut representation also vanishes upon loop integration, as follows from the contour integral in $\alpha$ in the construction of the Q-cuts. Therefore, $\cM^{(1)}$ is a valid representation of the loop integrand. Notice also that the degenerate solutions to the scattering equations do not contribute to the  singularities that give rise to the Q-cuts, since the latter arise from case II of figure \ref{fig:degen}, whereas the degenerate solutions are all case III of figure \ref{fig:degen}. Like the terms with spurious poles, the contributions from degenerate solutions vanish upon loop integration. Hence, we can define $\cM^{(1)}$ by summing over the regular solutions only.
\end{proof}

\subsection{Factorisation I - Scattering equations and measure}\label{sec:fact_SE}
As discussed above, poles of $\cM^{(1)}$ only occur  when a subset of the marked points
(possibly including $\sigma_{\ell^+}$ or $\sigma_{\ell^-}$) approach the same point, giving rise to a degeneration of the Riemann sphere into a pair of Riemann spheres connected at a double point. The scattering equations then imply that this pole is associated with a partial sum of the momenta becoming null.

Let $I$ be a subset of $\{\ell^+, 1,\ldots ,n,\ell^-\}$ that contains just one of the fixed points which we shall, with an abuse of notation, denote $\sigma_I$.   We shall always, however, write $k_I= \sum_{i\in I}k_i$, $k_{\bar I} =\sum_{i\in \bar I} k_i$ with $k_0=\ell_I$, $k_{n+1}=-\ell_I$.   Let the Riemann surface factorise\footnote{  Two of the three original fixed points must be in $\bar I$ and just one in $I$ to obtain a stable degeneration  as we cannot make two of the fixed points  approach each other but $\bar I$ cannot contain three as, after factorisation, it will also have the fixed point $\sigma_I$ which would be too many.} so that for $i\in I$ 
\begin{equation}
\sigma_i \rightarrow \sigma_I+\varepsilon x_i +\cO(\varepsilon^2)
\end{equation} 
for some small $\varepsilon$, with $x_I=0$, $x_m=1$ for some $m\in I$ and $x_i =O(1)$ for all other $i\in I$. Note that this implies that $\sigma_m=\epsilon$ is now also our small parameter.

 We first wish to see that with these assumptions, the scattering equations imply that $\tilde s_I =O(\epsilon)$.  Firstly we have
\begin{equation}
\begin{aligned}
 &P(\sigma)=\frac{1}{\varepsilon}P_I(x)+\tilde{P}_{\bar I}(\sigma_I)+\cO(\varepsilon) && \sigma=\sigma_I+\varepsilon x+\cO(\varepsilon^2)\\
 &P(\sigma)=P_{\bar I}(\sigma)+\cO(\varepsilon) && \sigma\neq\sigma_I+\varepsilon x+\cO(\varepsilon^2)
\end{aligned}
\end{equation}
where 
\begin{equation}
P_I(x)=\sum_{i\in I}\frac{k_i}{x-x_i} 
\, , \qquad \tilde{P}_{\bar I}(\sigma)=\sum_{i\in \bar I}\frac{k_i}{\sigma-\sigma_i}\,,\qquad P_{\bar I}(\sigma)=\sum_{p\in \bar I}\frac{k_p}{\sigma-\sigma_p}-\frac{k_I}{\sigma-\sigma_I}\,.
\end{equation}
Thus the scattering equations give\footnote{Setting $k_0=\ell_I$, $k_{n+1}=-\ell_I$, the scattering equations for the marked points $\sigma_{\ell^\pm}$ are
\begin{equation*}
 0=\mathrm{Res}_{\sigma_{\ell^\pm}}=\pm\sum_i \frac{\ell_I\cdot k_i}{\sigma_{\ell^\pm i}}=\pm\ell_I\cdot P(\sigma_{\ell^\pm})\,,
\end{equation*}
so the conclusions hold for $\sigma_{\ell^\pm}$ also.}
\begin{subequations}\label{eq:SE-degen}
    \begin{align}
      & 0=k_i\cdot P(\sigma_i)=\frac{1}{\varepsilon} k_i\cdot P_I( x_i) + k_i\cdot \tilde{P}_{\bar I}(\sigma_I) + O(\varepsilon)\, , && i\in I\\
      & 0=k_p\cdot P(\sigma_p)=k_p\cdot P_{\bar I}(\sigma_p)\,, &&
      p\in \bar I\,.
    \end{align}
\end{subequations}
In particular, for $i\in I$, this implies
$k_i\cdot P_I(x_i)=O(\varepsilon)$ as $\varepsilon\rightarrow 0$ since
\begin{equation}
 \label{scale}
 k_i\cdot P_I( x_i) =-\varepsilon k_i\cdot \tilde{P}_{\bar I}(\sigma_I) + O(\varepsilon^2).
\end{equation}  
By summing we obtain as an algebraic identity
\begin{equation}
\label{factorize}
\tilde{s}_I:= \frac12 \sum_{i\neq j\in I} k_i\cdot k_j=  \sum_{i,j\in I} \frac{x_ik_i\cdot  k_j}{ x_i-x_j}= \sum_{i\in I} x_i k_i \cdot P_I(x_i)
 \, ,
\end{equation}
so $\tilde{s}_I=O(\varepsilon)$, and any (potential) pole is
associated with the vanishing of $\tilde s_I$ where 
\begin{equation*}
  \tilde{s}_I=\begin{cases}
               k_I^2 & \sigma_{\ell^+},\sigma_{\ell^-}\in\bar I\,,\\
               2\ell\cdot k_I + k_I^2 & \sigma_{\ell^+}\in I, \sigma_{\ell^-}\in \bar I\,.
              \end{cases}
 \end{equation*}.

We now focus on the measure of the amplitude expression with a generic integrand
\begin{equation}
\cM=\int \cI \prod_{i=2}^n \bar{\delta}(k_i\cdot P(\sigma_i))d\sigma_i\, .
\end{equation}
We first determine the weight of the measure in $\varepsilon$ as $\varepsilon\rightarrow 0$ (the integrand $\cI$ will have some weight also which we discuss later).
For each $i\in I$, the scattering equations contribute
\begin{subequations}
  \begin{align}
    &\bar{\delta}(k_i\cdot P(\sigma_i))\,d \sigma_i=\varepsilon^2\,\bar{\delta}(k_i\cdot P_I(x_i))\,d x_i && i\in I\\
    &\bar{\delta}(k_p\cdot P(\sigma_p))\,d \sigma_p=\bar{\delta}(k_p\cdot P_{\bar I}(\sigma_p))\,d \sigma_p && p\in \bar I
  \end{align}
\end{subequations}

Thus we obtain  scattering equations on the factorised Riemann surface, multiplied by a factor of $\varepsilon^2$ for each $i\in I$. Note however that there is a subtlety; we expect three fixed marked points on each Riemann surface. On the Riemann surface $\Sigma_{\bar I}$, this is trivially true since there are two fixed points and the degeneration point $\sigma_I$. On $\Sigma_I$, the fixed points are given by the degeneration point $x_{\bar I}=\infty$ and $x_I=0$, and our choice of parametrisation for the degeneration $x_m=1$.
This gives the required independence of exactly $n_I-3$ scattering equations, but we still have the integration over $\sigma_m=\varepsilon$ and its associated delta function imposing its scattering equation associated to $\sigma_m$.  Using 
\begin{equation}
 s_I=\sum_{i\neq m \in I}x_i k_i\cdot P_I(x_i) + k_m \cdot P_I(x_m)\,,
\end{equation}
and the support of the remaining scattering equations \cref{eq:SE-degen}, we find
\begin{align}
 k_m\cdot P(\sigma_m)&=\frac{1}{\varepsilon}\left(k_m\cdot P_I(x_m)+\varepsilon\sum_{p\in \bar I}k_m\cdot \tilde{P}_{\bar I}(\sigma_I)+\cO(\varepsilon^2)\right)\nonumber\\
 &=\frac{1}{\varepsilon}\left(\tilde{s}_I+\varepsilon\sum_{i\in I,\, p\in \bar I}x_i k_i\cdot \tilde{P}_{\bar I}(\sigma_I)+\cO(\varepsilon^2)\right)\nonumber\\
 &\equiv\frac{1}{\varepsilon}\left(\tilde{s}_I+\varepsilon\mathcal{F}+\cO(\varepsilon^2)\right)\,.
\end{align}
Thus 
\begin{equation}
 \bar{\delta}(k_m\cdot P(\sigma_m))\,d \sigma_m=\varepsilon\,\bar{\delta}(\tilde{s}_I+\varepsilon\mathcal{F})\,d \varepsilon  \,,
\end{equation}
and  the measure factorises as
 \begin{equation}
  d \mu\equiv\prod_{i=2}^n \bar{\delta}(k_i\cdot P(\sigma_i))d\sigma_i=\varepsilon^{2(|\cI|-1)}\,\frac{d\varepsilon}{\varepsilon} \,\bar{\delta}(\tilde{s}_I+\varepsilon\mathcal{F})\,d\mu_I\, d\mu_{\bar I}\,.
\end{equation}

We  now distinguish three cases according to whether $\sigma_{\ell^\pm}$ are  in $I$ as in  \cref{fig:degen}.
\begin{enumerate}[\text{Case} I]
 \item \label{Case1} If $\sigma_{\ell^+}$ and $\sigma_{\ell^-}$ are not in $I$, and $I$ is a strict subset of $1,\ldots ,n$,  this is standard factorisation with  $s_I=k_I^2\rightarrow 0$. This gives a Riemann sphere connected to a nodal sphere, corresponding to a tree-level amplitude factorising from a one-loop amplitude. The measure is
 \begin{equation}
  d \mu^{(1)}=\varepsilon^{2(|\cI|-1)}\,\frac{d\varepsilon}{\varepsilon} \,\bar{\delta}(s_I+\varepsilon\mathcal{F})\,d\mu_I^{(0)}\, d\mu_{\bar I}^{(1)}\,.
\end{equation}
 \item \label{Case2} If without loss of generality $\sigma_{\ell^+} \in I$ but $\sigma_{\ell^-}\notin I$ the condition is 
\begin{equation}
\tilde{s}_I=k_I\cdot\ell + \frac12 k_I^2=\cO(\varepsilon)\, .
\end{equation}
as $\varepsilon\rightarrow 0$. This non-separating degeneration describes two Riemann spheres, connected at {\it two} double points, see \cref{fig:degen}. The corresponding measure is given by
\begin{equation}
  d \mu^{(1)}=\varepsilon^{2(|\cI|-1)}\,\frac{d\varepsilon}{\varepsilon} \,\bar{\delta}(\tilde{s}_I+\varepsilon\mathcal{F})\,d\mu_I^{(0)}\, d\mu_{\bar I}^{(0)}\,,
\end{equation}
leading to the expected poles from the Q-cut factorisation.
 \item \label{Case3} The case $I=\{\sigma_{\ell^+},\sigma_{\ell^-}\}$ is of particular interest since this configuration arises for large $\ell$, $\ell =\cO(\lambda^{-1})$ (see \eqref{SE2n}), and for the singular solutions in non-supersymmetric theories. It is discussed in \cref{lemma:powercounting}, and determines the UV behaviour of our 1-loop amplitudes.
\end{enumerate}

\begin{figure}[ht]
\begin{center} \vspace{-5pt}\begin{tikzpicture}[scale=3]
 \shade[shading=ball, ball color=light-gray] (3.6,2.4) circle [x radius=0.4, y radius=0.2];
  \draw [dashed,red] (3.52,2.49) to [out=110, in=180] (3.6,2.7) to [out=0, in=70] (3.68,2.49);
  \draw [fill] (3.35,2.46) circle [radius=.3pt];
  \draw [fill] (3.45,2.31) circle [radius=.3pt];
  \draw [fill] (3.8,2.33) circle [radius=.3pt];
  \draw [fill] (3.85,2.45) circle [radius=.3pt];
  \draw [fill,red] (3.52,2.49) circle [radius=.3pt];
  \draw [fill,red] (3.68,2.49) circle [radius=.3pt];
  \node at (4.2,2.4) {$\rightarrow$};
  \node at (2.5,2.4) {Case I:};
 \shade[shading=ball, ball color=light-gray] (4.8,2.4) circle [x radius=0.4, y radius=0.2];
 \shade[shading=ball, ball color=light-gray] (5.35,2.4) circle [x radius=0.15, y radius=0.13];
  \draw [dashed,red] (4.72,2.49) to [out=110, in=180] (4.8,2.7) to [out=0, in=70] (4.88,2.49);
  \draw [fill] (4.65,2.31) circle [radius=.3pt];
  \draw [fill] (5,2.33) circle [radius=.3pt];
  \draw [fill] (5.44,2.455) circle [radius=.3pt];
  \draw [fill] (5.4,2.355) circle [radius=.3pt];
  \draw [fill,red] (4.72,2.49) circle [radius=.3pt];
  \draw [fill,red] (4.88,2.49) circle [radius=.3pt];
  \draw [fill] (5.2,2.4) circle [radius=.3pt];
  \shade[shading=ball, ball color=light-gray] (3.6,1) circle [x radius=0.4, y radius=0.2];
  \draw [dashed,red] (3.52,1.09) to [out=110, in=180] (3.6,1.3) to [out=0, in=70] (3.68,1.09);
  \draw [fill] (3.35,1.06) circle [radius=.3pt];
  \draw [fill] (3.45,0.91) circle [radius=.3pt];
  \draw [fill] (3.8,0.93) circle [radius=.3pt];
  \draw [fill] (3.85,1.05) circle [radius=.3pt];
  \draw [fill,red] (3.52,1.09) circle [radius=.3pt];
  \draw [fill,red] (3.68,1.09) circle [radius=.3pt];
  \node at (4.2,1) {$\rightarrow$};
  \node at (2.5,1) {Case III:};
 \shade[shading=ball, ball color=light-gray] (4.8,1) circle [x radius=0.4, y radius=0.2];
 \shade[shading=ball, ball color=light-gray] (5.35,1) circle [x radius=0.15, y radius=0.13];
  \draw [dashed,red] (5.4,1.055) to [out=20, in=90] (5.65,1.005) to [out=270, in=340] (5.4,0.955);
  \draw [fill] (4.55,1.06) circle [radius=.3pt];
  \draw [fill] (4.65,0.91) circle [radius=.3pt];
  \draw [fill] (5,0.93) circle [radius=.3pt];
  \draw [fill] (4.88,1.12) circle [radius=.3pt];
  \draw [fill,red] (5.4,1.055) circle [radius=.3pt];
  \draw [fill,red] (5.4,0.955) circle [radius=.3pt];
  \draw [fill] (5.2,1) circle [radius=.3pt];
 \shade[shading=ball, ball color=light-gray] (3.6,1.7) circle [x radius=0.4, y radius=0.2];
  \draw [dashed,red] (3.52,1.79) to [out=110, in=180] (3.6,2) to [out=0, in=70] (3.68,1.79);
  \draw [fill] (3.35,1.76) circle [radius=.3pt];
  \draw [fill] (3.45,1.61) circle [radius=.3pt];
  \draw [fill] (3.8,1.63) circle [radius=.3pt];
  \draw [fill] (3.85,1.75) circle [radius=.3pt];
  \draw [fill,red] (3.52,1.79) circle [radius=.3pt];
  \draw [fill,red] (3.68,1.79) circle [radius=.3pt];
  \node at (4.2,1.7) {$\rightarrow$};
  \node at (2.5,1.7) {Case II:};
 \shade[shading=ball, ball color=light-gray] (4.75,1.7) circle [x radius=0.3, y radius=0.15];
 \shade[shading=ball, ball color=light-gray] (5.35,1.7) circle [x radius=0.3, y radius=0.15];
  \draw [dashed,red] (4.82,1.79) to [out=60, in=180] (5.07,1.95) to [out=0, in=120] (5.32,1.79);
  \draw [fill] (4.6,1.75) circle [radius=.3pt];
  \draw [fill] (4.87,1.63) circle [radius=.3pt];
  \draw [fill] (5.3,1.62) circle [radius=.3pt];
  \draw [fill] (5.52,1.69) circle [radius=.3pt];
  \draw [fill,red] (4.82,1.79) circle [radius=.3pt];
  \draw [fill,red] (5.32,1.79) circle [radius=.3pt];
  \draw [fill] (5.05,1.7) circle [radius=.3pt];
\end{tikzpicture}\end{center}
\caption{Different possible worldsheet degenerations.}
\label{fig:degen}
\end{figure}
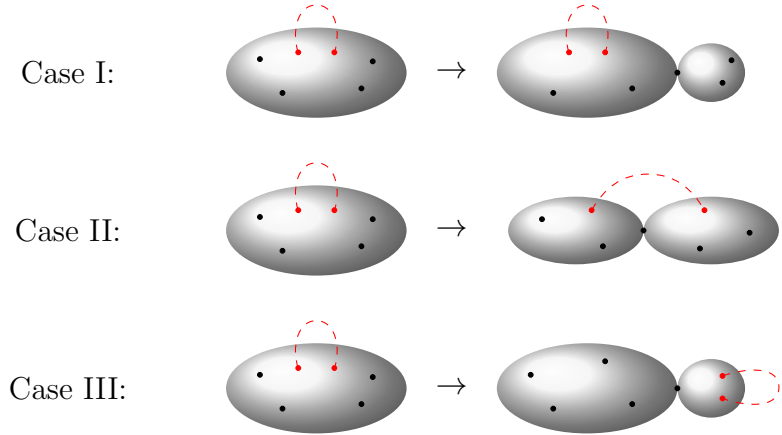

\subsection{Factorisation II - Integrands} \label{sec:fact_int}
Whether we actually have a pole or not in the factorization limit depends on the scaling behaviour of the integrand $\cI$ as $\varepsilon\rightarrow 0$.
In this section, we consider  the integrands  for the $n$-gons, Yang-Mills, gravity and the biadjoint scalar in more detail. In particular, we will find that all these integrands behave as
\begin{equation}
 \cI^{(1)}=\varepsilon^{-2(|\cI|-1)}\cI_I^{(0)}\cI_{\bar I}^{(1)}\,,
\end{equation}
in case I and 
\begin{equation} \label{eq:fact_intcase2}
 \cI^{(1)}=\frac{1}{\ell^2}\varepsilon^{-2(|\cI|-1)}\cI_I^{(0)}\cI_{\bar I}^{(0)}\,,
\end{equation}
in case II, where $\cI_I$ ($\cI_{\bar I}$) depends only on the on-shell momenta $k_I$ ($k_{\bar I}$) and $\ell_I=\ell+\eta$, $\eta^2=-\ell^2$. With the measure contributing a factor of $\varepsilon^{2(|\cI|-1)-1} \,\bar{\delta}(\tilde{s}_I+\varepsilon\mathcal{F})\,d\varepsilon$, the overall amplitude scales as $\varepsilon^{-1}$, and we can perform the integral against the $\delta$-function explicitly, leading to a pole in $\tilde{s}_I$. Therefore, the full amplitude factorises on the expected poles, with residues given by the corresponding subamplitudes. Moreover, as evident from above,  for case II this factorisation channel corresponds to a Q-cut;
\begin{equation}
\begin{split}
 \cM^{(1)}&=\int \cI^{(1)} d\mu=\frac{1}{\ell^2}\int  \varepsilon^{-1}d\varepsilon \,\bar{\delta}(\tilde{s}_I+\varepsilon\mathcal{F})\,\cI_I^{(0)}\cI_{\bar I}^{(0)}\,d\mu_I\, d\mu_{\bar I}\\
 &=\cM_I^{(0)}(\dots,\ell_I,\ell_I+k_I)\frac{1}{\ell^2}\frac{1}{\tilde{s}_I}\cM_{\bar I}^{(0)}(-\ell_I,-\ell_I-k_I,\dots)\,.
 \end{split}
\end{equation}
The analysis of the integrand can be further simplified by focussing on the `left' and `right' contributions $\cI^{L,R}$ to the integrand individually, where $\cI=\cI^L \cI^R$. From the discussion of the preceding sections, we identify $\cI^L=\cI^R=\sum_r \pf(M_{\text{NS}}^r)$ for pure gravity in $d$ dimensions, $\cI^L=\sum_r \pf(M_{\text{NS}}^r)$ and $\cI^R=\cI^{cPT}$ for pure Yang-Mills, and $\cI^L=\cI^R=\cI^{cPT}$ for the bi-adjoint scalar theory. Therefore, it is sufficient to prove that
\begin{equation}
 {\cI^{(1)}}^{L,R}=\varepsilon^{-(|\cI|-1)}{\cI^{(0)}}^{L,R}_{I}\,{\cI^{(0)}}^{L,R}_{\bar I}\,.
\end{equation}

\subsubsection{The \texorpdfstring{$n$}{n}-gon integrand}
Let us first consider the $n$-gon integrand
\begin{equation}
 \hat{\cI}^{n-gon}=\frac{1}{\sigma_{\ell^+\,\ell^-}^4}\prod_{i=1}^n \left(\frac{\sigma_{\ell^+\,\ell^-}}{\sigma_{i\,\ell^+}\sigma_{i\,\ell^-}}\right)^2\,.
\end{equation}
It is straighforward to see that case I cannot contribute since the integrand scales as $\varepsilon^0$, and thus the amplitude behaves as $\varepsilon^{2|I|-3}=\cO(\varepsilon)$. Therefore only case II contributes; and the integrand factorises as
\begin{equation}
 \hat{\cI}^{n-gon}=\varepsilon^{-2(|\cI|-1)}\frac{1}{\sigma_{I\,\ell^-}^4}\prod_{i\in I}\frac{1}{x_i^2}\,\prod_{p\in\bar I}\left(\frac{\sigma_{I\,\ell^-}}{\sigma_{p\,I}\sigma_{p\,\ell^-}}\right)^2=\varepsilon^{-2(|\cI|-1)}\,\cI_{I}^{n-gon}\cI_{\bar I}^{n-gon}\,.
\end{equation}
Note that we have used explicitly the chosen gauge fixing $x_{\bar I}=\infty$ for the second equality. In particular, since $\cI^{n-gon}_{I, \bar I}$ do not depend on $\ell$, this gives the correct residues for the respective Q-cuts of the $n$-gon. 

\subsubsection{The Parke-Taylor factor}
Consider next the Parke-Taylor-like integrands
\begin{align}
 \cI^{cPT}(\sigma_{\ell^+},\,\alpha,\,\sigma_{\ell^-})= \sum_{\alpha\in S_n}\cI^{PT}(\sigma_{\ell^+},\,\alpha,\,\sigma_{\ell^-}) =\sum_{\alpha\in S_n}\frac{\sigma_{\ell^+\,\ell^-}}{\sigma_{\ell^+\,\alpha(1)}\sigma_{\alpha(1)\,\alpha(2)}\dots\sigma_{\alpha(n)\,\ell^-}}\,.
\end{align}
If the set $I$ is not consecutive in any of the orderings of the Parke-Taylor factors in the cyclic sum above, the amplitude scales as $\cO(\varepsilon)$ and thus vanishes. Therefore the only non-vanishing contributions come from terms where all $\sigma_i$, $i\in I$ are consecutive with respect to the ordering defined by the Parke-Taylor factors.\\

In case II, with $\sigma_{\ell^+}\in I$, the only term contributing is $\cI^{PT}(\sigma_{\ell^+},\,I,\,\bar I,\,\sigma_{\ell^-})$,
and we find the correct scaling behaviour to reproduce the pole,
\begin{equation} \cI^{cPT}(\sigma_{\ell^+},\,\alpha,\,\sigma_{\ell^-})=\varepsilon^{-(|I|-1)}\,\cI^{PT}_I(x_{\ell^+},\,I,\,x_{\bar I})\,\cI^{PT}_{\bar I}(\sigma_{I},\,\bar I,\,\sigma_{\ell^-})\,.
\end{equation}
In particular, the integrands are again independent of the loop momentum $\ell$, and are straightforwardly identified as the tree-level Parke-Taylor factors, $\cI^{PT}=\cI^{PT,(0)}$.
Note furthermore that the reduction of the integrands from a sum over cyclic Parke-Taylor factors to single terms can be understood directly in terms of diagrams, as only a single diagram will contribute to a given pole. However, another nice interpretation can be given for the biadjoint scalar theory discussed in \cite{He:2015yua}: Here, the cyclic sum is understood as a tool to remove unwanted tadpole contributions to the amplitude. The factorising Riemann surface however separates the insertions of the loop momenta, and thereby automatically removes these tadpole diagrams.\\

In case I, the same argument as above can be used to deduce that the only terms contributing on the factorised Riemann sphere are those where all $i\in I$ appear in a consecutive ordering, and we find
\begin{equation}
 \cI^{cPT}(\sigma_{\ell^+},\,\alpha,\,\sigma_{\ell^-})=\varepsilon^{-(|I|-1)}\cI^{PT}_I(I)\, \cI^{cPT}_{\bar I}(\sigma_{\ell^+},\,\alpha(\bar I\cup\{\sigma_I\}),\,\sigma_{\ell^-})\,
\end{equation}
which again leads to the expected poles and residues, with $\cI^{cPT}_{\bar I}=\cI^{(1)}_{\bar I}$.

\subsubsection{Non-supersymmetric theories}
\label{sec:non-supersymm-theor-1}
In both the case of the $n$-gon and the Parke-Taylor factors, the integrand was independent of $\ell$, and thus Q-cuts were easily identified. For non-supersymmetric theories, with Pfaffians in the integrands, this identification becomes more involved. We will focus first on the NS sector\footnote{In the case of Yang-Mills, this is identical to the pure case, see \cref{eq:pure-YM}.},
\begin{equation}
 \hat{\cI}^{\text{NS}}= \sum_r\pf'(M_{\text{NS}}^r)\,,
\end{equation}
with $\pf'(M_{\text{NS}})\equiv \frac{1}{\sigma_{\ell^+\,\ell^-}}\,\pf({M_{\text{NS}}}_{(\ell^+,\ell^-)})$, and $M_{\text{NS}}$ defined in \cref{eq:defMNS_details}. As above, we have used the subscript $(ij)$ to denote that both the rows and the columns $i$ and $j$ have been removed from the matrix. Consider first again the case II where the Riemann sphere degenerates as $\sigma_i \rightarrow \sigma_I+\varepsilon x_i +\cO(\varepsilon^2)$ with $\sigma_{\ell^+}\in I$. Then the entries in $M_{\text{NS}}$ behave to leading order in $\varepsilon$ as
\begin{equation}
 \frac{1}{\sigma_{ij}}=\begin{cases} \varepsilon^{-1}\frac{1}{x_{ij}} & i,j\in I\\
                                     \frac{1}{\sigma_{Ij}} & i\in I,\, j\in \bar I\\
                                     \frac{1}{\sigma_{ij}} & i,j\in \bar I\\
                       \end{cases} ,\qquad P^\mu(\sigma_i)=\begin{cases} \varepsilon^{-1}P^\mu_I(x_i)&i\in I\\ P^\mu_{\bar I}(\sigma_i)&i\in\bar I \end{cases}\,.
\end{equation}
Using antisymmetry of the Pfaffian, we can rearrange the rows and columns such that $M_{\text{NS}}^r$ takes the following form:
\begin{equation}
\begingroup \renewcommand*{\arraystretch}{1.4}
 M_{\text{NS}}^r=\begin{pmatrix}\varepsilon^{-1} {M^{(0)}_I}_{(\bar I\bar I')} & N\\ -N^T & {M^{(0)}_{\bar I}}_{(II')}\end{pmatrix}\,,
 \endgroup
\end{equation}
where $M_{I}^{(0)}$ is the tree-level matrix, depending only on higher-dimensional on-shell deformations of the loop momentum $\ell_I=\ell+\eta$ with polarisation $\epsilon^r$, and the momenta $k_i$, $i\in I$. In particular, the diagonal entries $C_{ii}=\epsilon_i\cdot P(\sigma_i)$ respect this decomposition due to the one-form $P_\mu$ factorising appropriately. The matrices $N$ are defined (to leading order in $\varepsilon)$ by $N_{ij}=\mu_i\cdot \nu_i$, with 
\begin{align*}
 \mu=\left(\ell_I,k_i,\epsilon^r,\epsilon_i\right)\,,\qquad \nu=\left(\frac{-\ell_I}{\sigma_{I\ell^-}},\frac{k_j}{\sigma_{Ij}},\frac{\epsilon^r}{\sigma_{I\ell^-}},\frac{\epsilon_j}{\sigma_{Ij}}\right)\,,
\end{align*}
for $i\in I$, $j\in \bar I$, where $\ell_I=\ell+\eta$, with $\eta^2=-\ell^2$. Note in particular that this ensures that $N_{\ell^+\ell^-}=0$. To identify the scaling of the integrand $\cI_{\text{NS}}$, we have to consider the Pfaffian of the reduced matrix ${M_{\text{NS}}^r}_{(\ell^+\ell^-)}$;
\begin{equation}
\begingroup \renewcommand*{\arraystretch}{1.5}
 {M_{\text{NS}}^r}_{(\ell^+\ell^-)}=\begin{pmatrix}\varepsilon^{-1} {M^{(0)}_I}_{(\ell^+\bar I\bar I')} & N_{[\ell^+\ell^-]}\\ -N^T_{[\ell^+\ell^-]} & {M^{(0)}_{\bar I}}_{(\ell^- II')}\end{pmatrix}\,,
 \endgroup
\end{equation}
where in $N$ only the row (column) associated to $\ell^+$ ($\ell^-$) has been removed.
Note in particular that the matrices ${M^{(0)}_I}_{(\ell^+\bar I\bar I')}$ and ${M^{(0)}_{\bar I}}_{(\ell^- II')}$ have odd dimensions, so the scaling in $\varepsilon$ is non-trivial. To identify the leading behaviour of the bosonic integrand in $\varepsilon$, we will use the following lemma:
\begin{lemma}[Factorisation Lemma \cite{Dolan:2013isa}] \label{lemma:fact}
 Let $M_I$ and $M_{\bar I}$ be antisymmetric matrices of dimensions $m_I\times m_I$ and $m_{\bar I}\times m_{\bar I}$ respectively; and $N=\mu_i\nu_j$, with $d$-dimensional vectors $\mu_i=\mu_i^\mu$, $\nu_j=\nu_j^\mu$ for $i\in I$, $j\in \bar I$. Then the leading behaviour of the Pfaffian of
 \begin{equation}
  M=\begin{pmatrix} \varepsilon^{-1}M_I & N \\ -N^T & M_{\bar I} \end{pmatrix}
 \end{equation}
as $\varepsilon\rightarrow 0$ is given by
\begin{itemize}
 \item $m_I +m_{\bar I}$ is odd: $\pf(M)=0$.
 \item $m_I +m_{\bar I}$ is even, $m_I$ and $m_{\bar I}$ are even: $\pf(M)\sim\varepsilon^{-\frac{m_I}{2}}\pf(M_I)\pf(M_{\bar I})$
 \item $m_I +m_{\bar I}$ is even, $m_I$ and $m_{\bar I}$ are odd: 
 \begin{equation}
  \pf(M)\sim \varepsilon^{-\frac{m_I-1}{2}}\sum_s\pf(\widetilde{M}_I^s)\,\pf(\widetilde{M}_{\bar I}^s)\,,
 \end{equation}
where $s$ runs over a basis $\epsilon^s$, and 
\begin{equation}
 \widetilde{M}_I^s=\begin{pmatrix} M_I & (\mu\cdot\epsilon^s)^T\\ -\mu\cdot\epsilon^s & 0 \end{pmatrix}\,, \qquad \widetilde{M}_{\bar I}^s=\begin{pmatrix} 0 & \nu\cdot\epsilon^s\\ -(\nu\cdot\epsilon^s)^T & M_{\bar I} \end{pmatrix}\,.
\end{equation}
\end{itemize}
\end{lemma}
The interested reader is referred to \cite{Dolan:2013isa} for the proof of this lemma relying on basic properties of the Pfaffian. Applying \cref{lemma:fact} to the integrand $\cI_{\text{NS}}$, we can identify ${\widetilde{M}^{(0)}_I}$\hspace{-4pt}$_{(\ell^+\bar I\bar I')}$ with ${M^{(0)}_I}_{(\ell^+\bar I)}^{rs}$ by identifying the additional `$s$' row and column with the ones associated to the interchanged particle $\bar I'$ \footnote{and similarly for ${\widetilde{M}^{(0)}_{\bar I}}$\hspace{-4pt}$_{(\ell^- I I')}$ and ${M^{(0)}_{\bar I}}_{(\ell^-I)}^{rs}$}. To leading order in $\varepsilon$, the integrand therefore becomes
\begin{equation}
 \hat{\cI}^{\text{NS}}=
 \varepsilon^{-(|I|-1)} \frac{1}{\sigma_{I\ell^-}}\sum_{r,s} \pf\left({M^{(0)}_I}_{(\ell^+\bar I)}^{rs}\right)\pf\left({M^{(0)}_{\bar I}}_{(\ell^- I)}^{rs}\right)\,.
\end{equation}
Recalling furthermore the gauge fixing choices in degenerating the worldsheet with $x_{\bar I}=\infty$, this can be identified with a product of reduced Pfaffians,
\begin{equation}
 \hat{\cI}^{\text{NS}}=
 \varepsilon^{-(|I|-1)} \sum_{r,s} \pf'\left({M^{(0)}_I}^{rs}\right)\pf'\left({M^{(0)}_{\bar I}}^{rs}\right)\,.
\end{equation}
As seen from the discussion above, this provides both the correct scaling in the degeneration parameter $\varepsilon$ and the correct residues for the Q-cut factorisation.\\

The discussion for case I proceeds along similar lines: For convenience, we choose to remove rows and columns associated to one particle on each side of the degeneration from $M_{\text{NS}}$. Following through the same steps as for case II, the integrand then factorises as 
\begin{equation}
 \hat{\cI}^{\text{NS}}=
 \varepsilon^{-(|I|-1)} \sum_{r,s} \pf'\left({M^{(0)}_I}^{s}\right)\pf'\left({M^{(1)}_{\bar I}}^{rs}\right)\,.
\end{equation}
This correctly reproduces the poles and residues for the bubbling of a Riemann sphere: as a partial sum of the external momenta goes null, the residue is a product of a one-loop amplitude and a tree-level amplitude.\\

{\bf Factorisation for Pure Yang Mills and gravity amplitudes.} At this point it is easy to see how this analysis extends to pure Yang-Mills and gravity. Note first of all that for Yang-Mills, the NS and the pure sector are identical, see \cref{eq:pure-YM},
\begin{equation}
  \mathcal{I}_{\mathrm{pure-YM}}=
  ( \pf(M_{3})\big|_{\sqrt{q}}+ (d-2)\,\pf(M_{3})\big|_{q^{0}})\,\mathcal{I}^{PT}\,.
\end{equation}
For pure gravity, \cref{eq:pure-gravity}, we have
\begin{equation}
  \mathcal{I}_{\mathrm{pure-grav}}=( \pf(M_{3})\big|_{\sqrt{q}}+ (d-2)\,\pf(M_{3})\big|_{q^{0}})^{2}-\alpha\,(\pf(M_{3})\big|_{q^{0}})^{2}\,,
\end{equation}
where $\alpha=\frac{1}{2}(d-2)(d-3)+1$ is given by the degrees of freedom of the B-field and the dilaton. This factorises,
\begin{equation}
\begin{aligned}
  &\mathcal{I}_{\mathrm{pure-grav}}=\cI_L\,\cI_R\,,\qquad\text{with} &&\cI_L=\pf(M_{3})\big|_{\sqrt{q}}+ \alpha_d\,\pf(M_{3})\big|_{q^{0}}\,,\\
  & && \cI_R=\pf(M_{3})\big|_{\sqrt{q}}+ \tilde{\alpha}_d\,\pf(M_{3})\big|_{q^{0}}\,,
\end{aligned}
\end{equation}
where $\alpha_d=d-2+\sqrt{\alpha}$ and $\tilde{\alpha}_d=d-2-\sqrt{\alpha}$. An analogous calculation to \cref{sec:MNS} then straightforwardly leads to
\begin{equation}
 \hat{\cI}_L=\sum_r \pf'(M_{\alpha_d}^r)\,, \qquad \hat{\cI}_R=\sum_r\pf'(M_{\tilde{\alpha}_d}^r)\,,
\end{equation}
where $M_{\alpha_d}^r$ and $M_{\tilde{\alpha}_d}^r$ have been defined as $M_{\text{NS}}^r$, but with the element $B_{\ell^+\ell^-}^{\text{NS}}=\frac{d-2}{\sigma_{\ell^+\ell^-}}$ replaced by $B_{\ell^+\ell^-}^{\alpha_d}=\frac{\alpha_d}{\sigma_{\ell^+\ell^-}}$ and $B_{\ell^+\ell^-}^{\tilde{\alpha}_d}=\frac{\tilde{\alpha}_d}{\sigma_{\ell^+\ell^-}}$ respectively. Then the discussion given above for the NS sector generalises straightforwardly, and the factorisation lemma, in conjunction with the same identification of the matrices, yields again
\begin{equation}
 \hat{\cI}^{\text{pure}}=
 \varepsilon^{-(|I|-1)} \sum_{r,s} \pf'\left({M^{(0)}_I}^{rs}\right)\pf'\left({M^{(0)}_{\bar I}}^{rs}\right)\,.
\end{equation}
Note in particular that since only the matrix $N$ is affected by the change $d-2\rightarrow \alpha_d$, the residues are unchanged, and thus still correspond to the expected tree-level amplitudes for pure Yang-Mills and gravity. Again, case I proceeds in close analogy to the NS sector discussion above.

\subsection{UV behaviour of the one-loop amplitudes} \label{sec:powercounting}
Consider now the UV behaviour of the 1-loop amplitudes; $\ell\rightarrow\lambda^{-1}\ell$, with $\lambda\rightarrow 0$. In this set-up, the scattering equations only yield solutions if the two insertion points of the loop momentum coincide, $\sigma_{\ell^-}\rightarrow\sigma_{\ell^+}$. The factorisation of the scattering equations and the measure will be closely related to \cref{sec:fact_SE}, so we will restrict the discussion to highlight the differences due to the factor of $\lambda^{-1}$. As above, we will blow up this concentration point into a bubbled-off Riemann sphere,
\begin{equation}\label{eq:degen_UV}
 \sigma_{\ell^-}=\sigma_{\ell^+}+\varepsilon +\varepsilon^2 y_{\ell}+\cO(\varepsilon^3)\,,
\end{equation}
where we have used the M\"obius invariance on the sphere to fix $x_{\ell^-}=1$. 
We thus find the scattering equations
\begin{subequations}
  \begin{align}
    &0=\text{Res}_{\ell^+}P^2=\lambda^{-1} \sum_i \frac{\ell\cdot k_i}{\sigma_{\ell^+i}} \\
    &0=\text{Res}_{\ell^-}P^2=-\lambda^{-1} \sum_i \frac{\ell\cdot k_i}{\sigma_{\ell^+i}}+\varepsilon \lambda^{-1}\sum_i \frac{\ell\cdot k_i}{\sigma_{\ell^+i}^2}+\cO(\varepsilon^2)  \label{eq:SE_UV}\\
    &0=k_i\cdot P(\sigma_i)=-\varepsilon \lambda^{-1} \frac{\ell\cdot k_i}{\sigma_{\ell^+i}^2} +\sum_j \frac{k_i\cdot k_j}{\sigma_{ij}} +\cO(\varepsilon^2)\,.
  \end{align}
\end{subequations}
On the support of the  scattering equations at $\sigma_{\ell^{+}}$ and $\sigma_i$ for $i\neq i_1,i_2,i_3$, \cref{eq:SE_UV} simplifies to
\begin{align}
 0=\text{Res}_{\ell^-}P^2&=\varepsilon \lambda^{-1}\sum_{i=i_1,i_2,i_3} \frac{\ell\cdot k_i}{\sigma_{\ell^+i}^2}+\sum_{\substack{i,j\\i\neq i_1,i_2,i_3}}\frac{k_i\cdot k_j}{\sigma_{ij}}+\cO(\varepsilon^2)\\
 &\equiv \varepsilon \lambda^{-1} \mathcal{F}_1-\mathcal{F}_2\,,
\end{align}
where the explicit form of $\mathcal{F}_{1,2}$ will be irrelevant for the following discussion. Including the factor of $\ell^{-2}$, the measure therefore factorises (to leading order) as
\begin{equation}
\begin{split}
 d\mu&\equiv\frac{\lambda^2}{\ell^2}\bar{\delta}(\text{Res}_{\ell^+}P^2)\bar{\delta}(\text{Res}_{\ell^-}P^2)\prod_{i\neq i_1,i_2,i_3}\bar{\delta}(k_i\cdot P(\sigma_i))d\sigma_id\sigma_{\ell^+}d\sigma_{\ell^-}\\
 &=\lambda^4 \,\bar{\delta}\left(\varepsilon -\lambda \frac{\mathcal{F}_2}{\mathcal{F}_1}\right) \,d\varepsilon\, d\tilde{\mu}\,,
 \end{split}
\end{equation}
where $d\tilde{\mu}$ is independent of $\lambda$ and $\varepsilon$. The remaining delta-dunction thus fixes the scaling of the worldsheet degeneration $\varepsilon$, to be proportional to the UV scaling $\lambda$ of the loop momentum $\ell$.\\

Again, this factorisation behaviour of the measure is universal for all theories, and only the specific form of the integrand will dictate the UV scaling of the theory. Denoting the scaling of $\mathcal{I}_{L,R}$ in $\varepsilon$ by $N_{L,R}$, the scattering equation fixing $\varepsilon$ implies that the 1-loop amplitudes scale as
\begin{equation}
 \cM\rightarrow\lambda^{4+N_L+N_R}\cM\,.
\end{equation}
Let us now consider the different supersymmetric and non-supersymmetric theories discussed above.\\

\subsubsection{The \texorpdfstring{$n$}{n}-gon} The integrand of the $n$-gon,
\begin{equation}
 \hat{\cI}^{n-gon}=\frac{1}{\sigma_{\ell^+\ell^-}^4}\prod_{i=1}^n \left(\frac{\sigma_{\ell^+\ell^-}}{\sigma_{i\ell^+}\sigma_{i\ell^-}}\right)^2\,,
\end{equation}
manifestly scales as $\varepsilon^{2n-4}$ under the worldsheet degeneration \ref{eq:degen_UV}. The leading behaviour of the amplitudes is thus given by $\lambda^{4+N}=\lambda^4$ for $\lambda\rightarrow 0$, and therefore the $n$-gons scale as $\ell^{-2n}$ in the UV limit.\\

\subsubsection{The Parke-Taylor factor}
The Parke-Taylor integrand \cref{eq:PTloop-def} contributing in Yang-Mills and the biadjoint scalar theory is given by\footnote{Note that we have chosen to include a factor of $\sigma_{\ell^+\ell^-}$ symmetrically in each integrand $\cI^{L,R}$.}
\begin{equation}
 \hat{\cI}^{PT}=\frac{1}{\sigma_{\ell^+\ell^-}}\sum_{i=1}^n \frac{1}{\sigma_{\ell^+\,i}\sigma_{i+1\, i}\sigma_{i+2\, i+1}\ldots \sigma_{i+n\,\ell^-}}\, .
\end{equation}
While naively this scales as $\varepsilon^{-1}$, the leading order
cancels due to the photon decoupling identity --- a special case of
the KK relations ---
\begin{equation}
 0=\sum_{i=1}^n \frac{1}{\sigma_{\ell\,i}\sigma_{i+1\, i}\sigma_{i+2\, i+1}\ldots \sigma_{i+n\,\ell}}\, ,
\end{equation}
and the integrand thus scales as $\varepsilon^0$. In particular, this allows us to identify immediately the UV behaviour of the bi-adjoint scalar theory as $\ell^{-4}$. This result can be given an intuitive interpretation in terms of Feynman diagrams; the UV behaviour of the theory is determined by the diagrams involving bubbles, which scale as $\ell^{-4}$.

\subsubsection{Supersymmetric theories}
For supersymmetric theories, the UV behaviour is governed by the scaling of the integrand \cref{eq:IL0-sugra}
\begin{equation}
 \hat{\cI}_0^{L,R}=\frac{1}{\sigma_{\ell^+\ell^-}^2}\cI_0^{L,R}\,,\qquad \cI_0^{L,R}=\pf(M_{3}) \big|_{\sqrt{q}}+
  8\left(\pf(M_{3}) \big|_{q^0}-\pf(M_{2} )\big|_{q^0}\right)\,,
\end{equation}
under the worldsheet degeneration described above. Note first that the Szeg\H{o} kernels become (see \cref{eq:szegolimitsl2c})
\begin{align*}
 S_2(\sigma_{ij}) &= \frac{1}{2\sigma_{i\,j}} \left(\sqrt{\frac{\sigma_{i\,\ell^+}\,\sigma_{j\,\ell^-}}{\sigma_{j\,\ell^+}\,\sigma_{i\,\ell^-}}}+ \sqrt{\frac{\sigma_{j\,\ell^+}\,\sigma_{i\,\ell^-}}{\sigma_{i\,\ell^+}\,\sigma_{j\,\ell^-}}} \right)  \sqrt{d\sigma_ i} \sqrt{d\sigma_ j} &\to\sum_{m=0}^\infty  \varepsilon^m S_2^{(m)}\sqrt{d\sigma_ i} \sqrt{d\sigma_ j}\\
S_3(\sigma_{ij}) &=\frac{1}{\sigma_{i\,j}} \left(1 +\sqrt{q} \;\frac{(\sigma_{i\,j}\,\sigma_{\ell^+\,\ell^-})^2}{\sigma_{i\,\ell^+}\,\sigma_{i\,\ell^-}\,\sigma_{j\,\ell^+}\,\sigma_{j\,\ell^-}}\right) \sqrt{d\sigma_ i} \sqrt{d\sigma_ j} &\to \sum_{m=0}^\infty \varepsilon^m S_3^{(m)}\sqrt{d\sigma_ i} \sqrt{d\sigma_ j}\,,
\end{align*}
where
\begin{subequations}
  \begin{align}
    S_2^{(0)}&= \frac{1}{\sigma_{ij}} &  S_3^{(0)}&= \frac{1}{\sigma_{ij}}\\
    S_2^{(1)}&= 0 &  S_3^{(1)}&= 0\\
    S_2^{(2)}&= \frac{1}{8}\,\frac{\sigma_{ij}}{\sigma_{i\ell^+}^2\sigma_{j\ell^+}^2} &  S_3^{(2)}&= \frac{\sigma_{ij}}{\sigma_{i\ell^+}^2\sigma_{j\ell^+}^2}\\
    S_2^{(3)}&= \frac{1}{8}\left(\frac{\sigma_{ij}(\sigma_{i\ell^+}+\sigma_{j\ell^+})}{\sigma_{i\ell^+}^3\sigma_{j\ell^+}^3}+\frac{2\,y_\ell \,\sigma_{ij}}{\sigma_{i\ell^+}^2\sigma_{j\ell^+}^2}\right) &  S_3^{(3)}&=\frac{\sigma_{ij}(\sigma_{i\ell^+}+\sigma_{j\ell^+})}{\sigma_{i\ell^+}^3\sigma_{j\ell^+}^3}+\frac{2\,y_\ell \,\sigma_{ij}}{\sigma_{i\ell^+}^2\sigma_{j\ell^+}^2}\,.
  \end{align}
\end{subequations}
Expanding the integrand $\cI_0^{L,R}$ in powers of $\varepsilon$, $\pf(M_{3}) \big|_{q^0}$ and $\pf(M_{2}) \big|_{q^0}$ potentially contribute at order $\varepsilon^0$, whereas $\pf(M_{3}) \big|_{\sqrt{q}}$ can only contribute to $\varepsilon^2$. However, the leading contribution cancles among $\pf(M_{3}) \big|_{q^0}$ and $\pf(M_{2}) \big|_{q^0}$, as well as all higher order contribution (starting at order $\varepsilon^1$) coming from the diagonal entries of $C$. The scaling in $\varepsilon$ is thus governed by the higher order behaviour of the Szeg\H{o} kernels. Moreover, due to the factor of $1/8$ between $S_2^{(2,3)}$ and $S_3^{(2,3)}$, the terms of order $\cO(\varepsilon^2)$ and $\cO(\varepsilon^3)$ originating from $\pf(M_{2}) \big|_{q^0}$ cancel against the contributions from $\pf(M_{3}) \big|_{\sqrt{q}}$.\footnote{A bit more care is needed at $\cO(\varepsilon^3)$: While there are no contributions to $\cO(\varepsilon^2)$ from products of terms of order $\varepsilon$, these have to be taken into account at order $\cO(\varepsilon^3)$. However, the same reasoning as above guarantees their cancellation: The only possible origin for terms of order $\varepsilon$ are the diagonal entries of $C$, which coincide for $M_2$ and $M_3$. The cancellations to second order thus carry forwards to ensure that there will be no contributions from products of lower order terms up to $\cO(\varepsilon^3)$.} \\

A short investigation confirms that there are no further cancellations, and thus $\hat{\cI}_0^{L,R}$ scales as $\varepsilon^2$. In particular, using $N=4+N_L+N_R$, this implies that our one-loop supergravity amplitudes scale as $\ell^{-8}$ in the UV limit, and super Yang-Mills as $\ell^{-6}$ (using $N_R=0$ for the Parke-Taylor integrand derived above). Naively, this seems to be a weaker UV fall-off for super Yang-Mills than expected from the known expression for the integrand as a sum over box diagrams, which exhibits scaling $\ell^{-8}$. See, however, our previous comment on this issue (following \cref{lemma:powercounting}).

\subsubsection{Non-supersymmetric theories}
In the supersymmetric case discussed above, cancellations between the NS and the R sector ensured the correct scaling of the integrand. However, these cancellations are absent in the purely bosonic case,
\begin{equation}
 \hat{\cI}^{\text{NS}}= \frac{1}{\sigma_{\ell^+\ell^-}}\sum_r\pf({M_{\text{NS}}^r}_{(\ell^+\ell^-)})\,,
\end{equation}
so naively the integrand seems to scale as $\varepsilon^{-2}$. However, the leading contribution is given by the Pfaffian of the full tree-level matrix $M^{(0)}_{\text{NS}}=M^{(0)}$, which vanishes on the support of the scattering equations. More explicitly, let us expand the reduced matrix ${M_{\text{NS}}^r}_{(\ell^+\ell^-)}$ in $\varepsilon$;
\begin{equation}
 \pf({M_{\text{NS}}}_{(\ell^+\ell^-)})=\pf(M^{(0)})+\varepsilon \,\pf(M_{\text{NS}}^{(1)}) +\cO(\varepsilon)
\end{equation}
The vanishing of the leading term can then be seen from the existence of two vecors, 
\begin{equation}\label{eq:v_0}
 v_0 = (\sigma_1,\dots,\sigma_n,0\dots,0) \quad \text{and } \quad \tilde{v}_0 = (1,\dots,1,0\dots,0)
\end{equation}
in the kernel of $M_{\text{NS}}^{(0)}=M^{(0)}$. This argument can in fact be extended to subleading order: expanding both the matrix $M_{\text{NS}}$ and a potential vecor in the kernel to subleading order, 
\begin{equation}
 {M_{\text{NS}}}_{(\ell^+\ell^-)}=M^{(0)} +\varepsilon M^{(1)}_{\text{NS}}\,,\qquad v=v_0 +\varepsilon \,v_1
\end{equation}
we note that the condition for $\cI^{\text{NS}}$ to scale as $\cO(\varepsilon^0)$ is that $ {M_{\text{NS}}}_{(\ell^+\ell^-)}$ has co-rank two to order $\varepsilon$, and thus two vectors spanning its kernel. Finding one of these is enough, as it guarantees the existence of the second. But this on the other hand is equivalent to 
\begin{equation}
 \left(M^{(0)} +\varepsilon M^{(1)}_{\text{NS}}\right)\left(v_0 +\varepsilon\, v_1\right)=\cO(\varepsilon^2)
\end{equation}
vanishing to order $\cO(\varepsilon^2)$. Expanding this out, we obtain the conditions
\begin{align}\label{eq:cond_UV}
 M^{(0)}v_0=0\,,\qquad M^{(0)}\,v_1=-M^{(1)}_{\text{NS}}\, v_0\,.
\end{align}
As commented above, the first condition is satisfied with $v_0$ given in \cref{eq:v_0}. Note furthermore that the second condition cannot be straightforwardly inverted, since det($M^{(0)})=0$. The constraint for a solution to exist is thus the vanishing of both sides of the equation under a contraction with a vector in the kernel of the matrix. Since the only contribution to $M^{(1)}_{\text{NS}}$ to order $\cO(\varepsilon)$ comes from the diagonal entries $C_{ii}$, we get
\begin{equation}
 M^{(1)}_{\text{NS}}\, v_0= (0,\dots,0,\lambda^{-1}\frac{\epsilon_i\cdot\ell}{\sigma_{i\ell^+}},\dots)\,.
\end{equation}
This vanishes trivially when contracted with $v_0$ and $\tilde{v}_0$, and thus there exists a solution $v_1$ to \cref{eq:cond_UV}. The contributions of order $\cO(\varepsilon^{-1})$ to $\cI^{\text{NS}}$ therefore vanish, and the integrand scales as $\varepsilon^0$. In particular, this implies that both pure Yang-Mills and pure gravity one-loop amplitudes scale as $\ell^{-4}$ in the UV limit, which is the expected behaviour from both the Q-cut and the Feynman diagram expansion. Moreover, the analysis above is equally applicable to the NS and the pure theories, similarly to the discussion given in \cref{sec:fact_int}.\\

To conclude the proof of \cref{lemma:powercounting}, let us summarise these results for the UV scaling of our one-loop amplitudes:
\begin{center}
\begin{tabular}{l|c}
 theory & scaling $\lambda^N\sim\ell^{-N}$\\ \hline 
 $n$-gon & $N=2n$\\
 supergravity & $N=8$\\
 super Yang-Mills & $N=6$\\
 pure gravity & $N=4$\\
 pure Yang-Mills & $N=4$\\
 bi-adjoint scalar & $N=4$
\end{tabular}
\end{center}
Note that the scaling of the non-supersymmetric theories (pure gravity and Yang-Mills, as well as the bi-adjoint scalar) corresponds to Feynman diagrams involving bubbles, whereas the higher scaling of the supersymmetric theories ensures that only boxes contribute. As commented above, Yang-Mills exhibits a lower scaling than expected from the Feynman diagram expansion, but which coincides with the expected scaling in the Q-cut representation. \\

Let us comment briefly on the closely related discussions regarding the contribution of the singular solutions $\sigma_{\ell^-}=\sigma_{\ell^+} +\varepsilon +\cO(\varepsilon^2)$. The same arguments as above, without the rescaled loop momenta, ensure that the measure scales to leading order as
\begin{equation}
 d\mu=\bar{\delta}\left(\varepsilon - \frac{\mathcal{F}_2}{\mathcal{F}_1}\right) \,d\varepsilon\, d\tilde{\mu}\,,
\end{equation}
where the measure $d\tilde{\mu}$ is again independent of $\varepsilon$. Then the same powercounting argument in the degeneration parameter as above gives the following scaling for the different theories:
\begin{center}
\begin{tabular}{l|c}
 theory & scaling $\varepsilon^N$\\ \hline 
 $n$-gon & $N=2n-4$\\
 supergravity & $N=4$\\
 super Yang-Mills & $N=2$\\
 pure gravity & $N=0$\\
 pure Yang-Mills & $N=0$\\
 bi-adjoint scalar & $N=0$
\end{tabular}
\end{center}
The contribution from the singular solutions $\sigma_{\ell^-}=\sigma_{\ell^+} +\varepsilon +\cO(\varepsilon^2)$ to the $n$-gon and the supersymmetric theories thus vanishes, whereas they can clearly be seen to contribute for the bi-adjoint scalar theory and Yang-Mills and gravity in the absence of supersymmetry. Moreover, the scaling as $\varepsilon^0$ guarantees that the contributions are clearly finite. This complements the discussion given in \cref{sec:general-form-one}

\section{Discussion}\label{sec:discussion}
In giving the theory that underlies the CHY formulae for tree amplitudes, ambitwistor strings gave a route to conjectures for the extension of those formulae to loop amplitudes.  Being chiral string theories, ambitwistor strings potentially have more anomalies than conventional strings, but nevertheless the version appropriate to type II supergravity led to consistent proposals for amplitude formulae at one and two loops 
\cite{Adamo:2013tsa,Casali:2014hfa,Adamo:2015hoa}.  However, the other main ambitwistor string models would seem to have problems on the torus, either with anomalies, or because the full ambitwistor string theories have unphysical modes associated with their gravity sectors that would propagate in the loops and corrupt for example a  pure Yang-Mills loop amplitude.  Furthermore, once on the torus, it is a moot point as to how much can be done with the formulae, requiring as they do, the full machinery of theta functions.  Issues such as the Schottky problem will make higher loop/genus formulae difficult to write down explicitly.

In  \cite{Geyer:2015bja}, with the further details and extensions given in this paper, we have seen that the conjectures of \cite{Adamo:2013tsa,Casali:2014hfa}, with the adjustment to the scattering equations as described in \S\ref{sec:review}, are equivalent to much simpler conjectures on the Riemann sphere.  These formulae are now of the same complexity as the CHY tree-level scattering formulae on the Riemann sphere with the addition of two marked points, corresponding to loop momenta insertions.  It is therefore possible to apply methods that have been developed at tree-level on the Riemann sphere here also at 1-loop to both extend and prove the conjectures.

As far as extensions are concerned,  we were able in \cite{Geyer:2015bja} to make conjectures for 1-loop formulae for maximal super Yang-Mills for which there is not a good formula on the torus, by replacing one of the Pfaffians with a Parke-Taylor factor, symmetrized so as to run through the loop in all possible ways.  The approach was also suggestive of formulae for the biadjoint scalar theory also, but we were not able to confirm those numerically in \cite{Geyer:2015bja}.  However, these were studied further in \cite{Baadsgaard:2015hia,He:2015yua} where the difficulties that we had were to a certain extent resolved and are associated with degenerate solutions to the scatering equations.

 In our original formulation, we only considered $(n-1)!-2(n-2)!$ solutions to the scattering equations for an $n$ particle amplitude at 1-loop.   This counting was more clearly understood in \cite{He:2015yua}.  The $(n-1)!$ is the number of solutions that one obtains for $n+2$ points on the sphere with arbitrary null momenta at $n$ points, and off-shell momenta at the remaining two points (all summing to zero).  If one takes the forward limit in which  the two off-shell momenta become equal and opposite, one finds that there are two classes of  $(n-2)!$ degenerate solutions, in which the  two loop insertion points  come together (or alternatively all the other points come together); the two classes are distinguished by the rate at which the points come together as the forward limit is taken.   In the forward limit at which we are working, the most degenerate  class no longer applies but, in general, we can consider the other.  For amplitudes in supersymmetric Yang-Mills and gravity, these degenerate solutions give a vanishing contribution to the loop integrand. However, they do contribute in the case of the biadjoint scalar theory, as shown in \cite{He:2015yua}, and they also contribute in the cases of non-supersymmetric Yang-Mills and gravity presented in this paper. 
 
However,  as seen in \S\ref{sec:factorization}, the degenerate solutions do not contribute to the Q-cuts.   So,  to arrive at a loop integrand that computes the correct amplitude under dimensional regularisation, we  simply discard them in our proposed formulae also.  Having discarded these terms, our formulae will not then necessarily give the integrand itself as a sum of Feynman diagrams.  In the biadjoint scalar theory, for example, there will be terms that look like tree amplitudes with bubbles on each external leg that will vanish under dimensional regularisation.  These are correctly computed if the degenerate solutions are included as shown by \cite{He:2015yua}.  It would be interesting to see if this persists for all our formulae as we have seen that they make sense on the degenerate solutions.

It should also be possible to prove our one-loop formulae for supersymmetric theories via factorisation.  The  gap in our argument is that we do not have a good formula for the Ramond sector contributions at tree level, as would be required to prove factorisation.  Our  representation of the Ramond sector in the loop as the Pfaffian of $M_2$ should provide some hint as to how to do this. 

Ideally, there should be no need to solve the scattering equations explicitly. The main result of \cite{Geyer:2015bja}, which relied on the use of a residue theorem to localise the modular parameter, was inspired by \cite{Dolan:2013isa}, where the tree-level CHY integrals were computed by successive application of residue theorems, rather than by solving the scattering equations. The way forward is to use the map between integrals over the moduli space of the Riemann sphere and rational functions of the kinematic invariants, which is implicit in the scattering equations. Recently, there has been intense work on making this map more practical \cite{Kalousios:2015fya,Cachazo:2015nwa,Baadsgaard:2015voa,Baadsgaard:2015ifa,Baadsgaard:2015hia,Huang:2015yka,Sogaard:2015dba,Cardona:2015eba}. We expect that this will make the use of our formulae much more efficient.

It was argued in \cite{Geyer:2015bja} that the scheme explored here at one-loop has a natural extension to all loops.  Similarly, the Q-cut formalism of \cite{Baadsgaard:2015twa} also has a natural extension to all loops.  It will be interesting to see  whether the factorization strategy presented in this paper can be extended to give a correspondence with the Q-cut formalism at higher loop order.  Obtaining better control of higher loop Pfaffians will be crucial for  using these ideas to understand gauge theories and gravity. A formulation as correlators on the Riemann sphere, as suggested by our introduction of $M_{\text{NS}}$, may play a key role.

\section*{Acknowledgements}

We thank Tim Adamo, Eduardo Casali, Freddy Cachazo, Henrik Johansson,
David Skinner and Pierre Vanhove for useful suggestions and
comments, and Tim Adamo in particular for comments on the draft. PT is supported by the ERC Grant 247252 STRING, YG by an
EPSRC DTA award and the Mathematical Prizes Fund, LJM by EPSRC grant
EP/M018911/1.

\appendix

\section{Solutions to the four-point 1-loop scattering equations}\label{4pt-soln}

In this appendix, we briefly discuss the solutions to the one-loop scattering equations for $n=4$. There are two solutions to \eqref{SE2n}, in agreement with the counting $(n-1)!-2(n-2)!$. Fixing $\sigma_1=1$, along with the choices $\sigma_{\ell^+}=0$ and $\sigma_{\ell^-}=\infty$ understood in \eqref{SE2n}, these two solutions are given by
\begin{align}
\sigma_2 &= \frac{\ell\cdot k_2 \,(W+1) (k_2\cdot k_3+ \ell\cdot k_4 W)}{W (\ell\cdot k_4 (W
   (\ell+k_1)\cdot k_2+\ell\cdot k_2)+ k_2\cdot k_3\,\ell\cdot k_2)} \nonumber \\
\sigma_3 &= -\frac{(W+1) (\ell\cdot k_4 W-\ell\cdot k_2) (W k_1\cdot k_2+\ell\cdot k_2) (k_2\cdot k_3+\ell\cdot k_4 W)}{W (W
   (k_1\cdot k_2-\ell\cdot k_4)-k_2\cdot k_3+\ell\cdot k_2) (\ell\cdot k_4 (W (\ell+k_1)\cdot k_2+\ell\cdot k_2)+ k_2\cdot k_3\,\ell\cdot k_2)} \nonumber \\
\sigma_4 &= \frac{\ell\cdot k_4 \,(W+1) (W k_1\cdot k_2+\ell\cdot k_2)}{\ell\cdot k_4 (W
   (\ell+k_1)\cdot k_2+\ell\cdot k_2)+ k_2\cdot k_3\,\ell\cdot k_2}  
\end{align}
where $W$ can take the two values
\begin{equation}
W^{(\pm)} = \frac{k_1\cdot k_2\,\ell\cdot k_2-k_2\cdot k_4\,\ell\cdot k_3+k_2\cdot k_3\,\ell\cdot k_4-2 \ell\cdot k_2\,\ell\cdot k_4 \pm \sqrt{4\prod_{i=1}^4 \ell\cdot k_i + \det U}}{2 \ell\cdot k_4
   (k_1\cdot k_2+\ell\cdot (k_1+k_2))},
\end{equation}
\begin{equation}
U = \left( \begin{matrix}
k_i\cdot k_j \;\;& \ell\cdot k_i\\
\ell\cdot k_j  \;\;&  0
\end{matrix} \right) \quad i,j=1,2,3.
\end{equation}
The expressions for $\sigma_ i$ solve $f_2=f_4=0$ for any $W$. The expression for $W$ is then determined by solving $f_3=0$, which takes a quadratic form.

In the case of more general theories, as discussed in section \ref{sec:general-form-one}, there are two additional solutions contributing at four points, in agreement with $(n-2)!$, so that the total number of solutions is $(n-1)!-(n-2)!$. The `regular' solutions are the ones described above, but we should now express them in a different $SL(2,\mathbb{C})$ gauge, where we don't fix both $\sigma_{\ell^+}$ and $\sigma_{\ell^-}$. Let us use coordinates $\sigma'$ such that $(\sigma'_1,\sigma'_2,\sigma'_3)=(0,1,\infty)$. Then we obtain the two `regular' solutions from the expressions above through the change of coordinates
$$
\sigma' = \frac{\sigma_{23}}{\sigma_{21}} \, \frac{\sigma-\sigma_1}{\sigma-\sigma_3}.
$$
For the `singular' solutions, we have $\sigma'_{\ell^+}=\sigma'_{\ell^-}$. The remaining $\sigma'_i$ must satisfy the tree-level scattering equations, so that in our choice $\sigma'_4=-k_1\cdot k_4 / k_1\cdot k_2$. The two solutions for $\sigma'_\ell$ are then determined by
$$
\ell\cdot k_3\,{\sigma'_\ell}^2 + \left( \ell\cdot (k_1+k_4) + \sigma'_4\, \ell\cdot (k_1+k_2) \right) \,\sigma'_\ell  - \sigma'_4\, \ell\cdot k_1 =0.
$$

\section{Motivation from Ambitwistor heterotic models}
\label{sec:motiv-form-ambitw}

The single trace sector of the heterotic ambitwistor model was used to
derive the CHY formulae for gluon amplitudes in
eq.~\cite{Mason:2013sva}. It was however noted that generically these
amplitudes contained unphysical would-be gravitational degrees of
freedom, leading in particular to multi-trace interactions, absent
from Yang-Mills theories. At one loop, the presence of these would be
gravitational interactions leads to a double pole at the boundary of
the moduli space $dq/q^{2}$, coming from the bosonic sector of the
theory. In string theory, the level matching prevents these tachyonic
modes from propagating, and heterotic models were used to write down a
set of rules to compute gluon amplitudes in the 90's
\cite{Bern:1991aq}. Here this double pole simply renders the theory
ill-defined.

One could hope that a subsector of the theory may be well defined at
one-loop, just like at tree-level, but this is not the case.
Even by restricting to the single-trace sector, one does
not automatically decouples these additional states (just think about a
tree-level single-trace connected by an internal graviton in two points)
so we have to be more careful when attempting to extract a portion of
the heterotic amplitude. Let us start by writing it out;

\begin{equation}
  \label{eq:het-one-loop}
  \mathcal{M}^{\mathrm{n-gluon}\,(!!)}_{\mathrm{1-loop}}= \int
  \frac{dq}q \prod_{i}d z_{i} 
 \frac12 \sum_{\alpha}
  Z^{\mathrm{Het}}_{\bd\alpha}
  \mathrm{Pf}{(M^{\mathrm{col}}_{\bd\beta})}\times 
  \frac12\sum_{\bd\beta}
  Z_{\bd\beta}\mathrm{Pf}({M}_{\bd \beta} )
\end{equation}
where the symbol $(!!)$ is here to emphasize that this is not a well
defined amplitude in a well defined theory, but still we shall try to
extract parts of it below.
The matrix $M$ is the ``kinematical'' one of
eq.~\eqref{eq:Malpha}, while the partition functions $Z_{\bd \alpha}$
were defined in eq.~\eqref{eq:pt-funs}
The new ingredient here is the colour part which contains a partition
function for the $32$ Majorana-Weyl fermions that realise the current
algebra and the colour Pfaffian coming from Wick contractions between
them. The partition functions are given by
\begin{equation}
  \label{eq:het-partfun}
 Z^{\mathrm{Het}}_{\bd\alpha}=\frac{\theta_{\alpha}(0|\tau)^{16}}{\eta^{16}(\tau)}\frac{1}{\eta(\tau)^{8}}\,,
\end{equation}
The ``color'' Pfaffian is built out by application of Wick's theorem
in a standard way on the gauge currents
\begin{equation}
  \label{eq:gauge-currents}
  J^{a}=T^{a}_{ij}\psi^{i}\psi^{j}, \quad i,j=1,\ldots,16\,,
\end{equation}
using the fermion propagator
\begin{equation}
  \langle \psi^{i}(z)\psi^{j}(0)\rangle_{\bd \alpha} = S_{\bd \alpha}(z|\tau)
\end{equation}
in the spin structure $\bd\alpha$. In all we have the Pfaffian of a
matrix whose elements are $S_{\bd\alpha}(z_{i})-z_{j}|\tau)T^{a}_{i,j}$.
  However, contrary to the case of kinematics, were the Pfaffian
  structure is somewhat interesting (although it is very hard to read
  off the action of supersymmetry on them), in the case of colour we
  are more used to the colour ordering decomposition. For this reason
  and the one above on decoupling as many gravity states as possible,
  we shall from now on restrict our attention to one particular term
  in this Pfaffian, a single trace one, say
  $\mathrm{Tr}(T^{1}T^{2}T^{3}T^{4})$, such that
\begin{equation}
  \label{eq:pftosingletrace}
  \mathrm{Pf}_{\alpha}\to \prod_{i=1}^{4}
  S_{\alpha}(z_{i}-z_{i+1}|\tau) \mathrm{Tr}(T^{1}T^{2}T^{3}T^{4})\,,\quad z_{5}\equiv z_{1}\,.
\end{equation}

To apply the IBP procedure and write down the integrands down to the
sphere, we need the following $q$-expansions;
\begin{equation}
  \label{eq:het-partfun-qexp}
  \begin{aligned}
    \frac{\theta_{2}(0|\tau)^{16}}{\eta^{24}(\tau)} &= 2^{16}
    q+O(q^{2}),\\
    \frac{\theta_{3}(0|\tau)^{16}}{\eta^{24}(\tau)} &= \frac 1q +\frac{32}{\sqrt{q}}+504+  O(q),\\
    \frac{\theta_{4}(0|\tau)^{16}}{\eta^{24}(\tau)} &= \frac 1q -\frac{32}{\sqrt{q}}+504+O(q),\\
  \end{aligned}
\end{equation}
Since the Pfaffians, made of Szeg\H{o} kernels, are holomorphic in
$q$, one sees straight away that the spin structure 2 in the color
side does not contribute after the IBP. Using the $q$-expansions of
teh Szeg\H{o} kernels already given in the text,
\begin{align}
  \label{eq:szego-exp}
  S_{3} = \frac{\pi }{\sin (\pi  z)}+4 \pi  \left(-\sqrt{q}+q\right)
  \sin (\pi  z)\\
S_{4}=\frac{\pi }{\sin (\pi  z)}+4 \pi  \left(\sqrt{q}+q\right) \sin (\pi  z)
\end{align}
(in their torus parametrization), we have that
$Z_{3}^{\mathrm{Het}}\mathrm{Pf}(M^{\mathrm{col}})_{3}\big|_{q^{-1/2}}=
-Z_{4}^{\mathrm{Het}}\mathrm{Pf}(M^{\mathrm{col}})_{4}\big|_{q^{-1/2}}$.

Thus the integrand of \eqref{eq:het-one-loop} has a double pole in $q$
with coefficient $2\mathrm{Pf}_{3}\big|_{q^{0}}$ (we used $\mathrm{Pf}_{3}\big|_{q^{0}}=\mathrm{Pf}_{4}\big|_{q^{0}}$)
and a single pole in $q$, given by the leading order piece. 

The terms composing the leading order piece are of several kinds;
\begin{align}
  \frac{32}{\sqrt{q}}&\to 32 \frac{\sin(\pi
                      z_{12})} {\sin(\pi z_{23}) \sin(\pi
                      z_{34}) \sin(\pi z_{41})} + \mathrm{
                      cyclic} \\
  \frac 1q & \to  \frac{\sin(\pi
             z_{12})} {\sin(\pi z_{23}) \sin(\pi
             z_{34}) \sin(\pi z_{41})}  +\frac{\sin^{2}}{\sin^{2}} + \mathrm{
             perms}\\
  504 & \to \frac{504}{\sin(\pi
                      z_{12})\sin(\pi z_{23}) \sin(\pi
                      z_{34}) \sin(\pi z_{41})}
\end{align}
where $\frac{\sin^{2}}{\sin^{2}}$ indicate terms of the form
$\frac{\sin(\pi z_{12}) \sin(\pi z_{34})}{\sin(\pi z_{13}) \sin(\pi
  z_{24}) } $.
After the usual change of variable, the $"\sin/\sin^{3}"$ terms give
rise to the PT part of our beloved YM integrands given in
eq.~\eqref{eq:PTloop-def} (including the reversed ones);
\begin{equation}
  \label{eq:sin-sincube}
  \frac{1}{\sigma_{1}\sigma_{2}\sigma_{3}\sigma_{4}}\frac{\sin(\pi
    z_{12})} {\sin(\pi z_{23}) \sin(\pi
    z_{34}) \sin(\pi z_{41})} =
  \frac{\sigma_{41}}{\sigma_{1}\sigma_{4}\sigma_{12} \sigma_{23}\sigma_{34}}
\end{equation}
where an additional factor of
$(\sigma_{1}\sigma_{2}\sigma_{3}\sigma_{4})^{-1}$ has been taken from
the measure $\prod d\sigma_{i}/\sigma_{i}^{2}$. Note also that the
counting for these terms produces a numerical factor of $(32-1)=31$,
which, after suitable counting of the powers of 2, builds up 
\begin{equation}
  496 = 2^{4}\times 31
\end{equation}
which is the dimension of the adjoint of $SO(32)$. The fact that loops
in gauge theories come with a factor of $N$ at leading order is
well-known (the gluons are in the adjoint representation of the gauge
group in these models).

However, we find additional terms. We haven't been able to determine
their origin with precision, but we suspect that they could originate
from bi-adjoint scalars running in the loop, if they are not simple
artefacts of the inconsistency of the model.

\section{The NS part of the Integrand}\label{sec:MNS}
This appendix provides the poof for the equivalence of the two expressions for the NS sector of the integrand;
\begin{equation}
 \cI^{\text{NS}}\equiv(d-2)\pf(M_3)\,\big|_{q^0}+\pf(M_3)\,\big|_{\sqrt{q}} =\sum_r\pf'(M_{\text{NS}}^r)\,,
\end{equation}
with $\pf'(M_{\text{NS}})= \frac{1}{\sigma_{\ell^+\,\ell^-}}\,\pf({M_{\text{NS}}}_{(\ell^+,\ell^-)})$ and $d$ denoting the dimension of the space-time. Recall the definition of the matrix $M_{\text{NS}}$ given in \cref{sec:analys-indiv-gso};
\begin{equation}
 M_{\text{NS}}^r=\begin{pmatrix}A & -C^T\\ C & B \end{pmatrix}
\end{equation}
where the elements of $M_{\text{NS}}^r$ were defined by
\begin{align*}
 &A_{\ell^+\ell^+}=A_{\ell^-\ell^-}=0 && A_{\ell^+\ell^-}=0 && A_{\ell^+ i}=\frac{\ell\cdot k_i}{\sigma_{\ell^+i}} && A_{\ell^- i}=-\frac{\ell\cdot k_i}{\sigma_{\ell^-i}}\\
 & && && A_{ij}=\frac{k_i\cdot k_j}{\sigma_{ij}} && A_{ii}=0\\
 &B_{\ell^+\ell^+}=B_{\ell^-\ell^-}=0 && B_{\ell^+\ell^-}=\frac{d-2}{\sigma_{\ell^+\ell^-}} && B_{\ell^+ i} =\frac{\epsilon^r\cdot\epsilon_i}{\sigma_{\ell^+i}} && B_{\ell^- i}=\frac{\epsilon^r\cdot \epsilon_i}{\sigma_{\ell^-i}}\\
 & && && B_{ij}=\frac{\epsilon_i\cdot \epsilon_j}{\sigma_{ij}} && B_{ii}=0\\
 &C_{\ell^+\ell^+}=-\epsilon^r\cdot P(\sigma_{\ell^+}) && C_{\ell^+\ell^-}=-\frac{\epsilon^r\cdot\ell}{\sigma_{\ell^+\ell^-}} && C_{\ell^+ i}=\frac{\epsilon^r\cdot k_i}{\sigma_{\ell^+i}} && C_{i\ell^+}=\frac{\epsilon_i\cdot \ell}{\sigma_{i\ell^+}}\\
 &C_{\ell^-\ell^-}=-\epsilon^r\cdot P(\sigma_{\ell^-}) && C_{\ell^-\ell^+}=\frac{\epsilon^r\cdot\ell}{\sigma_{\ell^-\ell^+}} && C_{\ell^- i}=-\frac{\epsilon^r\cdot k_i}{\sigma_{\ell^-i}} && C_{i\ell^-}=-\frac{\epsilon_i\cdot \ell}{\sigma_{i\ell^-}}\\
 & && && C_{ij}=\frac{\epsilon_i\cdot k_j}{\sigma_{ij}} && C_{ii}=-\epsilon_i\cdot P(\sigma_i)\,.
\end{align*}
Using the recursive definition of the Pfaffian,
\begin{equation}
 \pf(M)=\sum_{j\neq i}(-1)^{i+j+1+\theta(i-j)}m_{ij}\,\pf(M_{(ij)})\,,
\end{equation}
we can expand the reduced Pfaffian $\pf'(M_{\text{NS}}^r)= \frac{1}{\sigma_{\ell^+\,\ell^-}}\,\pf({M_{\text{NS}}^r}_{(\ell^+,\ell^-)})$, in the remaining rows associated to $\ell^+$ and $\ell^-$.\footnote{For convenience, we have chosen to remove the rows and columns associated to $\ell^+$ and $\ell^-$ in the reduced matrix.} 
\begin{align*}
 \pf'(M_{\text{NS}}^r)=&\frac{1}{\sigma_{\ell^+\ell^-}} B_{\ell^+\ell^-} \,\pf(M_{(\ell^+\ell^-{\ell^+}'{\ell^-}')}) \\
 & + \frac{1}{\sigma_{\ell^+\ell^-}} \hspace{-7pt}\sum_{i,j\neq {\ell^+}',{\ell^-}'} \hspace{-12pt}(-1)^{1+i+j+\theta(n+1-i)+\theta(n+2-j)+\theta(i-j)} m_{{\ell^+}'i} m_{{\ell^-}'j}\, \pf(M_{(\ell^+\ell^-{\ell^+}'{\ell^-}'ij)})\,.
\end{align*}
To briefly comment on the notation used for $M_{\text{NS}}$, rows and columns in $\{1,\dots,n+2\}$ are denoted by indices $i$, whereas we use the conventional $i'=n+2+i$ for rows in $\{n+3,\dots,2(n+2)\}$. Now note that after summing over $r$, the entries and prefactors simplify,
\begin{align*}
 (-1)^{1+i+j+\theta(n+1-i)+\theta(n+2-j)+\theta(i-j)}&=\begin{cases} (-1)^{1+i+j+\theta(i-j)} & i,j\in I\,,\\
								    (-1)^{1+i+j+\theta(i-j)} & i,j\in I'\,,\\
								    (-1)^{i+j} & i\in I',\, j\in I\,,
                                                      \end{cases}\\
\sum_r m_{{\ell^+}'i} m_{{\ell^-}'j}&=\begin{cases} \frac{k_i\cdot k_j}{\sigma_{\ell^+i}\sigma_{\ell^-j}} & i,j\in I\,,\\
				       \frac{\epsilon_i\cdot \epsilon_j}{\sigma_{\ell^+i}\sigma_{\ell^-j}} & i,j\in I'\,,\\
				       \frac{\epsilon_i\cdot k_j}{\sigma_{\ell^+i}\sigma_{\ell^-j}} & i\in I',\, j\in I\,,
                         \end{cases}
\end{align*}
where $I=\{1,\dots,n\}$, $I'=\{1',\dots,n'\}$. Moreover, using $M_{(\ell^+\ell^-{\ell^+}'{\ell^-}')}=M_3$ and choosing for simplicity $\sigma_{\ell^+}=0$, $\sigma_{\ell^-}=\infty$, the result further simplifies to
\begin{equation}
 \begin{split}
 \sum_r \pf'(M_\text{NS}^r)=(d-2)\,\pf(M_3)& + \sum_{i,j\in I}(-1)^{1+i+j+\theta(i-j)} \frac{k_i\cdot k_j}{\sigma_i} \,\pf({M_3}^{ij}_{ij})\\
  &+ \sum_{i,j\in I'}(-1)^{1+i+j+\theta(i-j)} \frac{\epsilon_i\cdot \epsilon_j}{\sigma_i} \,\pf({M_3}^{i'j'}_{i'j'})\\
  &+\sum_{i\in I',j\in I}(-1)^{i+j} \left(\frac{\epsilon_i\cdot k_j}{\sigma_i}-\frac{\epsilon_i\cdot k_j}{\sigma_j}\right) \,\pf({M_3}^{i'j}_{i'j})\,.
 \end{split}
\end{equation}
To see that, as claimed above,
\begin{equation}
  \sum_r \pf'(M_\text{NS}^r)=(d-2)\, \pf(M_3) \big|_{q^0}+ \pf(M_3) \,\big|_{\sqrt{q}}\,,
\end{equation}
expand the Pfaffian on the RHS to order $\sqrt{q}$, using
\begin{equation}
 \pf(M)=\sum_{\alpha\in\Pi}\text{sgn}(\alpha) m_{i_1j_1}\dots m_{i_nj_n}=\frac{1}{2}\sum_{i_k,j_k}(-1)^{1+i_k+j_k+\theta(i_k-j_k)}m_{i_kj_k}\pf(M_{i_kj_k}^{i_kj_k})\,.
\end{equation}
This leads to
\begin{align*}
 \pf(M_3) \,\big|_{q^{1/2}}&=\frac{1}{2}\sum_{i_k,j_k}(-1)^{1+i_k+j_k+\theta(i_k-j_k)}m_{i_kj_k} \,\big|_{\sqrt{q}}\,\,\pf({M_3}_{i_kj_k}^{i_kj_k}) \,\big|_{q^{0}}\\
 &= \sum_{i,j\in I} (-1)^{1+i+j+\theta(i-j)} \frac{k_i\cdot k_j}{\sigma_i} \,\pf({M_3}^{ij}_{ij})\\
  &\qquad+ \sum_{i,j\in I'}(-1)^{1+i+j+\theta(i-j)} \frac{\epsilon_i\cdot \epsilon_j}{\sigma_i} \,\pf({M_3}^{i'j'}_{i'j'})\\
  &\qquad +\sum_{i\in I',j\in I}(-1)^{i+j} \left(\frac{\epsilon_i\cdot k_j}{\sigma_i}-\frac{\epsilon_i\cdot k_j}{\sigma_j}\right) \,\pf({M_3}^{i'j}_{i'j})\,.
\end{align*}
The diagonal terms $C_{ii}=\epsilon_i\cdot P(\sigma_i)$ do not contribute since $P(\sigma)$ exclusively contains terms of the form $k_i \tilde{S}_1\sim k_i \frac{\theta'_1}{\theta_1}$, which only contribute at higher order in $q$. This concludes the proof of \cref{eq:PfNs-def}.

\section{Dimensional reduction}
\label{sec:dimens-reduct}

In this section we discuss considerations that are general and
somewhat out of the scope of the main article.  The point is to
discuss how can one dimensionally reduce the ambitwistor string to $d$
dimensions. At the core of these considerations is the work of
ACS~\cite{Adamo:2014wea} where the ambitwistor string was be
formulated in generic (on-shell) curved spaces; toroidal
compactifications are just a subcategory of the latter spaces.

In the usual string compactified on a circle of radius $R$, wrapping
modes or worldsheet instantons are solutions that obey the periodicity
conditions $X=X+2\pi m R$. Their classical values are given by
$X_{class}=2\pi R(n \xi_{1}+m \xi_{2})$ where $\xi_{1,2}$ are the
worldsheet coordinates, such that $z=\xi_{1}-i\xi_{2}$. This cannot be
made holomorphic, except if $n=m=0$, therefore none of these can
contribute in the Ambitwistor string, holomorphic by essence.
This may not exclude the possibility of having other type of more
exotic instantons, as mentioned in the final section of
ref.~\cite{Adamo:2014wea}, but we shall proceed here and assume that
none of these are generated.
In total, the Kaluza-Klein reduction of the amplitude
above~\eqref{e:graveven} is simply obtained by replacing the
$10$-dimensional loop momentum integral by a $d$-dimensional one and a
$(10-d)$-dimensional discrete sum
\begin{align}
\label{eq:KKreduction}
  \int d^{10} \ell \to \frac{1}{R^{10-d}} \sum_{\bd{n}\in \mathbb{Z}^{{10-d}}}
  \int d^{d}\ell_{(d)}
\end{align}
where $\bd n$ is an integer valued $(10-d)$-dim vector. A $10-d$ torus
with $10-d$ different radii $R_{1},\ldots,R_{10-d}$ is dealt with at
the cost of minor obvious modifications of the previous expression.

The loop momentum square is then given by
\begin{align}
\label{eq:lsqKK}
\ell^{2}= \ell_{(d)}^{2}+\frac{ {\bd n}^{2}}{R^{2}}
\end{align}

In this way, the transformation rule $\ell\to\tau\ell$ of the loop
momentum after a modular transformation is generalized to the compact
dimensions by demanding $R\to R/\tau$ and the integral is still
modular invariant, in the sense of~\cite{Adamo:2013tsa}.

Ultimately, we take the radius $R$ of the torus to zero in order to
decouple the KK states. In this limit, $\ell$ simply becomes
$\ell_{d}$  wherever it appears, therefore this process
is achieved by, loosely speaking, restricting the loop momentum integral
by hand.

In conclusion, standard compactification techniques of string theory
on tori and orbifolds thereof apply straightforwardly. 

\bibliography{../twistor-bib}  
\bibliographystyle{JHEP}

\end{document}